\documentclass[10pt,twocolumn]{IEEEtran}

\usepackage{pdfpages}
\usepackage{url}
\usepackage[usenames,dvipsnames]{pstricks}
\usepackage{epsfig}
\usepackage{pst-grad} 
\usepackage{pst-plot} 
\usepackage{amsmath,epsfig}
\usepackage{amssymb}
\usepackage{psfrag}
\usepackage{color}
\usepackage{pstricks}
\usepackage{cite}
\usepackage{algorithm}
\usepackage{epsfig}
\usepackage{amsthm}
\usepackage{amssymb}
\usepackage{psfrag}
\usepackage{pstricks}
\usepackage{float}
\usepackage{stfloats}
\usepackage{wrapfig}
\usepackage{graphicx}
\usepackage{amsmath}
\usepackage{amsfonts}
\usepackage{amssymb}
\usepackage{blindtext}
\usepackage{array}
\usepackage{float}

\usepackage{arydshln}
\usepackage{pstricks,pst-plot}
\usepackage{pst-bezier}
\usepackage{pst-math}
\usepackage{algorithm,algcompatible,amsmath}
\usepackage{float}
\usepackage[mathscr]{eucal}
\usepackage{enumerate}
\usepackage{graphicx}
\usepackage{float}
\usepackage[caption = false]{subfig}

\DeclareFontFamily{OT1}{pzc}{}
\DeclareFontShape{OT1}{pzc}{m}{it}{<-> s * [1.35] pzcmi7t}{}
\DeclareMathAlphabet{\mathpzc}{OT1}{pzc}{m}{it}

\algnewcommand\REQUIRED{\item[\textbf{Required:}]}%
\algnewcommand\INPUT{\item[\textbf{Input:}]}%
\algnewcommand\OUTPUT{\item[\textbf{Output:}]}%

\def\bK{{\bf K}}

\def\bn{{\bf n}}

\def\bZ{{\bf Z}}
\def\bD{{\bf D}}
\def\by{{\bf y}}

\def\bV{{\bf V}}
\def\bW{{\bf W}}

\def\bu{{\bf u}}
\def\bh{{\bf h}}

\def\bH{{\bf H}}

\def\bS{{\bf S}}
\def\bQ{{\bf Q}}
\def\bv{{\bf v}}

\def\bC{{\bf C}}
\def\bB{{\bf B}}

\def\bF{{\bf F}}

\def\bR{{\bf R}}
\def\br{{\bf r}}
\def\bI{{\bf I}}

\def\bE{{\bf E}}

\def\bv{{\bf v}}
\def\bA{{\bf A}}
\def\bB{{\bf B}}

\def\bU{{\bf U}}

\def\bs{{\bf s}}

\def\bx{{\bf x}}
\def\bT{{\bf T}}

\def\bPsi{{\bf \Psi}}
\def\bPhi{{\bf \Phi}}

\def\bLambda{{\bf \Lambda}}
\def\bSigma{{\bf \Sigma}}

\def\bGamma{{\bf \bgamma}}

\def\bUpsilon{{\bf \Upsilon}}

\def\bGamma{{\bf \Gamma}}

\def\b0{{\bf 0}}

\DeclareMathAlphabet\mathbfcal{OMS}{cmsy}{b}{n}
\DeclareMathAlphabet\mathbfbb{OMS}{cmsy}{b}{n}
\DeclareMathOperator*{\argmin}{arg\,min}

\DeclareMathAlphabet\mathbfcal{OMS}{cmsy}{b}{n}

\newcommand{\refappendix}[1]{\hyperref[#1]{Appendix~\ref*{#1}}}

\newcommand{\changeblack}[1]{{\color{black}#1}}

\newcolumntype{L}{>{\centering\arraybackslash}m{3cm}}

\newtheorem{lemma}{Lemma}

\allowdisplaybreaks

\title{\LARGE Hybrid Analog and Digital Beamforming Design for Channel Estimation in Correlated Massive MIMO Systems}
\author{ Javad Mirzaei$^{\star}$, Shahram ShahbazPanahi$^{\star\dagger}$, Foad Sohrabi$^{\star}$ and Raviraj Adve$^{\star}$ \\
 {\small  $^{\star }$Department of Electrical and Computer Engineering, University of Toronto, Canada

 $^{\star\dagger}$Department of Electrical and Computer Engineering, Ontario Tech University, Canada
  }

 \thanks{\small This work has been presented in part in \cite{Javad_ICASSP_2020}.}

}

\begin{document}

\maketitle
\begin{abstract}\label{Abs}
In this paper, we study the channel estimation problem in correlated massive multiple-input-multiple-output (MIMO) systems with a reduced number of radio-frequency (RF) chains. Importantly, other than the knowledge of channel correlation matrices, we make no assumption as to the structure of the channel. To address the limitation in the number of RF chains, we employ hybrid beamforming, comprising a low dimensional digital beamformer followed by an analog beamformer implemented using phase shifters. Since there is no dedicated RF chain per transmitter/receiver antenna, the conventional channel estimation techniques for fully-digital systems are impractical. By exploiting the fact that the channel entries are uncorrelated in its eigen-domain, we seek to estimate the channel entries in this domain. Due to the limited number of RF chains, channel estimation is typically performed in multiple time slots. Under a total energy budget, we aim to design the hybrid transmit beamformer (precoder) and the receive beamformer (combiner) in each training time slot, in order to estimate the channel using the minimum mean squared error criterion. To this end, we choose the precoder and combiner in each time slot such that they are aligned to transmitter and receiver eigen-directions, respectively. Further, we derive a water-filling-type expression for the optimal energy allocation at each time slot. This expression illustrates that, with a low training energy budget, only significant components of the channel need to be estimated. In contrast, with a large training energy budget, the energy is almost equally distributed among all eigen-directions. Simulation results show that the proposed channel estimation scheme can efficiently estimate correlated massive MIMO channels within a few training time slots.

\end{abstract}
\section{Introduction}
Massive multiple-input-multiple-output (MIMO) is proven to be the most promising technology for a wide range of applications, such as the Internet of things (IoT) \cite{6971234, 8752284, 7395392}, in the next generation of wireless networks. Conventionally, in wireless systems, signal processing is performed at baseband. This requires that, at the receiver for example, the analog signal be filtered, down-converted and properly sampled. These tasks are carried out by hardware modules known as radio-frequency (RF) chains. Therefore,  conventional baseband signal processing algorithms require a dedicated RF chain for each transmit/receive antenna. Given a large number of antennas in massive MIMO systems, such a large number of RF chains may not be available due to their high cost and power consumption. To overcome this challenge, a hybrid analog-digital beamforming structure has been proposed for massive MIMO systems\cite{1519678}. In such systems, the signal is first  processed at baseband, and is then up-converted to the RF domain using a reduced number of RF chains. The signal is then passed through a network of phase shifters connected to all transmit antennas before transmission. At the receiver, the received signals first go through a network of phase shifters, then, are down-converted using a reduced number of RF chains, and finally are fed to baseband processing blocks. Therefore, the conventional beamforming techniques, performed at the baseband \cite{Telatar, 4599181, Luo_TSP}, may not be applicable. Recently, several techniques  have been proposed to design hybrid analog-digital beamformers (HB) \cite{7160780, 7448873, 6717211, Foad_JSTSP_2016, 7913599, DBLP:journals/corr/abs-1712-03485}.

 Achieving the performance  of the aforementioned hybrid design approaches requires that accurate channel state information (CSI) be available.
 However, in such systems, channel estimation is challenging because the high dimensional channel is observable only through the limited number of RF chains \cite{7400949}. This implies that, the conventional fully-digital channel estimation methods cannot be directly applied to obtain the CSI for massive MIMO systems with a hybrid structure \cite{6777295, 1597555, 6940305, Javad_Tcom_2019, Javad_globcom_2018}. \textit{In this paper, we aim to address the channel estimation problem for such a system by exploiting  the correlation among the entries of channel matrix.} Other than the knowledge of channel correlation matrices, we make no further assumption on the structure of the channel. Before we elaborate on the contributions of this paper, we briefly review the related work.

\textbf{Related Work:} In the context of millimeter-wave (mmWave) communication systems, the CSI acquisition challenge has been addressed by exploiting channel sparsity in mmWave frequencies \cite{6489376, AlkhateebHBest, 7370753, DBLP:journals/corr/AlkhateebLH15, 8093607, Foad_TCOM2019}. By exploiting the limited-scattering nature of mmWave channels, the authors in \cite{AlkhateebHBest,7370753} show that the MIMO channel can be represented using a parametric model which is sparse in the underlying parameters~\cite{AlkhateebHBest}. Therefore, for the sake of channel estimation, instead of estimating the entire channel matrix, only the angles of departure/arrival (AoD/AoA) of dominant paths and the corresponding path gains are  estimated. In  this case, the channel estimation is formulated as a sparse problem where the measurement matrices are expressed in terms of the hybrid precoders/combiners in the training phase. In \cite{AlkhateebHBest, 6489376, 6148295, 7400949 }, the measurement matrices are designed based on adaptive compressive sensing (CS). In particular, the authors of \cite{AlkhateebHBest} propose an adaptive algorithm to design the measurement matrices. The channel estimate is obtained iteratively by scanning through a set of training beamforming vectors at both transmitter and receiver sides. The complexity of the work in \cite{AlkhateebHBest} is dictated mainly by the level of channel sparsity; however, in rich scattering environments, the algorithm requires more resources with higher complexity to achieve the channel estimation with desired angular resolution.
In \cite{7306533}, coarse channel estimation is performed using the beam training approach, while a subsequent CS algorithm refines the resolution. As an alternative to adaptive CS algorithms, traditional random CS approaches exploiting pseudo-random weights are used in a phased array MIMO system with a hybrid structure \cite{6181796, DBLP:journals/corr/AlkhateebLH15}. The authors of \cite{7370753} apply random CS with simpler analog beamformers, where the network of switches is replaced with the phase-shifter network\footnote{With switches, the entries of analog beamforming matrix are either $0$ or $1$.}. In \cite{Evans_ICASSP18}, the CS-based channel estimation algorithm extended to a mmWave communication system equipped with one-bit analog-to-digital converters (ADCs) and HB.

The authors of \cite{8093607,7938435, 7955996, DBLP:journals/corr/BuzziD17a} propose non-CS-based techniques to estimate the channel in massive MIMO systems with the hybrid precoder/combiner structure. Assuming the parametric channel model, the authors in \cite{7938435} exploit two-dimensional beamspace multiple signal classification (MUSIC) to estimate the path directions followed by a least-squares estimate of the path gains. In \cite{8093607}, estimating signal parameters via rotational invariance techniques (ESPRIT) is employed to estimate the CSI. Note that applying the ESPRIT technique to massive MIMO systems with the hybrid structure is challenging, since this technique requires the shift-invariance property of the array response in the observation samples. To circumvent this issue, the authors of \cite{8093607} first design a training signal  for channel estimation in a way that the low-dimensional effective channel has this shift-invariance property. Then, exploiting the angular sparsity of the mmWave channel, the AoA and AoD estimates are obtained. The authors of \cite{7955996} and \cite{DBLP:journals/corr/BuzziD17a} exploit projection approximation subspace tracking with deflation (PASTd) algorithm to  track the subspace and estimate the right (left) singular vector at the transmitter (receiver).

In \cite{8227727, Eldar_9026804}, the authors consider a multi-user multi-cell communication system where each user is equipped with one antenna. Assuming a Kronecker channel model, with a priori known channel’s covariance matrices, the authors estimate the channel matrix by jointly designing the pilot sequences (to mitigate the effect of pilot contamination) and analog combiners that yield the minimum channel estimation error. It is shown that, regardless of the choice of the pilot sequences, the analog combiners can be designed using existing algorithms, such as the one proposed in \cite{DBLP:journals/corr/abs-1712-03485}. Then, the pilot sequence is designed by solving an optimization problem that yields the minimum sum mean squared estimate (MSE) of the channel estimate. Table \ref{table_lit} summarizes the existing channel estimation techniques in massive MIMO systems.
\begin{table*}[t]
\caption {Summary of related works} \label{table_lit}
\centering
\resizebox{\textwidth}{!}{
        \begin{tabular}{|c |c |c |c |c|}
            \hline
            \multicolumn{1}{|c|}{Frequency} & \multicolumn{1}{c}{Channel Model} & \multicolumn{1}{|c|}{Assumption} & \multicolumn{1}{c|}{Estimation Technique} & \multicolumn{1}{c|}{Reference}\\
            \hline
               &  &  &  CS-based &  \cite{7400949, 6489376, AlkhateebHBest, 7370753, DBLP:journals/corr/AlkhateebLH15, 6148295, 6181796, Evans_ICASSP18} \\
             \cline{4-5}
              mmWave &  Parameteric &  Sparsity  & \multicolumn{1}{m{3.5cm}|}{Beam training followed by CS-based resolution refining} &  \cite{7306533} \\
              \cline{4-5}
              &  &  &  ESPRIT &  \cite{8093607} \\
              \cline{4-5}
              &  &  &  MUSIC &  \cite{7938435} \\
              \cline{4-5}
              &  &  &  PASTd &  \cite{7955996, DBLP:journals/corr/BuzziD17a} \\
              \hline
              Microwave and mmWave & Kronecker  & \multicolumn{1}{m{3.5cm}|}{Known Channel Correlations} &  MMSE &  \cite{ 8227727, Eldar_9026804} \\
            \hline
        \end{tabular}}
\end{table*}

\textbf{Contributions and Methodology:} In this paper, we propose a training-based channel estimation technique for \emph{correlated} massive MIMO systems with the hybrid precoder/combiner structure. Using the {Kronecker model} to represent the channel statistics, we assume that the long-term transmitter and receiver correlation matrices are available at both transmitter and receiver. Since a correlated MIMO channel has fewer degrees of freedom than an uncorrelated channel, we expect that fewer parameters need to be estimated to effectively reconstruct the channel.

Due to the reduced number of RF chains at both transmitter and receiver, the channel training is performed in multiple time slots. To minimize the time used for training, we express the channel in its eigen-domain, where the channel is projected onto the eigenvectors of correlation matrices at transmitter and receiver. This channel representation is referred to as a \emph{virtual channel} \cite{1261339}. The channel matrix and its virtual representation are unitarily equivalent. Importantly, the entries of the virtual channel are uncorrelated -- a fact that we use to efficiently estimate the channels. To this end, we use the minimum mean squared error (MMSE) criterion and design the needed hybrid beamformers. Due to the constraint on analog beamformers, the optimization problem is non-convex. To efficiently solve this problem, we first solve for optimal \emph{fully-digital} beamformers. Then, we split the so-obtained {fully-digital} beamformers into digital and analog beamformers using a least-squares approach.

\label{contrib}The main contribution of this paper lies in the design of hybrid beamformers in order to estimate the channel matrix of a correlated massive MIMO with reduced RF chains. The hybrid beamforming design is carried out in two phases. In the first phase, we assume that there is no hybrid structure for the beamformers and we propose an MMSE-based technique to design the fully digital beamformers. To realize the hybrid structure, the beamformers are then split into the digital and analog parts using the already-existing solutions, such as those of \cite{Foad_JSTSP_2016, 7397861}.

The designed fully-digital precoder and combiner in each training time slot are, respectively, aligned to transmitter and receiver eigen-directions,  thereby allowing us to estimate a portion of the virtual channel along its eigen-directions. To minimize the MSE, more of training energy budget is allocated to stronger eigen-directions. With a low training energy budget, instead of estimating the channel in all eigen-directions, only the components along the stronger eigen-directions are estimated within a fewer training time slots. This is because the rest of the eigen-directions are not as important from an MMSE point of view, and therefore, they are not estimated. On the other hand, given a large training energy budget, the energy is almost equally allocated to all eigen-directions.

 Our system setup is different from the studies in \cite{1033689, 4133041, 4608751, 5340650, 4803746 , 1597555}. Here due to a reduced number of RF chains, the signals received in each time slot, belong to a space with a much smaller dimension than the number of transmitter/receiver antennas. Therefore, channel estimation has to be performed in multiple time slots. We seek to optimally design the hybrid beamformers in each time slot in order to estimate the channel given a training energy budget. The work in this paper differs from that \cite{AlkhateebHBest, 6489376, 6148295, 7400949, 7306533, DBLP:journals/corr/AlkhateebLH15, 8093607} in the sense that here we do not rely on any assumption of channel sparsity. Instead, our approach in this paper considers a general correlated MIMO channel model where such a correlation is exploited to design the beamformers for channel training efficiently.

The work in our paper is substantially different from  those in \cite{8227727, Eldar_9026804, 6484165}
(which are the most closely related work in the literature). Here, we consider a point-to-point MIMO system where the transmitter and receiver are both equipped with a reduced number of RF chains. Unlike the study in \cite{8227727, Eldar_9026804}, we estimate the channel in its eigen-domain, where the channel entries in this domain are uncorrelated. Leveraging this fact, we show that our channel estimation technique can be performed with fewer training time slots. Furthermore, to improve channel
estimation accuracy, we design the hybrid beamformers for each training time slot, as opposed to \cite{8227727, Eldar_9026804}, where the combiner is fixed throughout the channel training process. \label{Aug_paper}

The  data and channel models as well as the design criterion in our work  are  somewhat similar to those in \cite{6484165}. The main difference in the data models stems from the fact the study in \cite{6484165} focuses on designing training sequences for systems with one RF chain per antenna, while our goal is to design  hybrid beamformers for systems with reduced number of RF chains. As such, the result, the proofs, and the derivations of   \cite{6484165} is not applicable to our system setup.

{\textbf{Organization}}: This paper is organized as follows: We introduce the system model, including the received signal model and the channel model in Section~\ref{sec:Sys_model}. Section \ref{sec:Ch_est} introduces the channel estimation problem formulation. To estimate the channel, in Section  \ref{sec:FD_design}, we propose an MMSE-based fully-digital beamforming design. In Section  \ref{sec:FD_design}, we explain how to design the hybrid precoder and combiner by splitting the fully-digital beamformers into analog and digital components. Section~\ref{sec:Simulations} presents simulation results to illustrate the efficacy of the proposed approach. Finally, Section~\ref{sec:conclusion} concludes the paper.

{\textbf{Notation}}: We use bold upper and lower-case letters to denote matrices and vectors, respectively. $\mathbb{E}\{\cdot\}$ denotes statistical expectation. The transpose, hermitian, conjugate, and pseudo-inverse operations are represented as $(\cdot)^T$, $(\cdot)^H$, $(\cdot)^*$, and $(\cdot)^\dag$, respectively; $|x|$ is the absolute value of complex $x$; $\lceil x \rceil$ returns the nearest integer greater than or equal to $x$; ${\rm mod} (x,m)$ returns the remainder after division of $x$ by $m$. ${\rm Tr}(\bA)$ is used to denote the summation of diagonal entries of matrix $\bA$. ${\rm diag}(\bx)$ represents a diagonal matrix whose diagonal entries are the elements of the vector $\bx$, while ${\rm diag\left(\bA\right)}$ is a vector that captures the diagonal entries of matrix $\bA$. $\bI_N$ denotes an $N \times N$ identity matrix. ${\bf 1}_{N}$ denotes an $N \times 1$ all one vector. ${\rm blkdg\left[\bA_1,\bA_2,\cdots, \bA_N\right]}$ is a block diagonal matrix whose $n$th diagonal entry is given by matrix $\bA_n$, for $n=1,2, \cdots, N$. The $(i, j)$th element of a matrix $\mathbf{A}$ is denoted as $\bA(i,j)$ or $\bA_{ij}$, $\otimes$ is used for the Kronecker product, ${\rm vec}\{\cdot\}$ represents the matrix vectorization operation, and $\mathbb{R}_+$ denotes the set of non-negative real numbers.

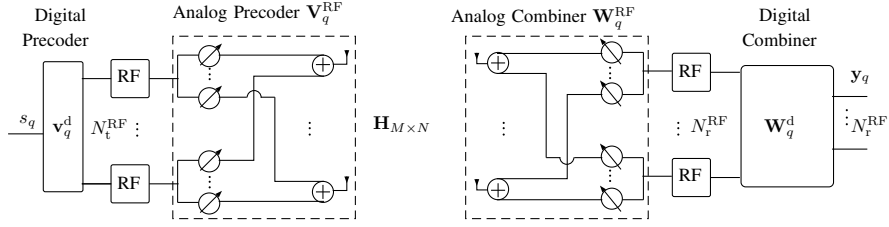
\begin{figure*}[t!]
    \centering

\psscalebox{0.7 .7} 
{
\begin{pspicture}(0,-2.0011537)(17.668888,2.0011537)
\psframe[linecolor=black, linewidth=0.02, dimen=outer, framearc=0.1](2.5185184,1.001555)(1.7453281,-1.4901876)
\psline[linecolor=black, linewidth=0.02](1.087066665,-0.37990654)(1.75510585,-0.37990654)

\psline[linecolor=black, linewidth=0.02](2.4870667,0.6702832)(3.0452063,0.6702832)
\psline[linecolor=black, linewidth=0.02](2.4870667,-1.3297168)(3.0452063,-1.3297168)
\psline[linecolor=black, linewidth=0.02](2.4870667,-1.3297168)(3.0452063,-1.3297168)
\psframe[linecolor=black, linewidth=0.02, dimen=outer, framearc=0.1](3.7760527,-0.9576238)(3.0318666,-1.7018099)
\rput[bl](3.1666667,-1.4232051){\rm RF}
\psline[linecolor=black, linewidth=0.02](3.7670667,0.6702832)(4.3252063,0.6702832)
\psline[linecolor=black, linewidth=0.02](4.306667,1.1167948)(4.306667,0.3167948)(4.306667,0.3167948)
\psline[linecolor=black, linewidth=0.02](4.306667,1.1167948)(4.7066665,1.1167948)
\psline[linecolor=black, linewidth=0.02](4.306667,0.3167948)(4.7066665,0.3167948)
\pscircle[linecolor=black, linewidth=0.02, dimen=outer](4.9066668,1.0967948){0.2}
\psline[linecolor=black, linewidth=0.02, arrowsize=0.05291667cm 2.0,arrowlength=1.4,arrowinset=0.0]{->}(4.7356324,0.9236914)(5.1494255,1.3374845)
\pscircle[linecolor=black, linewidth=0.02, dimen=outer](4.9066668,0.29679483){0.2}
\psline[linecolor=black, linewidth=0.02, arrowsize=0.05291667cm 2.0,arrowlength=1.4,arrowinset=0.0]{->}(4.7356324,0.123691365)(5.1494255,0.53748447)
\psline[linecolor=black, linewidth=0.02](5.0866666,1.1167948)(7.0866666,1.1167948)
\pscircle[linecolor=black, linewidth=0.02, dimen=outer](7.0666666,0.91679484){0.2}
\psline[linecolor=black, linewidth=0.02](7.0866666,0.7167948)(5.7533336,0.7167948)(5.7533336,-0.8832052)(5.0866666,-0.8832052)
\psline[linecolor=black, linewidth=0.02](7.0666,1.0701948)(7.0666,0.80352813)(7.0666,0.80352813)
\psline[linecolor=black, linewidth=0.02](6.9332666,0.93686146)(7.1999335,0.93686146)
\psline[linecolor=black, linewidth=0.02](3.7670667,-1.3297168)(4.3252063,-1.3297168)
\psline[linecolor=black, linewidth=0.02](4.306667,-0.8832052)(4.306667,-1.6832051)(4.306667,-1.6832051)
\psline[linecolor=black, linewidth=0.02](4.306667,-0.8832052)(4.7066665,-0.8832052)
\psline[linecolor=black, linewidth=0.02](4.306667,-1.6832051)(4.7066665,-1.6832051)
\pscircle[linecolor=black, linewidth=0.02, dimen=outer](4.9066668,-0.90320516){0.2}
\psline[linecolor=black, linewidth=0.02, arrowsize=0.05291667cm 2.0,arrowlength=1.4,arrowinset=0.0]{->}(4.7356324,-1.0763086)(5.1494255,-0.6625155)
\pscircle[linecolor=black, linewidth=0.02, dimen=outer](4.9066668,-1.7032052){0.2}
\psline[linecolor=black, linewidth=0.02, arrowsize=0.05291667cm 2.0,arrowlength=1.4,arrowinset=0.0]{->}(4.7356324,-1.8763087)(5.1494255,-1.4625155)
\rput{-1.1021739}(-0.0048984517,0.110853136){\psarc[linecolor=black, linewidth=0.02, dimen=outer](5.76,0.31006148){0.13333334}{0.0}{180.0}}
\psline[linecolor=black, linewidth=0.02](5.62,0.3167948)(5.0866666,0.3167948)
\psline[linecolor=black, linewidth=0.02](5.886667,0.3167948)(6.153333,0.3167948)
\psline[linecolor=black, linewidth=0.02, arrowsize=0.05291667cm 2.0,arrowlength=1.4,arrowinset=0.0]{-<}(7.253267,0.9435282)(7.519933,0.9435282)(7.519933,1.2101948)
\rput[bl](1.5932667,1.6911538){\rm Digital}
\rput[bl](1.3605057,1.290787785){\rm Precoder}
\rput[bl](1.932667,-0.47551286){$\bv_{q}^{\rm d}$}
\rput[bl](4.2153845,1.6911538){Analog Precoder $\bV_q^{\rm RF}$}
\psline[linecolor=black, linewidth=0.02](13.692307,-1.1854115)(13.134169,-1.1854115)
\psline[linecolor=black, linewidth=0.02](13.152708,-1.6319231)(13.152708,-0.8319231)(13.152708,-0.8319231)
\psline[linecolor=black, linewidth=0.02](13.152708,-1.6319231)(12.7527075,-1.6319231)
\psline[linecolor=black, linewidth=0.02](13.152708,-0.8319231)(12.7527075,-0.8319231)
\rput{-180.0}(25.105415,-3.2238462){\pscircle[linecolor=black, linewidth=0.02, dimen=outer](12.552708,-1.6119231){0.2}}
\rput{-180.0}(25.105415,-1.6238463){\pscircle[linecolor=black, linewidth=0.02, dimen=outer](12.552708,-0.81192315){0.2}}
\psline[linecolor=black, linewidth=0.02](12.372707,-1.6319231)(10.372707,-1.6319231)
\rput{-180.0}(20.785416,-2.8638463){\pscircle[linecolor=black, linewidth=0.02, dimen=outer](10.392708,-1.4319232){0.2}}
\psline[linecolor=black, linewidth=0.02](10.372707,-1.2319231)(11.706041,-1.2319231)(11.706041,0.36807686)(12.372707,0.36807686)
\psline[linecolor=black, linewidth=0.02](10.392775,-1.5853231)(10.392775,-1.3186564)(10.392775,-1.3186564)
\psline[linecolor=black, linewidth=0.02](10.526108,-1.4519898)(10.259441,-1.4519898)
\psline[linecolor=black, linewidth=0.02](13.692307,0.8145885)(13.134169,0.8145885)
\psline[linecolor=black, linewidth=0.02](13.152708,0.36807686)(13.152708,1.1680769)(13.152708,1.1680769)
\psline[linecolor=black, linewidth=0.02](13.152708,0.36807686)(12.7527075,0.36807686)
\psline[linecolor=black, linewidth=0.02](13.152708,1.1680769)(12.7527075,1.1680769)
\rput{-180.0}(25.105415,0.77615374){\pscircle[linecolor=black, linewidth=0.02, dimen=outer](12.552708,0.38807687){0.2}}
\rput{-180.0}(25.105415,2.3361537){\pscircle[linecolor=black, linewidth=0.02, dimen=outer](12.552708,1.1680769){0.2}}
\rput{-180.0}(20.785416,1.9361538){\pscircle[linecolor=black, linewidth=0.02, dimen=outer](10.392708,0.9680769){0.2}}
\psline[linecolor=black, linewidth=0.02](10.392708,0.83467686)(10.392708,1.1013435)(10.392708,1.1013435)
\psline[linecolor=black, linewidth=0.02](10.526041,0.9680102)(10.259375,0.9680102)
\rput{-181.10217}(23.38071,-1.8752688){\psarc[linecolor=black, linewidth=0.02, dimen=outer](11.699374,-0.82518977){0.13333334}{0.0}{180.0}}
\psline[linecolor=black, linewidth=0.02](11.839375,-0.8319231)(12.372707,-0.8319231)
\psline[linecolor=black, linewidth=0.02](11.572708,-0.8319231)(11.306041,-0.8319231)
\psline[linecolor=black, linewidth=0.02](11.3061075,-0.8319231)(11.3061075,0.7680769)(10.372774,0.7680769)(10.372774,0.7680769)
\psline[linecolor=black, linewidth=0.02](10.372707,1.1480769)(12.372707,1.1480769)(12.372707,1.1480769)
\rput{-180.0}(23.021313,-0.5151283){\psframe[linecolor=black, linewidth=0.02, linestyle=dashed, dash=0.17638889cm 0.10583334cm, dimen=outer](13.254246,1.4860256)(9.767067,-2.001154)}
\psline[linecolor=black, linewidth=0.02, arrowsize=0.05291667cm 2.0,arrowlength=1.4,arrowinset=0.0]{-<}(10.193795,-1.4457334)(9.988667,-1.4457334)(9.988667,-1.2406052)
\psline[linecolor=black, linewidth=0.02, arrowsize=0.05291667cm 2.0,arrowlength=1.4,arrowinset=0.0]{-<}(10.193795,0.9542666)(9.988667,0.9542666)(9.988667,1.1593949)
\psline[linecolor=black, linewidth=0.02, arrowsize=0.05291667cm 2.0,arrowlength=1.4,arrowinset=0.0]{->}(12.726666,0.94448715)(12.326667,1.4367948)
\psline[linecolor=black, linewidth=0.02, arrowsize=0.05291667cm 2.0,arrowlength=1.4,arrowinset=0.0]{->}(12.726666,0.14448713)(12.326667,0.6367948)
\psline[linecolor=black, linewidth=0.02, arrowsize=0.05291667cm 2.0,arrowlength=1.4,arrowinset=0.0]{->}(12.726666,-1.0555129)(12.326667,-0.5632052)
\psline[linecolor=black, linewidth=0.02, arrowsize=0.05291667cm 2.0,arrowlength=1.4,arrowinset=0.0]{->}(12.726666,-1.8555129)(12.326667,-1.3632052)
\psframe[linecolor=black, linewidth=0.02, dimen=outer, framearc=0.1](14.436052,1.1623948)(13.691867,0.41820878)
\rput[bl](13.826667,0.6967948){\rm RF}
\psframe[linecolor=black, linewidth=0.02, dimen=outer, framearc=0.1](14.436052,-0.8376052)(13.691867,-1.5817913)
\rput[bl](13.826667,-1.2832052){\rm RF}
\psline[linecolor=black, linewidth=0.02](14.427067,0.7903018)(14.985207,0.7903018)
\psline[linecolor=black, linewidth=0.02](14.427067,-1.2096982)(14.985207,-1.2096982)
\psframe[linecolor=black, linewidth=0.02, dimen=outer, framearc=0.1](16.77364,0.9289949)(14.9970665,-1.3976719)
\psline[linecolor=black, linewidth=0.02](16.735922,0.330261)(17.405235,0.330261)
\psline[linecolor=black, linewidth=0.02](16.735922,-0.669739)(17.405235,-0.669739)
\rput[bl](15.352381,1.6911538){\rm Digital}
\rput[bl](15.032464,1.290787785){\rm Combiner}
\rput[bl](15.448236,-0.47551286){$\bW_q^{\rm d}$}

\rput[bl](9.5,1.6056837){\rm{Analog Combiner} $\bW_q^{\rm RF}$}
\rput[bl](8.022222,-0.3720052){$\bH_{M\times N}$}

\psdots[linecolor=black, dotsize=0.04](3.5378666,-0.18980518)
\psdots[linecolor=black, dotsize=0.04](3.5378666,-0.2919385)
\psdots[linecolor=black, dotsize=0.04](3.5378666,-0.41433853)
\psdots[linecolor=black, dotsize=0.04](6.817867,-0.16980518)
\psdots[linecolor=black, dotsize=0.04](17.017866,0.06980518)
\psdots[linecolor=black, dotsize=0.04](17.017866,-0.0719385)
\psdots[linecolor=black, dotsize=0.04](17.017866,-0.19433852)
\psdots[linecolor=black, dotsize=0.04](13.817866,-0.16980518)
\psdots[linecolor=black, dotsize=0.04](13.817866,-0.2719385)
\psdots[linecolor=black, dotsize=0.04](13.817866,-0.39433852)
\psdots[linecolor=black, dotsize=0.04](10.497867,-0.16980518)
\psdots[linecolor=black, dotsize=0.04](10.497867,-0.2719385)
\psdots[linecolor=black, dotsize=0.04](10.497867,-0.39433852)

\psdots[linecolor=black, dotsize=0.04](12.517867,0.67980518)
\psdots[linecolor=black, dotsize=0.04](12.517867,0.7719385)
\psdots[linecolor=black, dotsize=0.04](12.517867,0.87433852)

\psdots[linecolor=black, dotsize=0.04](12.517867,-1.07980518)
\psdots[linecolor=black, dotsize=0.04](12.517867,-1.1719385)
\psdots[linecolor=black, dotsize=0.04](12.517867,-1.27433852)

\psdots[linecolor=black, dotsize=0.04](4.917867,0.65980518)
\psdots[linecolor=black, dotsize=0.04](4.917867,0.7519385)
\psdots[linecolor=black, dotsize=0.04](4.917867,0.85433852)

\psdots[linecolor=black, dotsize=0.04](4.917867,-1.19980518)
\psdots[linecolor=black, dotsize=0.04](4.917867,-1.2919385)
\psdots[linecolor=black, dotsize=0.04](4.917867,-1.39433852)

\rput[bl](2.6488667,-0.4720941){$N_{\rm t}^{\rm RF}$}
\rput[bl](14.048866,-0.4220941){$N_{\rm r}^{\rm RF}$}

\rput[bl](1.3,-0.2720052){$s_q$}
\rput[bl](17.088889,-0.42200518){$N_{\rm r}^{\rm RF}$}
\rput[bl](17.088889,0.5720052){$\by_q $}
\pscircle[linecolor=black, linewidth=0.02, dimen=outer](7.0666666,-1.4832052){0.2}
\psline[linecolor=black, linewidth=0.02](7.0666666,-1.3498052)(7.0666666,-1.6164719)(7.0666666,-1.6164719)
\psline[linecolor=black, linewidth=0.02](6.9333334,-1.4831386)(7.2,-1.4831386)
\psline[linecolor=black, linewidth=0.02](6.1532664,0.3167948)(6.1532664,-1.2832052)(7.0866,-1.2832052)(7.0866,-1.2832052)
\psline[linecolor=black, linewidth=0.02](7.0866666,-1.6632051)(5.0866666,-1.6632051)(5.0866666,-1.6632051)
\psline[linecolor=black, linewidth=0.02, arrowsize=0.05291667cm 2.0,arrowlength=1.4,arrowinset=0.0]{-<}(7.253267,-1.4564718)(7.519933,-1.4564718)(7.519933,-1.1898052)
\psframe[linecolor=black, linewidth=0.02, linestyle=dashed, dash=0.17638889cm 0.10583334cm, dimen=outer](7.6923075,1.4860256)(4.205128,-2.001154)
\psdots[linecolor=black, dotsize=0.04](6.817867,-0.2719385)
\psdots[linecolor=black, dotsize=0.04](6.817867,-0.39433852)
\psframe[linecolor=black, linewidth=0.02, dimen=outer, framearc=0.1](3.7760527,1.0423762)(3.0318666,0.29819018)
\rput[bl](3.1666667,0.59679484){\rm RF}
\end{pspicture}
}
\caption{System model for channel training during the $q$th time slot, where $q=1,2, \ldots, Q$. }
 \label{system_model}
\end{figure*}

\section{System Model}\label{sec:Sys_model}
\textbf{Channel Model:} We consider a narrow-band frequency-flat MIMO channel model, where the signals at the transmitter and receiver antenna arrays are correlated. Such a correlated channel is expressed by the following canonical statistical model, known as the \emph{Kronecker model} \cite{985982, 1261339}:
\begin{align} \label{kron_ch}
\bH=\bS_{\rm r}^{1/2} \bH_w \bS_{\rm t}^{1/2},
\end{align}
where the entries of $\bH_w$ are random variables drawn identically and independently from a complex Gaussian distribution with zero mean and unit variance. The matrices $\bS_{\rm t}$ and $\bS_{\rm r}$ are, respectively, the transmitter and receiver correlation matrices. This model has been considered in \cite{985982, 887129, 892194, 1261339, 7727938, 8227727} and has also been verified through the measurements \cite{965098, 1018011}.
To further characterize the channel model in \eqref{kron_ch}, let us denote the eigenvalue decomposition of $\bS_{\rm t}$ and $\bS_{\rm r}$, respectively, as $\bU_{\rm t}\bLambda_{\rm t}\bU_{\rm t}^H$ and $\bU_{\rm r}\bLambda_{\rm r}\bU_{\rm r}^H$. The matrices $\bU_{\rm t}$ and $\bU_{\rm r}$ are, respectively, the  eigenvector matrices of $\bS_{\rm t}$ and $\bS_{\rm r}$, while $\bLambda_{\rm t}$ and $\bLambda_{\rm r}$ are diagonal matrices capturing the eigenvalues of  $\bS_{\rm t}$ and $\bS_{\rm r}$, respectively. Using \eqref{kron_ch}, we can write
 \begin{align} \label{kron_ch_decompose}
\bH&=\bU_{\rm r}\bLambda_{\rm r}^{1/2}\bU_{\rm r}^H \bH_w \bU_{\rm t}
\bLambda_{\rm t}^{1/2}\bU_{\rm t}^H= \bU_{\rm r}\bH_v\bU_{\rm t}^H,
\end{align}
 where $\bH_v \triangleq \bLambda_{\rm r}^{1/2}\bar{\bH}_w \bLambda_{\rm t}^{1/2}$ and $\bar{\bH}_w \triangleq\bU_{\rm r}^H \bH_w \bU_{\rm t}$. Note that, using the fact that the entries of $\bH_w$ are independently and identically distributed (i.i.d.)\, and that  $\bU_{\rm t}$ and $\bU_{\rm r}$ are unitary matrices, we can show that the entries of $\bar{\bH}_w $ remain i.i.d.\ \cite{1261339}. Also, with pre- and post- multiplication of $\bLambda_{\rm r}^{1/2}$ and $\bLambda_{\rm t}^{1/2}$, the entries of $\bH_v $ are uncorrelated with diagonal covariance matrix $\bR_v$ defined as \cite{1261339}:
\begin{align} \label{Rv}
\bR_v\triangleq \mathbb{E}\{\bh_v\bh_v^H\}= \bLambda_{\rm t}
 \otimes \bLambda_{\rm r}= \bLambda,
\end{align}
where $\bh_v = {\rm vec}\{\bH_v\}$ and $\bLambda$ is a diagonal matrix.
\label{rankdefR}Note that we herein assume that $\bR_v$ is full rank.
The Kronecker structure of $\bLambda$ in \eqref{Rv} plays an important role in our proposed channel estimation technique.

\textbf{Signal Model:}
We consider a narrow-band point-to-point MIMO system with $N$ transmit antennas and $M$ receive antennas. We assume that the transmitter and the receiver are, respectively, equipped with $N_{\rm t}^{\rm RF}$ and $N_{\rm r}^{\rm RF}$ RF chains. We also assume that the channel training is performed in a maximum\footnote{\label{Qconst}Indeed, $Q$ is the maximum number of time slots available for channel training. As we will show in this paper, not all $Q$ time slots might be needed.}
 of  $Q$ time slots (we  elaborate on the minimum number of training time slots  in Section \ref{sec:Ch_est}). As shown in Fig.~\ref{system_model}, at the $q$th time slot, $q=1,2,\ldots, Q$, the training symbol, $s_q$, is first linearly  precoded by multiplying with a digital precoder $\bv_q^{\rm d} \in \mathbb{C}^{N_{\rm t}^{\rm RF} \times 1}$. The baseband signal, after the digital precoder, is written as $\bx_q =\bv_q^{\rm d}s_q$.

The signal is then up-converted to the carrier frequency by  $N_{\rm t}^{\rm RF}$ RF chains. To construct the final transmit signal, the $N\times N_{\rm t}^{\rm RF}$ RF precoder, $\bV_q^{\rm RF}$, is applied to the signal $\bx_q$ at the carrier frequency. 
In the literature, different architectures for $\bV_q^{\rm RF}$ have been proposed. In this paper, we use the fully-connected phase shifter structure, in which $\bV_q^{\rm RF}(i,j)$ is a complex number with $|\bV_q^{\rm RF}(i,j)|=1$. Under such hybrid beamforming, the transmitted signal can be represented as:
\begin{align} \label{transmit signal}
\bar{\bx}_q&= \bV_q^{\rm RF}{\bx_q}= \bV_q^{\rm RF}\bv_q^{\rm d}s_q.
\end{align}
This signal propagates through the $M \times N$ channel matrix $\bH$. We wish to estimate $\bH$. At the receiver, the signal is first processed by the analog combiner, $\bW_q^{\rm RF}\in \mathbb{C}^{M\times N_{\rm r}^{\rm RF}}$, implemented using analog adders and analog phase shifters, implying that $|\bW_q^{\rm RF}(i,j)|=1$,  $ \forall \, i, j$.  By passing through $N_{\rm r}^{\rm RF}$ RF chains, the signal is then  down-converted to baseband. Hence, the received symbol vector $\br_q$ after the down-conversion can be written as:
\begin{align}
\br_q \! \triangleq \! {\bW_q^{\rm RF}}^{H}\bH\bV_q^{\rm RF} \bx_q + \bW_q^{\rm RF}\bn_q
 \!= \! {\bW_q^{\rm RF}}^{H}\bH\bV_q^{\rm RF}\bv_q^{\rm d}s_q+ \bar{\bn}_q. \nonumber
\end{align}
Here, $\bar{\bn}_q\triangleq \bW_q^{\rm RF}\bn_q$, where $\bn_q\sim \mathcal{CN}(\b0, \sigma_{\rm n}^2\bI_M)$ is the receiver's noise vector during the $q$th time slot. Further baseband processing is carried out by applying an $N_{\rm r}^{\rm RF}\times N_{\rm r}^{\rm RF} $ digital combiner $\bW^{\rm d}_q$. Hence, the final processed received signal can be written as:
\begin{align} \label{received signal_final_train}
\by_q& \triangleq {\bW_q^{\rm d}}^H \br_q = \bW_q^H \bH \bv_q s_q+ \bW_q{\bn}_q,
\end{align}
where $\bW_q\triangleq \bW_q^{\rm RF}\bW_q^{\rm d}$ and $\bv_q\triangleq \bV_q^{\rm RF}\bv_q^{\rm d}$ are the overall hybrid precoder and combiner, respectively. Note that, in order to estimate $\bH$, we must design the hybrid beamformers, i.e., $\bW_q^{\rm RF}$, $\bW_q^{\rm d}$, $\bV_q^{\rm RF}$, and $\bv_q^{\rm d}$. Throughout the paper, for the ease of exposition, we set $N^{\rm RF}=N_{\rm t}^{\rm RF}=N_{\rm r}^{\rm RF}$ and { $N^{\rm RF}\geq 2$.}

To efficiently estimate the channel matrix $\bH$, we note that the beamformers across different training time slots may not necessarily be the same. Let $\left\{ \bv_q,  \bW_q \right\}_{q=1}^Q$ denote the set of beamformers used at the $q$th time slot.  Due to the hybrid beamforming structure, $\bv_q$ and $\bW_q$ are constructed as  cascades of analog and digital components, i.e.,
\begin{align}
\bv_q&= \bV_q^{\rm RF}\bv_q^{\rm d}, \qquad  \bW_q= \bW_q^{\rm RF}\bW_q^{\rm d}.\label{vq_Wq_split}
\end{align}
 Without loss of generality, we consider $s_q=1$ as the training symbol. From \eqref{received signal_final_train}, the received signal at the $q$th time slot is given by:
\begin{align} \label{received_signal_multi__qth}
\by_q&=\bW_q^H\bH \bv_q s_q + \bW_q^H \bn_q =\bW_q^H\bU_{\rm r}\bH_v\bU_{\rm t}^H \bv_q + \bW_q^H \bn_q \nonumber\\
&= (\bv_q^T \otimes\bW_q^H){\rm vec}\{\bU_{\rm r}\bH_v\bU_{\rm t}^H\}  + \bW_q^H \bn_q \nonumber\\
&= (\bv_q^T \otimes\bW_q^H)\bPsi\bh_v  + \bW_q^H \bn_q,
\end{align}
where $\bPsi\triangleq \left(\left(\bU_{\rm t}^H\right)^T \otimes\bU_{\rm r}\right)= \left[\bPsi_1\;\; \bPsi_2\;\; \cdots\;\; \bPsi_{Q}\right]$.
In order to characterize the $MN \times N^{\rm RF}$ matrix $\bPsi_q$, using $\nu \triangleq M/N^{\rm RF}$, and
\begin{align}
n_q=   \lceil q/\nu \rceil \label{nq_multi}, \,\,\,
m_q= \left\{
        \begin{array}{ll}
         \nu, & \hbox{{ if }} \, {\rm mod}\, (q , \nu)=0;\\
          {\rm mod}\, (q , \nu), & \hbox{ otherwise,}
        \end{array}
      \right.
\end{align}
we define
$\tilde{\bu}_{\rm t}^{(q)}\triangleq {\bu}_{\rm t}^{(n_q)}$ and
$\tilde{\bU}_{\rm r}^{(q)}\triangleq {\bU}_{\rm r}^{(m_q)}$,
where ${\bu}_{\rm t}^{(i)}$ is the $i$th column of $\bU_{\rm t}^* $, and ${\bU}_{r}^{(i)}$ is an $M \times N^{\rm RF}$ sub-matrix of ${\bU}_{\rm r}$ with the $(m,n)$th element given by ${\bU}_{r}^{(i)}(m,n) \triangleq {\bU}_{r}(m, (i-1)N^{\rm RF} +n) $. Finally, $\bPsi_q$ is written as:
$\bPsi_q\triangleq \tilde{\bu}_{\rm t}^{(q)} \otimes \tilde{\bU}_{\rm r}^{(q)}$.

After $Q$ training time slots, we stack all $L\triangleq QN^{\rm RF} $ measurements in the vector $\by $ as

\begin{align} \label{y_stack_multi}
\by
\triangleq\left[
\begin{array}{c}
\by_1\\
\by_2 \\
\vdots  \\
\by_{Q}
\end{array}
\right] &= \underbrace{\left[
\begin{array}{c}
\bv_1^T \otimes\bW_1^H\\
\bv_2^T \otimes\bW_2^H \\
\vdots  \\
\bv_{Q}^T \otimes\bW_{Q}^H
\end{array}
\right]}_{\triangleq \bPhi(\mathcal{V}, \mathcal{W})} \bPsi\bh_v  + \left[
\begin{array}{c}
\bW_1^H \bn_1\\
\bW_2^H \bn_2 \\
\vdots  \\
\bW_{Q}^H \bn_Q
\end{array}
\right]\nonumber\\
&=\bF(\mathcal{V}, \mathcal{W}) \bh_v +\bar{\bn}.
\end{align}
Here, we define $\bF(\mathcal{V}, \mathcal{W})\triangleq \bPhi(\mathcal{V}, \mathcal{W})\bPsi$, $\mathcal{V}$ and $\mathcal{W}$ are tuples defined, respectively, as
$\mathcal{V}\triangleq \left(\bv_1, \bv_2, \ldots, \bv_Q\right)$ and $\mathcal{W}\triangleq \left(\bW_1, \bW_2, \ldots, \bW_{Q}\right)$,
and we use the following definitions: $\bar{\bn} \triangleq \bar{\bW}\bn$, $\bar{\bW} \triangleq {\rm blkdg}\left[\bW_1^H, \ldots, \bW_{Q}^H\right]$, and $\bn\triangleq [\bn_1^T, \ldots, \bn_{Q}^T]^T$. Next, we describe the channel estimation problem formulation.

\section{Channel Estimation Problem Formulation} \label{sec:Ch_est}
In this section, we formulate the channel estimation problem for the system described in Section \ref{sec:Sys_model}. Assuming the availability of long-term statistics of \changeblack{transmitter and receiver }correlation matrices at both transmitter and receiver (i.e., $\bS_{\rm t}$ and $\bS_{\rm r}$), we aim to optimally design the hybrid precoders and combiners, $\bV_q^{\rm RF}$, $\bv_q^{\rm d}$, $\bW_q^{\rm d}$, and $\bW_q^{\rm RF}$ for $q=1,2, \cdots, Q$, in order to estimate the channel matrix using the MMSE criterion. In general, to estimate $\bH$, we must estimate all $MN$ unknown entries of the channel matrix. Since, we only have $N^{\rm RF}$ RF chains at the receiver side,  only $N^{\rm RF}$ observations are available in each time slot. \label{Q_dep}Therefore, to estimate the entire channel matrix, in general, we require at least $MN/N^{\rm RF}$ training time slots. Note that, from \eqref{kron_ch_decompose}, we obtain
\begin{align} \label{kron_ch_decompose1}
\bH_v= \bU_{\rm r}^H\bH\bU_{\rm t},
\end{align}
which implies that $\bH$ and $\bH_v$ are unitarily equivalent. \label{itcanbeshow}It can be shown that for highly correlated channel matrix $\bH$, the virtual channel matrix (also referred to as the eigen-domain channel\cite{1261339}) $\bH_v$ tends to be a \emph{sparse-like}\footnote{\label{semisparse} Based on \eqref{kron_ch_decompose1}, $\bH_v$ is the channel representation in the eigen-domain. Due to the non-zero entries of correlation matrices $\bS_{\rm t}$ and $\bS_{\rm r}$, the channel components along different eigenvectors have different levels of significance measured by the eigenvalues of these correlation matrices (i.e., the significance of the entries of $\bh_v={\rm vec}(\bH_v)$ is given by the diagonal entries of $\bLambda$). For uncorrelated channels, where we have $\bS_t = \bS_r= \bI$, and therefore, $\bLambda =\bI$,  the entries of $\bh_v$ have the same level of significance. However as the channel entries become more correlated, some of the diagonal components of $\bLambda$ become more significant compared to others, meaning that some of the entries of  $\bh_v$  are  more significant than the others. That is, we state that $\bH_v$ tends to be a sparse-like matrix. Further details on this issue is given in \cite{1261339}.} matrix, meaning that most of its elements have very small amplitudes. Our goal here is to exploit such semi-sparsity to effectively estimate $\bH_v$ within a fewer time slots than $MN/N^{\rm RF}$.

We now aim to formulate the channel estimation problem. The goal is to design the beamformers in order to estimate $\bh_v$ based on the MMSE criterion. The linear MMSE estimate of $\bh_v$ is given by:
\begin{align} \label{MMSE_estimate}
\hat{\bh}_v= \bA_{\rm o}(\mathcal{V}, \mathcal{W})\by,
\end{align}
where $\bA_{\rm o}(\mathcal{V}, \mathcal{W})$ is obtained as:
\begin{align} \label{MMSE_estimate1}
&\bA_{\rm o}(\mathcal{V}, \mathcal{W})= \argmin_{\bA}\mathbb{E}\left\{\parallel{\bh}_v - \bA\by\parallel^2\right\}\\
&= \bR_v\bF^H(\mathcal{V}, \mathcal{W})\left(\bF(\mathcal{V}, \mathcal{W}) \bR_v\bF^H(\mathcal{V}, \mathcal{W}) + \bR_n(\mathcal{W})\right)^{-1}, \nonumber
\end{align}
where $\bR_n (\mathcal{W}) \triangleq \mathbb{E}\left\{\bar{\bn}\bar{\bn}^H\right\}= \bar{\bW} \mathbb{E}\left\{{\bn}{\bn}^H\right\} \bar{\bW}^H = \sigma_{\rm n}^2 \bar{\bW}\bar{\bW}^H ={ \rm blkdg} \left[\sigma_{\rm n}^2\bW_q^H\bW_q \right]_{q=1}^Q$. The covariance matrix of the estimation error is given by:
\begin{align} \label{MMSE_estimate_cov_mtx}
&\bC_{\rm MMSE}(\mathcal{V}, \mathcal{W})= \mathbb{E}\left\{\left({\bh}_v - \hat{\bh}_v\right)\left({\bh}_v - \hat{\bh}_v\right)^H\right\}\nonumber\\
&= \resizebox{.95\hsize}{!}{$\bR_v - \bR_v\bF^H(\mathcal{V}, \mathcal{W})\left(\bF(\mathcal{V}, \mathcal{W})\bR_v\bF^H(\mathcal{V}, \mathcal{W})+\bR_n(\mathcal{W})\right)^{-1}$}\nonumber\\
&\qquad \times \bF(\mathcal{V}, \mathcal{W})\bR_v \nonumber\\
&=\left(\bR_v^{-1} + \bF^H(\mathcal{V}, \mathcal{W})\bR_n^{-1}(\mathcal{W})\bF(\mathcal{V}, \mathcal{W}) \right)^{-1}.
\end{align}
Accordingly, the MMSE estimation error is computed as:
\begin{align}  \label{MMSE_est_error}
&J_{\rm MMSE}(\mathcal{V}, \mathcal{W})= {\rm Tr}\left(\bC_{\rm MMSE}(\mathcal{V}, \mathcal{W}) \right)\nonumber\\
&= {\rm Tr}\left(\left(\bR_v^{-1} + \bF^H(\mathcal{V}, \mathcal{W})\bR_n^{-1}(\mathcal{W})\bF(\mathcal{V}, \mathcal{W}) \right)^{-1} \right).
\end{align}
Using \eqref{Rv}, and defining
\begin{align} \label{gamma2_def}
\bGamma^2(\mathcal{V},  \mathcal{W}) \triangleq \bF^H(\mathcal{V},  \mathcal{W})\bR_n^{-1}(\mathcal{W})\bF(\mathcal{V},  \mathcal{W}),
\end{align}
we obtain $J_{\rm MMSE}(\mathcal{V}, \mathcal{W})= {\rm Tr}\left(\left(\bLambda^{-1}  + \bGamma^2(\mathcal{V},  \mathcal{W})  \right)^{-1}\right)$. Under a total training energy budget $E_T$, the MMSE-based beamforming design amounts to solving the following optimization problem:
\begin{subequations}\label{min_MMSE_HB}
\begin{align}
\min_{\mathcal{V}, \mathcal{W}}\,\,\,\, &J_{\rm MMSE}(\mathcal{V}, \mathcal{W})\,\\ {\rm s.t.}\,\,\,\,  &\sum_{q=1}^{Q}\|\bv_q\|^2\leq E_T,\label{min_MMSE_cons1}\\
&\bv_q= \bV_q^{\rm RF}\bv_q^{\rm d},\label{vq_split_const}\\
&\bW_q= \bW_q^{\rm RF}\bW_q^{\rm d},\label{Wq_split_const}\\
& \bV_q^{\rm RF}\in \mathcal{A}_v \,\,\,\,\,{\rm and }\,\,\,\, \bW_q^{\rm RF}\in \mathcal{A}_w,\label{norm_cons}
\end{align}
\end{subequations}
where $P_T= E_T/Q$ is  the average transmit power during training, and
sets $\mathcal{A}_w$ and $\mathcal{A}_v$ are defined as $\mathcal{A}_w\triangleq\Big\{\bW \in \mathcal{C}^{M\times N^{\rm RF}} \Bigm| |\bW(i,j)|^2=1\Big\}$ and $\mathcal{A}_v\triangleq\Big\{\bV \in \mathcal{C}^{N\times N^{\rm RF}}\Bigm| |\bV(i,j)|^2=1\Big\}$.
Note that the optimization problem in \eqref{min_MMSE_HB} is non-convex due to the constraints in \eqref{norm_cons}, and may not be amenable to a computationally  efficient   solution. To tackle this problem, we first relax the  constraints in \eqref{norm_cons} (thereby turning  \eqref{vq_split_const} and \eqref{Wq_split_const} trivial), and solve for the optimal $\mathcal{V}$, and $\mathcal{W}$, denoted as $\mathcal{V}^{\rm o}$, and $\mathcal{W}^{\rm o}$, respectively. This solution is referred to as unconstrained \emph{fully-digital} (FD) solution\footnote{Note that the unconstrained {fully-digital} case is different from the widely known FD beamformers, where a dedicated RF chain is assigned to each antenna. Here, we refer to $\mathcal{V},\mathcal{\tilde{V}}, \mathcal{W}$ the unconstrained FD beamformers where the structures in \eqref{vq_split_const} and \eqref{Wq_split_const} are ignored. }. \changeblack{Then, using the so-obtained  $\mathcal{V}^{\rm o}$ and $ \mathcal{W}^{\rm o}$} along with \eqref{vq_split_const}, \eqref{Wq_split_const} and the  constraints in \eqref{norm_cons}, we seek to find the hybrid beamfomers $\bV_q^{\rm RF},\bv_q^{\rm d},\bW_q^{\rm RF}$, and $\bW_q^{\rm d}$ such that \eqref{vq_split_const} and \eqref{Wq_split_const} are satisfied. We study the unconstrained {fully-digital} solution for Case 1) $Q=MN/N^{\rm RF}$, i.e., the number of time slots required to estimate the entire channel matrix; and Case 2) $Q<MN/N^{\rm RF}$ \footnote{Note that the case $Q>MN/N^{\rm RF}$ is not considered since, under the energy constraint, increasing the number of training time slots beyond $MN/N^{\rm RF}$ does not improve the channel estimation from an MMSE point of view \cite{5464944, 1597555}. }.

\section{Fully-digital Beamforming Design} \label{sec:FD_design}
In this section, we consider a fully-digital beamforming design scenario where there is no constraint on $\bv_q$ and $\bW_q$, except the total energy constraint. We seek to find optimal $\bv_q$ and $\bW_q$ based on the MMSE criterion. Before we specify the fully-digital problem formulation, let us simplify $\bGamma^2(\mathcal{V}
,\mathcal{W})$. Let $\bW_q= \bK_q\bD_q\bZ_q^H$ denote the singular value decomposition of $\bW_q$, where $\bK_q \in \mathbb{C}^{M\times N^{\rm RF}}$ and $\bZ_q\in \mathbb{C}^{N^{\rm RF}\times N^{\rm RF}}$ are, respectively, the matrices of the left and the right singular vectors, and $\bD_q\in \mathbb{R}_+^{N^{\rm RF}\times N^{\rm RF}}$ is the diagonal matrix of the singular values. We show in Appendix \ref{Simp_gamma2} that $\bGamma^2(\mathcal{V},  \mathcal{W})$ is independent of $\bD_q$ and $\bZ_q$, and it can be expressed as:
\begin{align} \label{gamma_long_mtx_body}
\bGamma^2(\mathcal{Z}, \mathcal{\tilde V}, \mathcal{{K}})= \frac{1}{\sigma_{\rm n}^2} \bPsi^H \bUpsilon \left(\mathcal{\tilde{V}}, \mathcal{{K}}\right) \mathbfcal{D}\left(\mathcal{Z}\right) \bUpsilon^H\left(\mathcal{\tilde{V}}, \mathcal{{K}}\right) \bPsi.
\end{align}
Here, with small abuse of notation, we use $\bGamma^2(\mathcal{Z},\mathcal{\tilde V},  \mathcal{{K}})$ instead of  $\bGamma^2(\mathcal{V},  \mathcal{{W}})$ where $\mathcal{Z}=   (z_1, z_2,\ldots, z_Q) \triangleq (\|\bv_1\|, \|\bv_2\|,\ldots, \|\bv_Q\|)\succcurlyeq {\bf 0}$, $\mathcal{{K}}\triangleq \left(\bK_1, \cdots, \bK_Q\right)$, $ \mathcal{\tilde V}\triangleq \left(\tilde \bv_1, \tilde \bv_2, \ldots, \tilde \bv_Q\right)$, and $\tilde \bv_q \triangleq \bv_q /\| \bv_q\|$. Also, $\mathbfcal{D}\left(\mathcal{Z}\right)$ and $\bUpsilon \left(\mathcal{\tilde{V}}, \mathcal{{K}}\right)$ are, respectively, an $L\times L$ diagonal matrix and an $MN\times L$ matrix defined as:
\begin{align} \label{gamma_long_mtx_D_def}
\mathbfcal{D} \left(\mathcal{Z}\right) &\triangleq  {\rm blkdg}\left(
  z_1^2  \bI,z_2^2  \bI , \; \ldots \;, z_Q^2  \bI
  \right)
,\\
\bUpsilon \left(\mathcal{\tilde{V}}, \mathcal{{K}}\right)&\triangleq
\left(
  \begin{array}{cccc}
  \tilde{\bv}_1^T\otimes \bK_1^H  \\
     \tilde{\bv}_2^T\otimes \bK_2^H  \\
      \vdots  \\
      \tilde{\bv}_{Q}^T\otimes \bK_{Q}^H
  \end{array}
\right)^H. \label{uupppsiloon}
\end{align}
In the rest of the paper, we replace $J_{\rm MMSE}(\mathcal{V}, \mathcal{W})$ with $J_{\rm MMSE}(\mathcal{Z}, \mathcal{\tilde V}, \mathcal{W})$, thereby emphasizing that the optimization variables are now $\mathcal{Z}$, $\mathcal{\tilde V}$, and $\mathcal{W}$. Note that $\mathcal{Z}$ denotes the power in the transmit beamformers, and $\mathcal{\tilde V}$ is their corresponding directions. To specify the fully-digital beamforming design problem, we relax the minimization \eqref{min_MMSE_HB} by ignoring {\eqref{vq_split_const}, \eqref{Wq_split_const}, and \eqref{norm_cons}, and solve the following minimization:
\begin{align}  \label{J_MMSE_FD_opt}
\min_{\mathcal{Z},\mathcal{\tilde{V}}, \mathcal{K}}\, &{\rm Tr}\left(\left({\bLambda}^{-1}  + {\bGamma}^2\left(\mathcal{Z}, \mathcal{\tilde V},  \mathcal{{K}} \right) \right)^{-1}\right), \;\;
{\rm s.t.}\, &\sum_{q=1}^{Q}z_q^2\leq E_T
\end{align}
where we have dropped the non-negative constraints on  $\{z_q\}_{q=1}^Q $, simply because if, at the optimum, $z_q < 0$, for some $q$,   flipping the sign of $z_q$  will not violate the constraint in  \eqref{J_MMSE_FD_opt}, neither does it affect the cost function.
Using the fact that
\begin{eqnarray} \label{diag_fact}
{\rm Tr} (\bA^{-1}) \ge \sum_i  (\bA^{}(i,i))^{-1},
\end{eqnarray}
with equality if and only if $\bA$ is a diagonal matrix \cite{1597555}, we can easily see that solving the minimization
\begin{align}  \label{J_MMSE_FD_opt_lowerbound}
\min_{\mathcal{Z},\mathcal{\tilde{V}}, \mathcal{K}}\,  &\sum_i\left({\bLambda}^{-1}_{ii}  + {\bGamma}^2_{ii}\left(\mathcal{Z}, \mathcal{\tilde V},  \mathcal{{K}} \right) \right)^{-1}\,
 {\rm s.t.}\, &\sum_{q=1}^{Q}z_q^2\leq E_T
\end{align}
yields a lower bound to the optimization problem in \eqref{J_MMSE_FD_opt}. From \eqref{diag_fact}, the only way to achieve this lower bound is when ${\bGamma}^2\left(\mathcal{Z}, \mathcal{\tilde V},  \mathcal{{K}} \right)$ is diagonal. By choosing $\bPsi^H \bUpsilon \left(\mathcal{\tilde{V}}, \mathcal{{K}}\right) = \bI$, or equivalently, by choosing $\mathcal{\tilde{V}}$ and $ \mathcal{K}$ such that $ \bUpsilon \left(\mathcal{\tilde{V}}, \mathcal{{K}}\right) = \bPsi$ holds true, we can ensure that $ {\bGamma}^2\left(\mathcal{Z}, \mathcal{\tilde V},  \mathcal{{K}} \right)$ is diagonal. We emphasize that the optimality of $ {\bGamma}^2\left(\mathcal{Z}, \mathcal{\tilde V},  \mathcal{{K}} \right)$ diagonal is yet to be proven. Due to the block structure of $\mathbfcal{D} \left(\mathcal{Z}\right)$ given in \eqref{gamma_long_mtx_D_def}, the existing proof for diagonality of  ${\bGamma}^2\left(\mathcal{Z}, \mathcal{\tilde V},  \mathcal{{K}} \right)$ in the literature (like in \cite{1261339, 1597555}), does not apply here. Note that, since the unitary matrix $\bPsi$ and $ \bUpsilon \left(\mathcal{\tilde{V}}, \mathcal{{K}}\right) $ are $MN \times MN$ and $MN\times L$ matrices, we first need to ensure $L = MN$. This is possible only when $Q=MN/N^{\rm RF}$. In what follows, we show how to find $\left(\mathcal{Z}, \mathcal{\tilde V},  \mathcal{{K}} \right)$ first for the case $Q=MN/N^{\rm RF}$  and later look at the case  $Q<MN/N^{\rm RF}$.

\subsection{Fully-Digital Beamforming Design for $Q=MN/N^{\rm RF}$} \label{Section:All_Q}

We now show how the values of $\mathcal{\tilde{V}}$ and $\mathcal{K}$ can be chosen such that
$\bUpsilon\left(\mathcal{\tilde{V}}^{\rm o}, \mathcal{K}^{\rm o} \right)=\bPsi $ holds true, where $\mathcal{\tilde{V}}^{\rm o}, \mathcal{K}^{\rm o}$ are used to denote the solution of $\mathcal{\tilde{V}}$ and $\mathcal{K}$ in \eqref{J_MMSE_FD_opt}, respectively. To this end, we choose
\begin{align} \label{mathcal_V_tilde_opt}
\mathcal{\tilde{V}}^{\rm o}&=\left(\tilde{\bu}_{\rm t}^{(1)},  \ldots, \tilde{\bu}_{\rm t}^{(Q)}\right),\nonumber\\
 \mathcal{K}^{\rm o}&=\left(\tilde{\bU}_{\rm r}^{(1)}, \ldots, \tilde{\bU}_{\rm r}^{(Q)} \right).
\end{align}
To find $\mathcal{{Z}}^{\rm o}$, we solve the following minimization problem:
\begin{align} \label{min_MMSE_FD_mim_min_outer}
\min_{\mathcal{Z}}\, &{\rm Tr}\left(\left(\bLambda^{-1}  + \bGamma^2(\mathcal{Z}, \mathcal{\tilde{V}}^{\rm o}, \mathcal{K}^{\rm o}) \right)^{-1}\right),\,
 {\rm s.t.} &\sum_{q=1}^{Q}z_q^2\leq E_T,
\end{align}
where
\begin{align} \label{gamma_short_mtx_v_tilde_opt}
\bGamma^2(\mathcal{Z}, \mathcal{\tilde{V}}^{\rm o}, \mathcal{K}^{\rm o})=\frac{1}{\sigma_{\rm n}^2} {\rm blkdg}\left(z_1^2 \bI_{N^{\rm RF}}, \cdots, z_Q^2\bI_{N^{\rm RF}}\right).
\end{align}

Note that, at the optimum, the total energy constraint in \eqref{min_MMSE_FD_mim_min_outer} has to be satisfied with equality. Otherwise, if at the optimum  $\sum_{q=1}^{Q}z_q^{,2} < E_T$ holds, then $\{z_q\}_{q=1}^Q$ can be scaled up such that this constraint is satisfied with equality. Doing so reduces the objective function, thereby contradicting the optimality. Defining $\alpha_q\triangleq z_q^2$, we can now define  the Lagrangian function as:
\begin{align} \label{Lagrangian_multi1}
\mathcal{L}(\alpha_1, \ldots,\alpha_Q, \mu)&= \sum_{q=1}^{Q}{\rm Tr}\left(\left( \left({{\tilde{\bLambda}}^{(q)}}\right)^{-1}+ \frac{1}{\sigma_{\rm n}^2}\alpha_q \bI_{N^{\rm RF}}\right)^{-1}\right)\nonumber\\
&+\mu_0 \left[ \sum_{q=1}^{Q} \alpha_q-E_T\right] -\sum_{q=1}^Q \mu_q \alpha_q,
\end{align}
where $\mu_0$ is the lagrange multiplier, and ${\tilde{\bLambda}}^{(q)}$ denotes an $N^{\rm RF}\times N^{\rm RF}$ diagonal matrix corresponding to $q$th block diagonal matrix  of ${{{\bLambda}}}$. For $\alpha_q > 0$ (which implies $\mu_q = 0$, due to the complementary slackness condition), setting $\partial \mathcal{L}(\alpha_1, \ldots,\alpha_Q, \mu)/\partial \alpha_q = 0$, $ \forall \, q$, yields
\begin{align} \label{Lagrangian_multi1_derive}
\left[ \sum_{k=1}^{N^{\rm RF}}\left(\left({{\tilde{\lambda}}_q^{(k)}}\right)^{-1}+\frac{1}{\sigma_{\rm n}^2}\alpha_q \right)^{-2}- \mu_0 \sigma_{\rm n}^2\right]=0,
\end{align}
where ${\tilde{\lambda}}_q^{(k)}$ denotes the $k$th diagonal entry of ${{\tilde{\bLambda}}^{(q)}}$. Therefore, the optimal positive $\alpha_q$ is obtained as the positive solution of the following equation:
  \begin{align}
 &\mu_0= \frac{1}{\sigma_{\rm n}^2} \sum_{k=1}^{N^{\rm RF}}\left(\left({{\tilde{\lambda}}_q^{(k)}}\right)^{-1}+\frac{1}{\sigma_{\rm n}^2}\alpha_q \right)^{-2}, \qquad \forall q. \label{water_fill_mod1_eqs}
\end{align}
while satisfying $\sum_{q=1}^{Q}\alpha_q =  E_T$.

Note that if \eqref{water_fill_mod1_eqs} does not have a positive solution for $\alpha_q$, then we have to choose $\alpha_q = 0$. Note also that the solution for $\alpha_q$ in \eqref{water_fill_mod1_eqs} is a non-increasing function of $\mu_0$. For a given $\mu_0$, once all $\alpha_q$'s are obtained, we need to verify the constraint $
        \sum_{q=1}^{Q}\alpha_q =  E_T$. If $\sum_{q=1}^{Q}\alpha_q >  E_T$, then  at least one of the  $\alpha_q$'s has to be reduced meaning that the optimal value of $\mu_0$ is larger than the given  $\mu_0$,  If $\sum_{q=1}^{Q}\alpha_q <  E_T$, then at least one of the $\alpha_q$'s has to be increased meaning that the optimal value of $\mu_0$ is smaller than the given  $\mu_0$. Based on this explanation, we can devise a bisection-based algorithm to find the optimal value of $\mu_0$, and consequently, the optimal values of $\{\alpha_q\}_{q=1}^{Q}$.

{\bf Algorithm to solve \eqref{water_fill_mod1_eqs}}:
 Let $\mu_0$ be a value in the interval $\left[\mu_{\rm min}, \mu_{\rm max}\right]$, where $\mu_{\rm min}$ and  $\mu_{\rm max}$ are respectively the lower- and upper-bound for $\mu_0$. We set $\mu_0= \frac{1}{2}\left(\mu_{\rm min}+\mu_{\rm max}\right)$ and solve for $\alpha_q$ in \eqref{water_fill_mod1_eqs} using the \emph{Newton-Raphson} method\cite{Ortega_book}. If $\alpha_q < 0$, then we set  $\alpha_q =0$. Now, if $\sum_{q=1}^{Q}\alpha_q < E_T$, implying that one or more of $\left\{\alpha_q\right\}_{q=1}^{Q} $ have to be increased. To do so, we decrease $\mu_{\rm max}$ by setting $\mu_{\rm max}=\mu_0$. Otherwise, if $\sum_{q=1}^{Q}\alpha_q >  E_T$, then $\mu_{\rm min}=\mu_0$. We continue this process until $|\sum_{q=1}^{Q}\alpha_q -  E_T|^2\leq \varepsilon$, for some small $\varepsilon$.

We now explain how $\mu_{\rm min}$ and  $\mu_{\rm max}$ can be chosen. For extremely high $E_T$, i.e., $E_T \rightarrow \infty$, each $\left\{\alpha_q\right\}_{q=1}^{Q}$ is unbounded, implying that $\mu_0 \rightarrow 0$. Therefore, we set $\mu_{\rm min}=0$. On the other hand, when $E_T \rightarrow 0$, each $\left\{\alpha_q\right\}_{q=1}^{Q}$ becomes zero. Defining $\displaystyle\mu_0^{(q)}\triangleq\frac{1}{\sigma_{\rm n}^2} \sum_{k=1}^{N^{\rm RF}} \left({{\tilde{\lambda}}_q^{(k)}}\right)^{2}$
by  setting $\alpha_q=0$ in \eqref{water_fill_mod1_eqs}, we can then
choose \begin{align} \label{mu_max}
\mu_{\rm max}=\max\left\{\mu_0^{(1)},\mu_0^{(2)}, \cdots, \mu_0^{(Q)}\right\}.
\end{align}
The bisection-based algorithm is summarized in Algorithm \ref{alg_water}.
\begin{algorithm}
    \caption{Algorithm to solve \eqref{water_fill_mod1_eqs}  }
    \textbf{Inputs}: $\bLambda$, $E_T$, $M$, $N$, $N^{\rm RF}$, and $\varepsilon$.\\
    \textbf{Outputs}: $\alpha_q^{\rm o}$ for $q=1,2, \cdots, Q$.
    \begin{algorithmic}[1]
        \State \label{Step2}Set $\mu_{\rm min}=0$ and obtain $\mu_{\rm max}$ as in \eqref{mu_max}, and choose an arbitrarily small value for $\varepsilon$,
        \State \hskip0.0em \textbf{Repeat}
        \State \hskip1.0em  Choose $\mu_0= \frac{1}{2}\left(\mu_{\rm min}+\mu_{\rm max}\right)$.
        \State \hskip1.0em Solve for $\alpha_q$ in system of equations in \eqref{water_fill_mod1_eqs} using the \emph{Newton-Raphson} method
        \State \hskip1.0em \textbf{If}  $\alpha_q < 0$ : Set $\alpha_q = 0$.
        \State \hskip1.0em \textbf{If}  $\sum_{q=1}^{\tilde{Q}}\alpha_q < E_T$: Set $\mu_{\rm max}=\mu_0$.
        \State \hskip1.0em \textbf{else If}  $\sum_{q=1}^{\tilde{Q}}\alpha_q > E_T$: Set $\mu_{\rm min}=\mu_0$.
        \State \hskip0.0em \textbf{Until} $|\sum_{q=1}^{\tilde{Q}}\alpha_q -   E_T|^2\leq \varepsilon$
    \end{algorithmic}
    \label{alg_water}
\end{algorithm}

Once we calculate $\alpha_q$, $\forall q$, we obtain $\mathcal{Z}^{\rm o}$ as:
\begin{align} \label{mathcal_Z_tilde_opt}
\mathcal{{Z}}^{\rm o}=\left( \sqrt{\alpha_1^{\rm o}},  \sqrt{\alpha_2^{\rm o}}, \cdots,  \sqrt{\alpha_Q^{\rm o}}\right).
\end{align}
Now, given $\left(\mathcal{Z}^{\rm o}, \mathcal{\tilde{V}}^{\rm o}, \mathcal{K}^{\rm o}\right)$, we explain how to find the beamformers $\mathcal{V}^{\rm o}$ and $ \mathcal{W}^{\rm o}$. To find $\mathcal{W}^{\rm o}$, we note that the objective function as well as the constraint in \eqref{J_MMSE_FD_opt} do not depend on $\bD_q$ and $\bZ_q$ (see \eqref{gamma_long_mtx_body}), therefore, there are infinite solution for $\mathcal{W}^{\rm o}$. For the sake of simplicity, we choose  $\bD_q=\bI$ and $\bZ_q=\bI$. One solution for $\mathcal{W}^{\rm o}$ as well as the solution for  $\mathcal{V}^{\rm o}$ are given by:
\begin{align} \label{V_W_opt_sol}
\mathcal{W}^{\rm o}&=\mathcal{K}^{\rm o}, \nonumber\\
\mathcal{V}^{\rm o}&=\left(\| \bv_1^{\rm o}\| \tilde{\bv}_1^{\rm o}, \| \bv_2^{\rm o}\| \tilde{\bv}_2^{\rm o}, \cdots, \| \bv_Q^{\rm o}\| \tilde{\bv}_Q^{\rm o} \right).
\end{align}
Recall that $\tilde{\bv}_q^{\rm o}$ and $\| \bv_q^{\rm o}\|$ are the $q$th entries of $\mathcal{\tilde{V}}^{\rm o}$ and $\mathcal{Z}^{\rm o}$, respectively.
\subsection{Fully-Digital Beamforming Design for $Q< MN/N^{\rm RF}$}\label{SectionB}
In this subsection, we consider the case where the number of available training time slots, $Q$, is less than $MN/N^{\rm RF}$, and aim to solve \eqref{J_MMSE_FD_opt}.
 Since $Q< MN/N^{\rm RF}$ (or equivalently $L < MN$), the dimensions of $\bUpsilon \left(\mathcal{\tilde{V}}, \mathcal{{K}}\right) $ and $\mathbfcal{D}\left(\mathcal{Z}\right)$ are different from the case where $Q= MN/N^{\rm RF}$. Indeed, $\bGamma^2(\mathcal{Z}, \mathcal{\tilde{V}}, \mathcal{{K}})$ is an $MN \times MN$ matrix whose rank is  at most $L$. As we have seen in the case $Q = MN/N^{\rm RF}$, not all $Q$ time slots may be used. Indeed, given the  transmitter and receiver correlation matrices\footnote{Although obtaining $\bS_{\rm t}$ and $\bS_{\rm r}$ requires some training, we argue that the rate at which $\bS_{\rm t}$ and $\bS_{\rm r}$ change depends on large-scale variations in the scattering environment (such as array geometry, pattern of applied antenna, path loss exponent, and delay profile) and is typically slow \cite{1494825, 1021913, 656151, Gazor_space}. The actual channel coefficients, on the other hand, can vary at a much faster rate due to the phase variations induced by the relative motion between the transmitter/receiver and  scattering objects. Thus, in general, $\bS_{\rm t}$ and $\bS_{\rm r}$ vary much more slowly and do not need to be estimated in each coherence interval. This allows for the possibility of covariance feedback to the transmitter \cite{937059}. Such an assumption is well adopted in the literature, see for example  \cite{4803746, 8227727, Eldar_9026804, 1261339}.}, i.e., $\bS_{\rm t}$ and $\bS_{\rm r}$, we can obtain the corresponding eigenvalue matrices $ \bLambda_{\rm t}$ and $\bLambda_{\rm r}$, and then, calculate $ \bLambda = \bLambda_{\rm t} \otimes \bLambda_{\rm r}$, from which we can obtain the power of each transmit beamformer. That is, we can in advance determine the number of transmit beamformers with non-zero powers. We use $Q_{\rm nz}$ to denote this number, with $Q_{\rm nz} \in \{1,2, \cdots, MN/N^{\rm RF}\}$. If $Q_{\rm nz} < Q$, the solution can be obtained from the approach obtained in the previous subsection. Otherwise, the solution in the previous subsection cannot be applied and we use the solution presented below.

 To efficiently solve the  problem in this case, inspired by the solution in the previous subsection, we impose a diagonal structure on $\bGamma^2\left(\mathcal{Z}, \mathcal{\tilde{V}}, \mathcal{{K}} \right)$, and solve the following optimization problem instead:
\begin{align} \label{min_MMSE_FD_mim_min_outer_less_snap_diag}
 \min_{\mathcal{Z}, \mathcal{\tilde{V}}, \mathcal{K}}\,\,\,\, &{\rm Tr}\left(\left(\bLambda^{-1}  + \bGamma^2(\mathcal{Z}, \mathcal{\tilde{V}}, \mathcal{{K}}) \right)^{-1}\right) \\
{\rm s.t.}\,\,\,\, & \sum_{q=1}^{Q}z_q^2\leq E_T,\quad {\rm and } \quad
{\bGamma}^2\left(\mathcal{Z},\mathcal{\tilde{V}}, \mathcal{K} \right) \,\,\, \mbox{ is diagonal}.\nonumber
\end{align}
Let  $\tilde{\bLambda}\triangleq {\rm diag}\left({\rm Tr}\left({\tilde{\bLambda}}^{(1)}\right), {\rm Tr}\left({\tilde{\bLambda}}^{(2)}\right), \ldots, {\rm Tr}\left({\tilde{\bLambda}}^{(Q)}\right)\right)$ be a $Q \times Q$ diagonal matrix, with  ${\tilde{\bLambda}}^{(q)}$ defined right after \eqref{Lagrangian_multi1}. Let also $\tilde{\bLambda} = \mathbfcal{\tilde{U}} \tilde{\bLambda}_{\rm s}  \mathbfcal{\tilde{U}}^H$ denote the eigenvalue decomposition of $\tilde{\bLambda}$
where, $\tilde{\bLambda}_{\rm s}$ is a diagonal matrix of eigenvalues in decreasing order, and $\mathbfcal{\tilde{U}}$ is the corresponding matrix of eigenvectors. Since $\tilde{\bLambda}_{\rm s}$ and $\tilde{\bLambda}$ are both diagonal matrices, $\mathbfcal{\tilde{U}}$ is just a permutation matrix. We define ${\mathbfcal{U}}\triangleq \mathbfcal{\tilde{U}} \otimes \bI_{N^{\rm RF}}$, and $ {\bLambda}_{\rm s} \triangleq\mathbfcal{{U}}^H {\bLambda}   {\mathbfcal{U}} $, and using the fact that $\mathbfcal{\tilde{U}}$ is a permutation matrix, the diagonal entries of ${\bLambda}_{\rm s}$ and ${\bLambda}$ are the same but in different order}. Using this notation, the objective function in \eqref{min_MMSE_FD_mim_min_outer_less_snap_diag} is expressed as:
\begin{align} \label{obj_svd}
&J_{\rm MMSE}\left(\mathcal{Z},\mathcal{\tilde{V}}, \mathcal{K} \right) = {\rm Tr}\left(\left( \mathbfcal{{U}} {\bLambda}_{\rm s}^{-1}  {\mathbfcal{U}}^H  + \bGamma^2(\mathcal{Z}, \mathcal{\tilde{V}}, \mathcal{{K}}) \right)^{-1}\right) \nonumber\\
&= {\rm Tr}\left(\left({\bLambda}_{\rm s}^{-1}  + \frac{1}{\sigma_{\rm n}^2}  \tilde{ \bPsi}^H \bUpsilon \left(\mathcal{\tilde{V}}, \mathcal{{K}}\right) \mathbfcal{D}\left(\mathcal{Z}\right) \bUpsilon^H \left(\mathcal{\tilde{V}}, \mathcal{{K}}\right)  \tilde{\bPsi} \right)^{-1}\right)\nonumber\\
&\triangleq {\rm Tr}\left(\left({\bLambda}_{\rm s}^{-1}  + \tilde{\bGamma}^2\left(\mathcal{Z},\mathcal{\tilde{V}}, \mathcal{K} \right) \right)^{-1}\right),
\end{align}
where $ \tilde{\bPsi}\triangleq \bPsi  \mathbfcal{{U}} $. Therefore, the optimization problem \eqref{min_MMSE_FD_mim_min_outer_less_snap_diag} is equivalent to:
\begin{align} \label{min_MMSE_FD_mim_min_outer_less_snap_order}
\min_{\mathcal{Z}, \mathcal{\tilde{V}}, \mathcal{K}}\,\,\,\, & {\rm Tr}\left(\left({\bLambda}_{\rm s}^{-1}  + \tilde{\bGamma}^2\left(\mathcal{Z},\mathcal{\tilde{V}}, \mathcal{K} \right) \right)^{-1}\right)\\
{\rm s.t.}\,\,\,\, & \sum_{q=1}^{Q}z_q^2\leq E_T, \quad {\rm and } \quad \tilde{\bGamma}^2\left(\mathcal{Z},\mathcal{\tilde{V}}, \mathcal{K} \right) \,\,\, \mbox{ is diagonal}. \nonumber
\end{align}
Since $\tilde{\bGamma}^2\left(\mathcal{Z} ,\mathcal{\tilde{V}} , \mathcal{K}  \right)$ in \eqref{min_MMSE_FD_mim_min_outer_less_snap_order} is an $MN \times MN$ diagonal matrix of rank $L < MN$, then, there have to be at least $MN-L$ zero entries on its diagonal. We show in the Appendix \ref{gama_prop} that at the optimum, the non-zero diagonal entries of
$\tilde{\bGamma}^2\left(\mathcal{Z^{\rm o}},\mathcal{\tilde{V}}^{\rm o},
\mathcal{K^{\rm o}} \right)$ are the same as the diagonal entries of
$\mathbfcal{D}(\mathcal{Z}^{\rm o})$. This implies that the non-zero diagonal entries of $\tilde{\bGamma}^2\left(\mathcal{Z}^{\rm o},\mathcal{\tilde{V^{\rm o}}},
\mathcal{K^{\rm o}} \right)$ have multiplicity order $N^{\rm RF}$, therefore, it can be expressed as $\tilde{\bGamma}^2\left(\mathcal{Z^{\rm o}},
\mathcal{\tilde{V}}^{\rm o}, \mathcal{K^{\rm o}} \right)= \frac{1}{\sigma_{\rm n}^2} \bSigma^{\rm o}\otimes \bI_{N^{\rm RF }} $,  where $\bSigma^{\rm o} \triangleq {\rm diag}\left(\gamma_1^{\rm o}, \gamma_2^{\rm o} \ldots,
\gamma^{\rm o}_{\frac{MN}{N^{\rm RF}}}\right)$. Given this observation, in order to solve \eqref{min_MMSE_FD_mim_min_outer_less_snap_order}, we use the following strategy. We first find  $\bSigma^{\rm o}$ by solving
\begin{align} \label{gamma_opt_prob}
\min_{\bSigma \in \mathcal{D}}\,\,\,\, &{\rm Tr}\left(\left({\bLambda}_{\rm s}^{-1}  + \frac{1}{\sigma_{\rm n}^2} \bSigma \otimes \bI_{N^{\rm RF }} \right)^{-1}
\right)\nonumber\\
{\rm s.t.}\,\,\,\,  &{\rm Rank}\left(\bSigma\right)\le Q, \qquad {\rm Tr}\left(\bSigma\right)\leq  E_T,
\end{align}
where $\mathcal{D}\triangleq\left\{\bSigma \mid \bSigma \in R_{+}^{ \frac{MN}{N^{\rm RF}}\times \frac{MN}{N^{\rm RF}}}, \hbox{$\bSigma$ is a diagonal matrix}\right\}$. Then, using the obtained $\bSigma^{\rm o}$, we find
$\mathcal{V^{\rm o}}$, $\mathcal{\tilde{V^{\rm o}}}$ and $\mathcal{K^{\rm o}} $ such that
\begin{align} \label{opt_eq}
\tilde{\bGamma}^2\left(\mathcal{Z}^{\rm o},\mathcal{\tilde{V}}^{\rm o}, \mathcal{K}^{\rm o} \right)=\frac{1}{\sigma_{\rm n}^2} \bSigma^{\rm o} \otimes \bI_{N^{\rm RF }}.
\end{align}
Next, we show how to find $\mathcal{Z}^{\rm o}$, $\mathcal{\tilde{V}}^{\rm o}$, and $\mathcal{K}^{\rm o} $ based on $\bSigma^{\rm o}$.

\subsubsection{Finding $\mathcal{Z}^{\rm o}$}
  Since $\bSigma^{\rm o}$ is a rank-$Q$ diagonal matrix, it has $Q$ non-zero entries on its diagonal. Let
\begin{align} \label{svd_Sigma}
\bSigma^{\rm o}\otimes \bI_{N^{\rm RF }} = \mathbfcal{F}\,\, {\rm blkdg}\left(\bSigma_{\rm s}^{\rm o}, \b0\right)\,\, \mathbfcal{F}^H
\end{align}
denote the eigenvalue decomposition of $\bSigma^{\rm o}\otimes \bI_{N^{\rm RF }}$, where $\bSigma_{\rm s}^{\rm o}$ is an $L\times L$ diagonal matrix\footnote{Recall that $L=QN^{\rm RF}$.} with entries in decreasing order and $\mathbfcal{F}$ is a unitary matrix. Using the fact that $\mathbfcal{D}\left(\mathcal{Z}^{\rm o}\right) =\bSigma_{\rm s}^{\rm o}$ (see Appendix \ref{gama_prop}), $\mathcal{Z}^{\rm o}$ can be directly obtained from the diagonal entries of $\bSigma_{\rm s}^{\rm o}$. Note that the diagonal entries of $\bSigma_{\rm s}^{\rm o}$ have multiplicity order of $N^{\rm RF}$, so we can write $\bSigma_{\rm s}^{\rm o}= \check{\bSigma}_{\rm s}^{\rm o}\otimes \bI_{N^{\rm RF}}$, where $\check{\bSigma}_{\rm s}^{\rm o}$ is a $Q\times Q$ diagonal matrix with unit multiplicity order. Therefore, $\mathcal{Z}^{\rm o} = {\rm diag}\left(\check{\bSigma}_{\rm s}^{\rm o}\right)$.

\subsubsection{Finding $\mathcal{\tilde{V}}^{\rm o}$ and $\mathcal{K}^{\rm o} $} Using \eqref{opt_eq} and \eqref{svd_Sigma} along with the fact that $\mathbfcal{D}\left(\mathcal{Z}^{\rm o}\right) =\bSigma_{\rm s}^{\rm o}$, we can write
\begin{align} \label{vtild_K_Qless}
 &\mathbfcal{F}^H\tilde{ \bPsi}^H \bUpsilon \left(\mathcal{\tilde{V}}^{\rm o}, \mathcal{{K}}^{\rm o}\right) \mathbfcal{D}\left(\mathcal{Z}^{\rm o}\right) \bUpsilon^H \left(\mathcal{\tilde{V}}^{\rm o}, \mathcal{{K}}^{\rm o}\right)  \tilde{\bPsi} \mathbfcal{F}\nonumber\\
 &= {\rm blkdg}\left(\mathbfcal{D}\left(\mathcal{Z}^{\rm o}\right), \b0\right).
\end{align}
To satisfy \eqref{vtild_K_Qless}, we choose $\mathcal{\tilde{V}}^{\rm o}$ and $\mathcal{{K}}^{\rm o}$ such that the $MN\times L$ matrix $\bUpsilon\left(\mathcal{\tilde{V}}^{\rm o}, \mathcal{{K}}^{\rm o}\right)$ corresponds to the first $L$ columns of $ \tilde{\bPsi} \mathbfcal{F}$, or equivalently, $ \bPsi  \mathbfcal{{U}} \mathbfcal{F}$. Specifically, due to the structure of $\mathbfcal{{U}}$ and  $\mathbfcal{F}$, we express $\mathbfcal{{U}} \mathbfcal{F} = \mathbfcal{{T}} \otimes \bI_{N^{\rm RF}}$, where $\mathbfcal{{T}}$ is a permutation matrix. Let $i_q$ denote the index of the non-zero entry of the $q$th column of $\mathbfcal{{T}}$, then  $\mathcal{\tilde{V}}^{\rm o}$ and $\mathcal{K}^{\rm o}$ are given by:
\begin{align} \label{mathcal_V_tilde_opt_case_B}
&\mathcal{\tilde{V}}^{\rm o}=\left(\tilde{\bu}_{\rm t}^{(i_1)}, \ldots, \tilde{\bu}_{\rm t}^{(i_Q)}\right), \,\,\,\mathcal{K}^{\rm o}=\left(\tilde{\bU}_{\rm r}^{(i_1)}, \ldots, \tilde{\bU}_{\rm r}^{(i_Q)} \right).
\end{align}

\subsubsection{Solving \eqref{gamma_opt_prob} to find $\bSigma^{\rm o}$} Due to the rank constraint, the optimization problem \eqref{gamma_opt_prob} is non-convex, and in general, may not be amenable to a computationally affordable solution. However, we propose an algorithm to find a suboptimal but efficient solution. To this end, we propose to allocate $E_T$ to the first $Q$ diagonal entries of $\bSigma^{\rm o} $ while the last $\frac{MN}{L}-Q$ diagonal entries of $\bSigma$ are set to zero, i.e.,  $\gamma_q^{\rm o}=0$ for $q=\frac{MN}{L}-Q+1, \dots, \frac{MN}{L}$. To find $\gamma_q^{\rm o}$, $\forall\, q$, let $\check{\bSigma} \triangleq {\rm blkdg}\left(\gamma_1\otimes \bI_{N^{\rm RF}}, \gamma_2\otimes \bI_{N^{\rm RF}}, \ldots, \gamma_Q\otimes \bI_{N^{\rm RF}}\right)$, where $\gamma_q \in \mathbb{R}_+$, $\forall\, q$. Let also $\check{\tilde{{\bLambda}}}_{\rm s}$ be an $L \times L$ diagonal matrix that captures the first $L$ diagonal entries of ${\tilde{{\bLambda}}}_{\rm s}$. Then, $\left\{\gamma_q^{\rm o}\right\}_{q=1}^{Q}$ is the solution to the following optimization problem:
\begin{align} \label{gamma_opt_prob_reduced}
\min_{\gamma_1, \ldots, \gamma_Q}\,{\rm Tr}\left(\left(\check{\tilde{{\bLambda}}}_{\rm s}^{-1}  + \frac{1}{\sigma_{\rm n}^2} \check{\bSigma} \right)^{-1}\right), \,\,\,{\rm s.t.}\,\sum_{q=1}^{Q}\gamma_q\leq E_T.
\end{align}
The above minimization is the same as \eqref{min_MMSE_FD_mim_min_outer}. Therefore, we use Algorithm \ref{alg_water} to find $\left\{\gamma_1^{\rm o}, \ldots, \gamma_Q^{\rm o}\right\}$.

The proposed fully-digital beamforming design procedure is summarized in the following algorithm.
\begin{algorithm}
    \caption{Fully-digital beamforming design algorithm}
    \textbf{Inputs}: $\bLambda$, $E_T$, $M$,$N$, $N^{\rm RF}$, and $Q$.\\
    \textbf{Outputs}:  $\mathcal{V}^{\rm o}$  and $\mathcal{W}^{\rm o}$.
    \begin{algorithmic}[1]
        \State \label{Step1}Find $\alpha_q^{\rm o}$, $q=1,2, \cdots, MN/N^{\rm RF}$ using the Algorithm \ref{alg_water}.
        \State Find $Q_{\rm nz}$, the number of non-zero $\alpha_q^{\rm o}$ obtained in Step \ref{Step1}.
        \State \textbf{If} $Q_{\rm nz} \leq Q$:
        \State \hskip1.0em Obtain $\mathcal{Z}^{\rm o} $, $\mathcal{\tilde{V}}^{\rm o}$ and $\mathcal{K}^{\rm o} $ from \eqref{mathcal_Z_tilde_opt} and \eqref{mathcal_V_tilde_opt}, respectively.
        \State \textbf{else}:
        \State \hskip1.0em Obtain $\mathcal{\tilde{V}}^{\rm o}$ and $\mathcal{K}^{\rm o} $ from \eqref{mathcal_V_tilde_opt_case_B}.
        \State \hskip1.0em Find $\check{\bSigma}_{\rm s}^{\rm o}$ by solving \eqref{gamma_opt_prob_reduced} and using  Algorithm \ref{alg_water}. Then, set $\mathcal{Z}^{\rm o} = {\rm diag}\left(\check{\bSigma}_{\rm s}^{\rm o}\right)$.
        \State Obtain $\mathcal{V}^{\rm o}$  and $\mathcal{W}^{\rm o}$ using \eqref{V_W_opt_sol}.
    \end{algorithmic}
    \label{alg_summary}
\end{algorithm}

\section{Hybrid Beamformering Design} \label{sec:HB_design}
So far, we assumed that $\bv_q$ and $\bW_q$ are both \emph{fully-digital} without any constraint except the total energy constraint imposed on $\bv_q$. In this section, we consider \emph{hybrid} structures for  both $\bv_q$ and $\bW_q$, where they are both constructed from the cascades of digital-only and analog-only beamformers as in \eqref{vq_Wq_split}, respectively. Specifically, given the unit norm constraints in \eqref{norm_cons}, we find the analog and digital precoders and combiners to realize $\bv_q$ and $\bW_q$ by minimizing the Euclidean distance between the corresponding fully-digital beamformers and the hybrid structured beamformers. This widely-used technique is shown to be effective in the literature \cite{7370753, 7397861, DBLP:journals/corr/abs-1712-03485, 6717211}.

\subsection{Hybrid Precoder Design} \label{HB_Precoder_design}
Here, we aim to show how to split the unconstrained precoders, designed in the previous section, into digital and RF precoders. Specifically, given $\bv_q$ we seek to find  $\bV_q^{\rm RF}$ and $\bv_q^{\rm d}$ such that $\bv_q= \bV_q^{\rm RF}\bv_q^{\rm d}$. Specifically, we solve
\begin{align}\label{precoder_split}
\min_{\bV_q^{\rm RF},\bv_q^{\rm d}}&\parallel \bv_q - \bV_q^{\rm RF} {\bv_q^{\rm d}} \parallel^2,\;\;\;\;
 {\rm s.t.}\,\,\,&\bV_q^{\rm RF} \in \mathcal{A}_v.
\end{align}
To solve \eqref{precoder_split}, we use the following proposition.

{\textbf{Proposition 1}} \label{proposition}: Any fully-digital beamformer $\bV_q \in \mathbb{C}^{N\times t}$ can be \emph{ideally} realized by a cascade of an analog beamformer   $\bV_q^{\rm RF} \in \mathcal{A}_v$ and a digital beamformer   $\bV_q^{\rm d}\in \mathbb{C}^{N^{\rm RF}\times t}$, if $N^{\rm RF}\geq 2t$.

\emph{Proof}: See \cite{Foad_JSTSP_2016}.

Since $\bv_q$ is an $N^{\rm RF}\times 1$ vector and $N^{\rm RF}\geq 2$, we use the above proposition to express $\bv_q$ as a cascade of  $\bV_q^{\rm RF} \in \mathcal{A}_v$ and $\bv_q^{\rm d}\in \mathcal{C}^{N^{\rm RF}\times 1}$. Let $ f_{i}e^{j\omega_{i}}$ be the $i$th entry of $\bv_q$, $\bv_q^{\rm d} = [v_{1}^{\rm d}\; \cdots\; v_{N^{\rm RF}}^{\rm d}]^T$, and $\bV_q^{\rm RF} (i,k)=e^{j\phi_{i,k}}$. Then, following the result of \label{cite_foad}\cite{Foad_JSTSP_2016}, if we choose
\begin{align}\label{v_sol}
v_{k}^{\rm d}&=
\left\{
  \begin{array}{ll}
    f_{\rm max}, &k=1,2, \\
    0, &k=3,\ldots, N^{\rm RF}
  \end{array}
\right., \nonumber\\
\phi_{i,k}&=
\left\{
  \begin{array}{ll}
    \omega_{i}-\cos^{-1}\left(\frac{f_i}{2f_{\rm max}}\right), & k=1, \\
    \omega_{i}+\cos^{-1}\left(\frac{f_i}{2f_{\rm max}}\right), &  k=2,\\
     0,  &   k=3, \ldots, N^{\rm RF},
  \end{array}
\right.
\end{align}
where $f_{\rm max}= \max_{i}f_{i}$, then $\bv_q= \bV_q^{\rm RF}\bv_q^{\rm d}$ holds true. Next, we explain how to find the hybrid combiners.
\subsection{Hybrid Combiner Design}\label{HB_Combiner_design}
Given the unconstrained optimal fully-digital combiner matrices $\{\bW_q\}_{q=1}^{Q}$, in this subsection, we explain how to split $\bW_q$ into its digital and analog counterparts. Considering the norm constraint on each entries of $\bW_q^{\rm RF}$, we aim to find $\bW_q^{\rm RF}$ and $\bW_q^{\rm d}$ such that their product is sufficiently close to $\bW_q$. Mathematically, we solve the following minimization problem:
\begin{align}\label{Joint_min}
\min_{\bW_q^{\rm RF},\bW_q^{\rm d}}&\parallel \bW_q - \bW_q^{\rm RF}\bW_q^{\rm d} \parallel^2,
 \,\,\,{\rm s.t.}\,\,\,\, \bW_q^{\rm RF} \in \mathcal{A}_w.
\end{align}
Due to the non-convex unit-modulus constraint in \eqref{Joint_min}, the joint minimization in \eqref{Joint_min} is highly complicated, and in general, may not be amenable to a computationally affordable  solution. \label{PE_Altmin}Inspired by the work in \cite{7397861}, we tackle \eqref{Joint_min} by decoupling the optimization over $\bW_q^{\rm RF}$ and $\bW_q^{\rm d}$. Indeed, we alternately solve for $\bW_q^{\rm RF}$ while fixing $\bW_q^{\rm d}$ and vice versa. To this end, two alternating approaches, namely MO-AltMin technique and PE-AltMin algorithm, have been proposed in \cite{7397861}. In this paper, we rely on the PE-AltMin technique which is summarized as follows. In the first step of this technique, given $\bW_q^{\rm RF}$, we obtain $\bW_q^{\rm d}$ as $\bW_q^{\rm d}= \bV_q\bU_{q}^H$,
where $\bU_{q}$  and $\bV_q$ are $N^{\rm RF}\times N^{\rm RF}$ matrices whose columns are, respectively,  the left and the right singular vectors of $\bW_q^H \bW_q^{\rm RF}$. In the second step, given $\bW_q^{\rm d}$, we obtain the phase of each elements of $\bW_q^{\rm RF}$ equal to the phase of the corresponding element in $\bW_q {\bW_q^{\rm d}}^H$. The PE-AltMin algorithm is summarized in Algorithm \ref{alg_HB}. Compared to PE-AltMin technique, the MO-AltMin algorithm is computationally more demanding while does not provide much performance gain as it is shown in our simulation results.

\textbf{Remark 1}: \label{remark1} It is worth mentioning that in our proposed scheme, determining $\mathcal{W}^{\rm o}$ requires only the knowledge of $\bU_{\rm r}$, (i.e., the eigen-vectors of $\bS_{\rm r}$, and not that of $\bS_{\rm t}$). On the other hand, obtaining  $\mathcal{\tilde{V}}^{\rm o}$ only requires the knowledge of $\bU_{\rm t}$ (i.e.,  the eigenvectors of $\bS_{\rm t}$, and not that of $\bS_{\rm r}$). In addition, finding $\mathcal{Z}^{\rm o}$ does not require full knowledge of matrix  $\bS_{\rm r}$ but only  the eigenvalues of $\bS_{\rm r}$ (i.e., $\bLambda_{r}$). These eigenvalues can be made available to the transmitter with a low feedback overhead.

\begin{algorithm}[t]
    \caption{PE-AltMin algorithm}
    \textbf{Inputs}: $\bW_q$, $N^{\rm RF}$.\\
    \textbf{Outputs}: $\bW_q^{\rm RF}$ and $\bW_q^{\rm d}$.
    \begin{algorithmic}[1]
        \State Choose ${\bW_q^{\rm RF}}^{(0)} $ with random phases and set $k=0$.
        \State \hskip0.0em \textbf{Repeat}:
        \State \hskip1.0em Compute the SVD $\bW_q^H{\bW_q^{\rm RF}}^{(k)} = \bU_{q}^{(k)}\bS^{(k)}{\bV_q^{(k)}}^H$.
        \State \hskip1.0em Set ${\bW_q^{\rm d}}^{(k)}= \bV_q^{(k)}{\bU_q^{(k)}}^H$.
        \State \hskip1.0em Set $\arg[{\bW_q^{\rm RF}}^{(k+1)}]_{ij}= \arg[\bW_q {{\bW_q^{\rm d}}^{(k)}}^H ]_{ij}$.
        \State \hskip1.0em Set $k \leftarrow k+1$.
        \State \hskip0.0em \textbf{Until} a stoping criterion triggers.
    \end{algorithmic}
    \label{alg_HB}
\end{algorithm}

{\bf Remark 2:} \label{remark2}
The main computational cost for each iteration of Algorithm \ref{alg_water} is to solve the non-linear equation in \eqref{water_fill_mod1_eqs} using {Newton-Raphson} method. For each value of $q$, $q= 1, \cdots, Q$, the complexity of solving \eqref{water_fill_mod1_eqs} is given by $\mathcal{O}\left( N^{\rm RF} \log(1/\varepsilon) \right)$, where $\varepsilon$ is the tolerance error term. Therefore, for $Q=MN/N^{\rm RF}$, the overall complexity of \eqref{water_fill_mod1_eqs} is given by $\mathcal{O}\left( MN \log(1/\varepsilon) \right)$. Moreover, the number of iterations needed
for the bisection method is $\lceil \log_2 ({\mu_{\rm max}}/{\varepsilon}) \rceil$ (which is directly
proportional to the logarithm of the initial value $\mu_{\rm max}$). Thus, a proper selection
for $\mu_{\rm max}$, such as in \eqref{mu_max}, will reduce the total cost.
In summary, the computational cost of Algorithm \ref{alg_water} is given by  $\mathcal{O}\left( MN \log(1/\varepsilon)\log_2 ({\mu_{\rm max}}/{\varepsilon}) \right)$.

In Algorithm \ref{alg_summary}, we first use Algorithm \ref{alg_water}  to find $\alpha_q^{\rm o}$ and $Q_{\rm nz}$, with complexity $\mathcal{O}\left( MN \log(1/\varepsilon)\log_2 ({\mu_{\rm max}}/{\varepsilon}) \right)$. The complexity of the remaining part of Algorithm \ref{alg_summary} is dominated by the eigenvalue decomposition of $\bS_t$ and $\bS_r$, which are of complexity $\mathcal{O}\left( N^3 \right)$ and $\mathcal{O}\left( M^3 \right)$, respectively. Overall, the complexity of Algorithm \ref{alg_summary} is  $\mathcal{O}\left( (MN)^3 \log(1/\varepsilon)\log_2 ({\mu_{\rm max}}/{\varepsilon}) \right)$. In each iteration of Algorithm \ref{alg_HB}, the
update of the analog precoder is simply realized by the SVD of $\bW_q^H{\bW_q^{\rm RF}}^{(k)}$ followed by  a phase extraction operation of the matrix $\bW_q {{\bW_q^{\rm d}}^{(k)}}^H$, whose dimension is $M \times N^{\rm RF}$. Therefore, the computational complexity of the Algorithm  \ref{alg_HB} is dominated by  the SVD operation in each iteration which is $\mathcal{O}\left({N^{\rm RF}}^3\right)$.  It should be noted that there are numerous algorithms to reduce the computation complexity of SVD such as the techniques in \cite{5158128, 7345023}.
 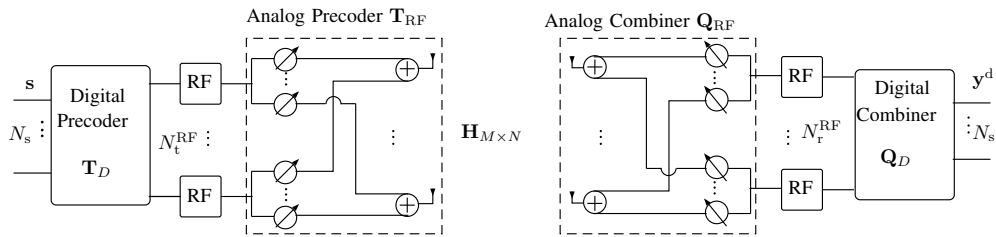
\begin{figure*}
    \centering
\psscalebox{0.75 .75} 
{
\begin{pspicture}(0,-2.0011537)(17.668888,2.0011537)
\psframe[linecolor=black, linewidth=0.02, dimen=outer, framearc=0.1](2.5185184,1.001555)(0.7453281,-1.4901876)
\psline[linecolor=black, linewidth=0.02](0.087066665,0.37990654)(0.75510585,0.37990654)
\psline[linecolor=black, linewidth=0.02](0.087066665,-0.9139004)(0.75510585,-0.9139004)
\psline[linecolor=black, linewidth=0.02](2.4870667,0.6702832)(3.0452063,0.6702832)
\psline[linecolor=black, linewidth=0.02](2.4870667,-1.3297168)(3.0452063,-1.3297168)
\psline[linecolor=black, linewidth=0.02](2.4870667,-1.3297168)(3.0452063,-1.3297168)
\psframe[linecolor=black, linewidth=0.02, dimen=outer, framearc=0.1](3.7760527,-0.9576238)(3.0318666,-1.7018099)
\rput[bl](3.1666667,-1.4232051){\rm RF}
\psline[linecolor=black, linewidth=0.02](3.7670667,0.6702832)(4.3252063,0.6702832)
\psline[linecolor=black, linewidth=0.02](4.306667,1.1167948)(4.306667,0.3167948)(4.306667,0.3167948)
\psline[linecolor=black, linewidth=0.02](4.306667,1.1167948)(4.7066665,1.1167948)
\psline[linecolor=black, linewidth=0.02](4.306667,0.3167948)(4.7066665,0.3167948)
\pscircle[linecolor=black, linewidth=0.02, dimen=outer](4.9066668,1.0967948){0.2}
\psline[linecolor=black, linewidth=0.02, arrowsize=0.05291667cm 2.0,arrowlength=1.4,arrowinset=0.0]{->}(4.7356324,0.9236914)(5.1494255,1.3374845)
\pscircle[linecolor=black, linewidth=0.02, dimen=outer](4.9066668,0.29679483){0.2}
\psline[linecolor=black, linewidth=0.02, arrowsize=0.05291667cm 2.0,arrowlength=1.4,arrowinset=0.0]{->}(4.7356324,0.123691365)(5.1494255,0.53748447)
\psline[linecolor=black, linewidth=0.02](5.0866666,1.1167948)(7.0866666,1.1167948)
\pscircle[linecolor=black, linewidth=0.02, dimen=outer](7.0666666,0.91679484){0.2}
\psline[linecolor=black, linewidth=0.02](7.0866666,0.7167948)(5.7533336,0.7167948)(5.7533336,-0.8832052)(5.0866666,-0.8832052)
\psline[linecolor=black, linewidth=0.02](7.0666,1.0701948)(7.0666,0.80352813)(7.0666,0.80352813)
\psline[linecolor=black, linewidth=0.02](6.9332666,0.93686146)(7.1999335,0.93686146)
\psline[linecolor=black, linewidth=0.02](3.7670667,-1.3297168)(4.3252063,-1.3297168)
\psline[linecolor=black, linewidth=0.02](4.306667,-0.8832052)(4.306667,-1.6832051)(4.306667,-1.6832051)
\psline[linecolor=black, linewidth=0.02](4.306667,-0.8832052)(4.7066665,-0.8832052)
\psline[linecolor=black, linewidth=0.02](4.306667,-1.6832051)(4.7066665,-1.6832051)
\pscircle[linecolor=black, linewidth=0.02, dimen=outer](4.9066668,-0.90320516){0.2}
\psline[linecolor=black, linewidth=0.02, arrowsize=0.05291667cm 2.0,arrowlength=1.4,arrowinset=0.0]{->}(4.7356324,-1.0763086)(5.1494255,-0.6625155)
\pscircle[linecolor=black, linewidth=0.02, dimen=outer](4.9066668,-1.7032052){0.2}
\psline[linecolor=black, linewidth=0.02, arrowsize=0.05291667cm 2.0,arrowlength=1.4,arrowinset=0.0]{->}(4.7356324,-1.8763087)(5.1494255,-1.4625155)
\rput{-1.1021739}(-0.0048984517,0.110853136){\psarc[linecolor=black, linewidth=0.02, dimen=outer](5.76,0.31006148){0.13333334}{0.0}{180.0}}
\psline[linecolor=black, linewidth=0.02](5.62,0.3167948)(5.0866666,0.3167948)
\psline[linecolor=black, linewidth=0.02](5.886667,0.3167948)(6.153333,0.3167948)
\psline[linecolor=black, linewidth=0.02, arrowsize=0.05291667cm 2.0,arrowlength=1.4,arrowinset=0.0]{-<}(7.253267,0.9435282)(7.519933,0.9435282)(7.519933,1.2101948)
\rput[bl](1.0932667,0.27519482){\rm Digital}
\rput[bl](0.8605057,-0.10787785){\rm Precoder}
\rput[bl](1.2932667,-0.97551286){$\bT_D$}
\rput[bl](4.2153845,1.6911538){Analog Precoder $\bT_{\rm RF}$}
\psline[linecolor=black, linewidth=0.02](13.692307,-1.1854115)(13.134169,-1.1854115)
\psline[linecolor=black, linewidth=0.02](13.152708,-1.6319231)(13.152708,-0.8319231)(13.152708,-0.8319231)
\psline[linecolor=black, linewidth=0.02](13.152708,-1.6319231)(12.7527075,-1.6319231)
\psline[linecolor=black, linewidth=0.02](13.152708,-0.8319231)(12.7527075,-0.8319231)
\rput{-180.0}(25.105415,-3.2238462){\pscircle[linecolor=black, linewidth=0.02, dimen=outer](12.552708,-1.6119231){0.2}}
\rput{-180.0}(25.105415,-1.6238463){\pscircle[linecolor=black, linewidth=0.02, dimen=outer](12.552708,-0.81192315){0.2}}
\psline[linecolor=black, linewidth=0.02](12.372707,-1.6319231)(10.372707,-1.6319231)
\rput{-180.0}(20.785416,-2.8638463){\pscircle[linecolor=black, linewidth=0.02, dimen=outer](10.392708,-1.4319232){0.2}}
\psline[linecolor=black, linewidth=0.02](10.372707,-1.2319231)(11.706041,-1.2319231)(11.706041,0.36807686)(12.372707,0.36807686)
\psline[linecolor=black, linewidth=0.02](10.392775,-1.5853231)(10.392775,-1.3186564)(10.392775,-1.3186564)
\psline[linecolor=black, linewidth=0.02](10.526108,-1.4519898)(10.259441,-1.4519898)
\psline[linecolor=black, linewidth=0.02](13.692307,0.8145885)(13.134169,0.8145885)
\psline[linecolor=black, linewidth=0.02](13.152708,0.36807686)(13.152708,1.1680769)(13.152708,1.1680769)
\psline[linecolor=black, linewidth=0.02](13.152708,0.36807686)(12.7527075,0.36807686)
\psline[linecolor=black, linewidth=0.02](13.152708,1.1680769)(12.7527075,1.1680769)
\rput{-180.0}(25.105415,0.77615374){\pscircle[linecolor=black, linewidth=0.02, dimen=outer](12.552708,0.38807687){0.2}}
\rput{-180.0}(25.105415,2.3361537){\pscircle[linecolor=black, linewidth=0.02, dimen=outer](12.552708,1.1680769){0.2}}
\rput{-180.0}(20.785416,1.9361538){\pscircle[linecolor=black, linewidth=0.02, dimen=outer](10.392708,0.9680769){0.2}}
\psline[linecolor=black, linewidth=0.02](10.392708,0.83467686)(10.392708,1.1013435)(10.392708,1.1013435)
\psline[linecolor=black, linewidth=0.02](10.526041,0.9680102)(10.259375,0.9680102)
\rput{-181.10217}(23.38071,-1.8752688){\psarc[linecolor=black, linewidth=0.02, dimen=outer](11.699374,-0.82518977){0.13333334}{0.0}{180.0}}
\psline[linecolor=black, linewidth=0.02](11.839375,-0.8319231)(12.372707,-0.8319231)
\psline[linecolor=black, linewidth=0.02](11.572708,-0.8319231)(11.306041,-0.8319231)
\psline[linecolor=black, linewidth=0.02](11.3061075,-0.8319231)(11.3061075,0.7680769)(10.372774,0.7680769)(10.372774,0.7680769)
\psline[linecolor=black, linewidth=0.02](10.372707,1.1480769)(12.372707,1.1480769)(12.372707,1.1480769)
\rput{-180.0}(23.021313,-0.5151283){\psframe[linecolor=black, linewidth=0.02, linestyle=dashed, dash=0.17638889cm 0.10583334cm, dimen=outer](13.254246,1.4860256)(9.767067,-2.001154)}
\psline[linecolor=black, linewidth=0.02, arrowsize=0.05291667cm 2.0,arrowlength=1.4,arrowinset=0.0]{-<}(10.193795,-1.4457334)(9.988667,-1.4457334)(9.988667,-1.2406052)
\psline[linecolor=black, linewidth=0.02, arrowsize=0.05291667cm 2.0,arrowlength=1.4,arrowinset=0.0]{-<}(10.193795,0.9542666)(9.988667,0.9542666)(9.988667,1.1593949)
\psline[linecolor=black, linewidth=0.02, arrowsize=0.05291667cm 2.0,arrowlength=1.4,arrowinset=0.0]{->}(12.726666,0.94448715)(12.326667,1.4367948)
\psline[linecolor=black, linewidth=0.02, arrowsize=0.05291667cm 2.0,arrowlength=1.4,arrowinset=0.0]{->}(12.726666,0.14448713)(12.326667,0.6367948)
\psline[linecolor=black, linewidth=0.02, arrowsize=0.05291667cm 2.0,arrowlength=1.4,arrowinset=0.0]{->}(12.726666,-1.0555129)(12.326667,-0.5632052)
\psline[linecolor=black, linewidth=0.02, arrowsize=0.05291667cm 2.0,arrowlength=1.4,arrowinset=0.0]{->}(12.726666,-1.8555129)(12.326667,-1.3632052)
\psframe[linecolor=black, linewidth=0.02, dimen=outer, framearc=0.1](14.436052,1.1623948)(13.691867,0.41820878)
\rput[bl](13.826667,0.6967948){\rm RF}
\psframe[linecolor=black, linewidth=0.02, dimen=outer, framearc=0.1](14.436052,-0.8376052)(13.691867,-1.5817913)
\rput[bl](13.826667,-1.2832052){\rm RF}
\psline[linecolor=black, linewidth=0.02](14.427067,0.7903018)(14.985207,0.7903018)
\psline[linecolor=black, linewidth=0.02](14.427067,-1.2096982)(14.985207,-1.2096982)
\psframe[linecolor=black, linewidth=0.02, dimen=outer, framearc=0.1](16.77364,0.9289949)(14.9970665,-1.3976719)
\psline[linecolor=black, linewidth=0.02](16.735922,0.330261)(17.405235,0.330261)
\psline[linecolor=black, linewidth=0.02](16.735922,-0.669739)(17.405235,-0.669739)
\rput[bl](15.352381,0.41523412){\rm Digital}
\rput[bl](15.032464,0.004977723){\rm Combiner}
\rput[bl](15.448236,-0.81553507){$\bQ_D$}
\rput[bl](9.5,1.6056837){\rm{Analog Combiner} $\bQ_{\rm RF}$}
\rput[bl](8.022222,-0.3720052){$\bH_{M\times N}$}
\psdots[linecolor=black, dotsize=0.047876142](0.5979428,-0.0080518)
\psdots[linecolor=black, dotsize=0.047876142](0.5979428,-0.1519385)
\psdots[linecolor=black, dotsize=0.047876142](0.5979428,-0.25433852)
\psdots[linecolor=black, dotsize=0.04](3.5378666,-0.18980518)
\psdots[linecolor=black, dotsize=0.04](3.5378666,-0.2919385)
\psdots[linecolor=black, dotsize=0.04](3.5378666,-0.41433853)
\psdots[linecolor=black, dotsize=0.04](6.817867,-0.16980518)
\psdots[linecolor=black, dotsize=0.04](17.017866,0.06980518)
\psdots[linecolor=black, dotsize=0.04](17.017866,-0.0719385)
\psdots[linecolor=black, dotsize=0.04](17.017866,-0.19433852)
\psdots[linecolor=black, dotsize=0.04](13.817866,-0.16980518)
\psdots[linecolor=black, dotsize=0.04](13.817866,-0.2719385)
\psdots[linecolor=black, dotsize=0.04](13.817866,-0.39433852)
\psdots[linecolor=black, dotsize=0.04](10.497867,-0.16980518)
\psdots[linecolor=black, dotsize=0.04](10.497867,-0.2719385)
\psdots[linecolor=black, dotsize=0.04](10.497867,-0.39433852)

\psdots[linecolor=black, dotsize=0.04](12.517867,0.67980518)
\psdots[linecolor=black, dotsize=0.04](12.517867,0.7719385)
\psdots[linecolor=black, dotsize=0.04](12.517867,0.87433852)

\psdots[linecolor=black, dotsize=0.04](12.517867,-1.07980518)
\psdots[linecolor=black, dotsize=0.04](12.517867,-1.1719385)
\psdots[linecolor=black, dotsize=0.04](12.517867,-1.27433852)

\psdots[linecolor=black, dotsize=0.04](4.917867,0.65980518)
\psdots[linecolor=black, dotsize=0.04](4.917867,0.7519385)
\psdots[linecolor=black, dotsize=0.04](4.917867,0.85433852)

\psdots[linecolor=black, dotsize=0.04](4.917867,-1.19980518)
\psdots[linecolor=black, dotsize=0.04](4.917867,-1.2919385)
\psdots[linecolor=black, dotsize=0.04](4.917867,-1.39433852)

\rput[bl](2.6488667,-0.5720941){$N_{\rm t}^{\rm RF}$}
\rput[bl](14.048866,-0.3720941){$N_{\rm r}^{\rm RF}$}
\rput[bl](0.0,-0.3720052){$N_{\rm s}$}
\rput[bl](0.3,0.5720052){$\bs $}
\rput[bl](17.088889,-0.35200518){$N_{\rm s}$}
\rput[bl](17.088889,0.5720052){$\by^{\rm d} $}
\pscircle[linecolor=black, linewidth=0.02, dimen=outer](7.0666666,-1.4832052){0.2}
\psline[linecolor=black, linewidth=0.02](7.0666666,-1.3498052)(7.0666666,-1.6164719)(7.0666666,-1.6164719)
\psline[linecolor=black, linewidth=0.02](6.9333334,-1.4831386)(7.2,-1.4831386)
\psline[linecolor=black, linewidth=0.02](6.1532664,0.3167948)(6.1532664,-1.2832052)(7.0866,-1.2832052)(7.0866,-1.2832052)
\psline[linecolor=black, linewidth=0.02](7.0866666,-1.6632051)(5.0866666,-1.6632051)(5.0866666,-1.6632051)
\psline[linecolor=black, linewidth=0.02, arrowsize=0.05291667cm 2.0,arrowlength=1.4,arrowinset=0.0]{-<}(7.253267,-1.4564718)(7.519933,-1.4564718)(7.519933,-1.1898052)
\psframe[linecolor=black, linewidth=0.02, linestyle=dashed, dash=0.17638889cm 0.10583334cm, dimen=outer](7.6923075,1.4860256)(4.205128,-2.001154)
\psdots[linecolor=black, dotsize=0.04](6.817867,-0.2719385)
\psdots[linecolor=black, dotsize=0.04](6.817867,-0.39433852)
\psframe[linecolor=black, linewidth=0.02, dimen=outer, framearc=0.1](3.7760527,1.0423762)(3.0318666,0.29819018)
\rput[bl](3.1666667,0.59679484){\rm RF}
\end{pspicture}
}
\caption{System model for data transmission. }
 \label{system_model1}
\end{figure*}
\section{Simulation Results} \label{sec:Simulations}
 \subsection{Experiment Setup}
 We consider a correlated massive MIMO\footnote{\label{ft8} For example, this could be a massive MIMO system implemented with less than  half a wavelength antenna spacing. Note that  a massive MIMO system, even with antenna spacing greater than  half a wavelength, may exhibit correlations among the transmit and receive antennas mainly due to lack of a sufficient number of scatters around the transmitter or receiver \cite{6415388}. In the case where the propagation environment is not scattering-rich, significant correlation exists among the transmit antennas  and among receiver antennas \cite{Rappaport_book1, Janaswamy_book}, \cite[Section~8.2.2]{Ravi_course}). Similar analysis has been carried out in \cite{4277071}, where the inter-element correlation is shown to depend not only on  the antenna spacing but also on the  azimuth power spectrum and on the mean AOA/AOD of each clustered scatters.}
system in a rich-scattering environment. The spatial correlation matrices  $\bS_{\rm t}$ and $\bS_{\rm r}$ are given as \cite{4803746}:
\begin{align}
\bS_{\rm t} (i,j)= \rho_{\rm t}^{|i-j|}  \,\,\,\, {\rm and } \,\,\,\, \bS_{\rm r} (i,j)= \rho_{\rm r}^{|i-j|}, \nonumber
\end{align}
where we assume $\rho_{\rm r}=\rho_{\rm t} \triangleq |\rho | e^{j\theta}$,  $|\rho |<1$, and $\theta\in [0, 2\pi]$.

Given the hybrid precoder and combiner design provided in Sections \ref{HB_Precoder_design} and \ref{HB_Combiner_design}, the estimated $\bH$, denoted as $\hat{\bH}$, can be directly obtained from \eqref{kron_ch_decompose1} and  \eqref{MMSE_estimate}. To quantify the performance of the proposed channel estimation method, we use the normalized mean squared error in channel estimation defined as $ {\rm NMSE} = {||\hat{\bH}- {\bH}||_2^2}/{||{\bH}||_2^2}$ . We also compare the spectral efficiency of the proposed water-filling-type scheme with that of an equal power scheme and with the perfect CSI case. Both schemes use the same set of beamformers but different power allocations. In the equal power allocation scheme, we set $\alpha_q=E_T/Q$, for $q=1,2, \ldots, Q$, while in water-filling scheme $\alpha_q$, is obtained using Algorithm \ref{alg_water}. The annotated numbers on each simulation point indicate that, out of $Q$, how many training time slots are essentially being used.

To evaluate the spectral efficiency, we use the communication scheme  shown in Fig.~\ref{system_model1}. The system in  Fig.~\ref{system_model1} is the same as the one in Fig.~\ref{system_model}, except that we use $N_{\rm s}\leq N^{\rm RF}$ data streams at the input and collect $N_{\rm s}$ data streams at the output. The received signal $\by$ is written in terms of $\bs$ as:
\begin{align} \label{received signal_final}
\by^{\rm d} = \bQ_{\rm t}^H \bH \bT_{\rm t} \bs+ \bQ_{\rm t}{\bn},
\end{align}
where $\bs \triangleq [s_1 \,\, s_2 \,\,\ldots\,\ s_{N_{\rm s}} ]^T$ is the vector of transmit symbols, $\by^{\rm d} \triangleq [y^{\rm d}_1 \,\, y^{\rm d}_2 \,\,\ldots\,\ y^{\rm d}_{N_{\rm s}} ]^T$ is the vector of the final processed received signal. We define  $\bQ_{\rm t}\triangleq \bQ_{\rm RF}\bQ_D$, and $\bT_{\rm t}\triangleq \bT_{\rm RF}\bT_D$, where $\bT_{\rm RF}\in \mathbb{C}^{N\times N^{\rm RF}}$ and $\bQ_{\rm RF}\in \mathbb{C}^{M\times N^{\rm RF}}$ are, respectively, the analog precoder and combiner, implemented using the phase shifters, i.e., $|\bQ_{\rm RF}(i,j)|^2=1$, $|\bT_{\rm RF}(i,j)|^2=1$. Further, $\bT_D \in \mathbb{C}^{N^{\rm RF} \times N_{\rm s}}$ and $\bQ_D \in \mathbb{C}^{N^{\rm RF} \times N_{\rm s}}$ are, respectively, baseband digital linear precoder and combiner.

Before we calculate the spectral efficiency, we first need to find the optimal hybrid beamformers, denoted as $\bT_{\rm t}^{\rm o}$ and $\bQ^{\rm o}$, using the channel matrix estimates  $\hat{\bH}$. This can be done by any of the beamforming design approaches proposed in \cite{Foad_JSTSP_2016} and \cite{7397861}. To calculate the spectral efficiency, let $\bH_{\rm e}$ denote the error in channel estimation, defined as $\bH_{\rm e}\triangleq\bH-\hat{\bH}$; then $\by^{\rm d}$ in \eqref{received signal_final} can be written as:
\begin{align} \label{received signal_est}
\by^{\rm d}&={\bQ_{\rm t}^{\rm o}}^H \left(\hat{\bH}+\bH_{\rm e}\right) \bT_{\rm t}^{\rm o} \bs+ {\bQ_{\rm t}^{\rm o}}^H{\bn}\nonumber\\
&={\bQ_{\rm t}^{\rm o}}^H \hat{\bH} \bT_{\rm t}^{\rm o} \bs+ {\bQ_{\rm t}^{\rm o}}^H {\bH}_{\rm e} \bT_{\rm t}^{\rm o} \bs+{\bQ_{\rm t}^{\rm o}}^H{\bn}.
\end{align}
Therefore, the spectral efficiency is calculated as:
\begin{align} \label{spectral_eff} {\rm
SE}= \eta \log\left|\bI+\bE^{-1}{\bQ_{\rm t}^{\rm o}}^H \hat{\bH} \bT_{\rm t}^{\rm o}{\bT_{\rm t}^{\rm o}}^H\hat{\bH}^H \bQ_{\rm t}^{\rm o}\right|,
\end{align}
where $\bE \triangleq {\bQ_{t}^{\rm o}}^H \left({\bH}_{\rm e} \bT_{\rm t}^{\rm o} {\bT_{\rm t}^{\rm o}}^H{\bH}_{\rm e}^H + \sigma^2_{n} \bI\right){\bQ_{\rm t}^{\rm o}}$. The pre-log factor $\eta$ is used to denote the ratio of effective channel use and given by $\eta \triangleq 1- Q_{\rm nz}/T_{\rm c}B_{w}$,
where $B_{w}$ is the communication bandwidth, $T_{\rm c}$ is the coherence time, calculated as $T_{\rm c} \triangleq \sqrt{\frac{9}{16\pi}}\frac{1}{f_m}$ \cite{Rappaport_book}, $f_m\triangleq \frac{v}{c}f_c$ is the maximum Doppler spread, $v$ is the velocity, $c$ is the speed of light, and $f_c$ is the carrier frequency. The performance of the proposed channel estimation method in terms of the channel estimation accuracy and the achievable spectral efficiency is investigated through the simulations in the next subsections. The simulation parameters are given in Table \ref{tab:sim_table1}.
\vspace{-3 mm}
\begin{table}[h]
\caption {Simulation parameters} \label{tab:sim_table1}
\centering
\begin{tabular}{ c|c|c|c|c|c|c|c|c }
\hline
$|\rho|$&$N$& $M$&$N_{\rm s}$&$\sigma_{\rm n}^2$&$\varepsilon$&$B_w$&$v$&$f_c$\\
\hline
0.8&32&16& 4&1& $10^{-6}$&1.5 MHz&50 Km/h&2 GHz \\
 \hline
 \end{tabular}
\end{table}

\begin{figure*}[t]
    \centering
    \begin{minipage}[t]{.47\textwidth}
        \centering{
\psfrag{MSE (dB)}[c][c]{NMSE (dB)}
\psfrag{Total transmit power per stream(dB)}[c][c]{$E_T/\sigma_{\rm n}^2$ (dB)}
\psfrag{Equal power, FD}{ \small{Equal Power, FD}}
\psfrag{Equal power, NRF=8, HB}{\small{ Equal Power, $N^{\rm RF}=8$, HB, {PE-AltMin \cite{7397861}}}}
\psfrag{Water-filling, FD}{\small{ Water-filling, FD}}
\psfrag{Water-filling, NRF=8, HB}{\small{ Water-filling, $N^{\rm RF}=8$, HB, PE-AltMin \cite{7397861}}}
\psfrag{Water-filling, NRF=8, MOAltmin, HB, 123456789012}{\small{ Water-filling, $N^{\rm RF}=8$, HB, MO-AltMin \cite{7397861}}}
\resizebox{!}{6.6cm}{\includegraphics{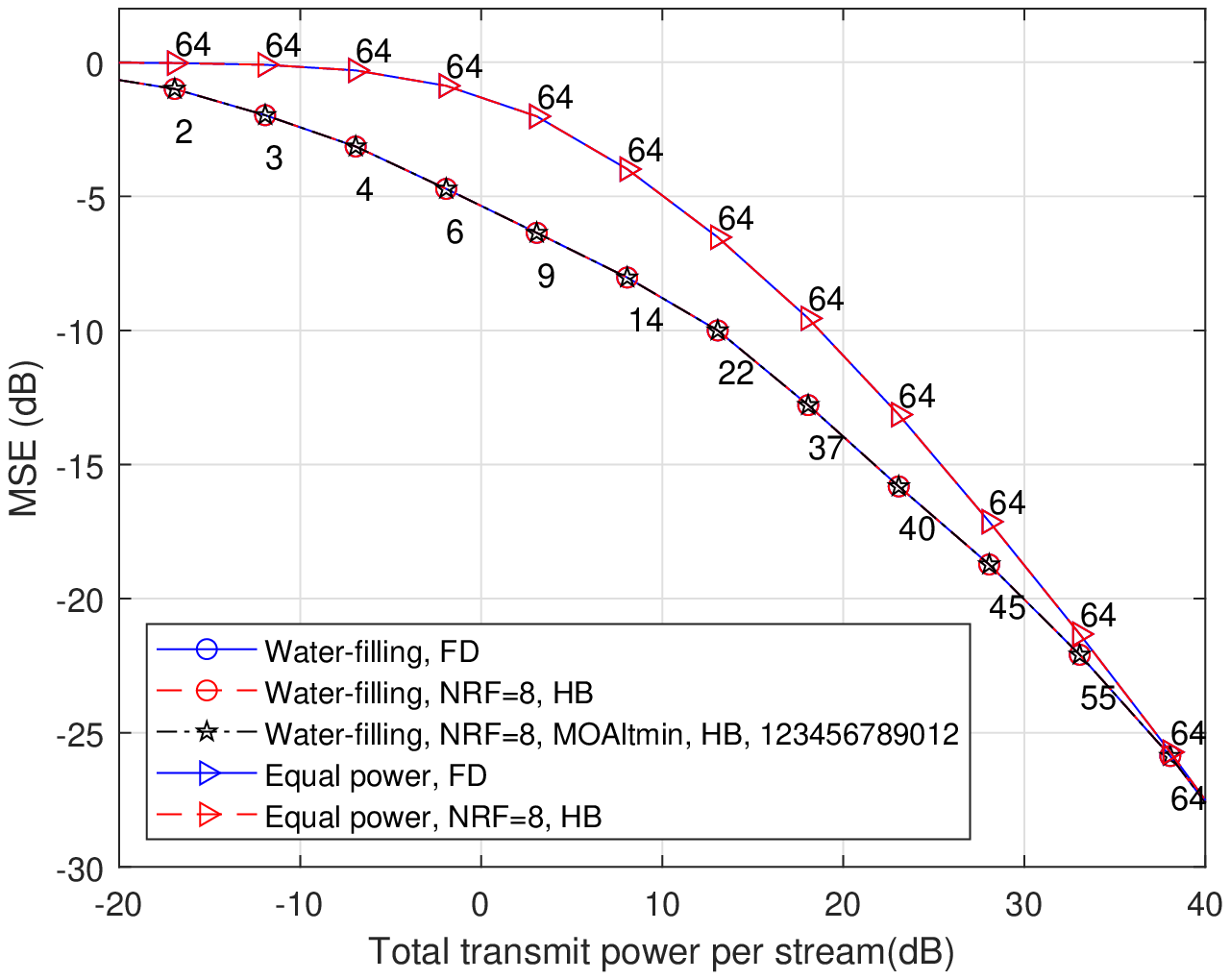}}
\caption{NMSE vs $E_T/\sigma_{\rm n}^2$, $M=16$, $N=32$, $N_{\rm s}=4$, $|\rho|=.8$, $Q=64$.}
\label{MSE_vs_P}
}
    \end{minipage}%
    \hfill
    \begin{minipage}[t]{0.47\textwidth}
       \centering{
\psfrag{Spectral efficiency }{SE (bits/s/Hz)}
\psfrag{Total transmit power (dB)}[c][c]{$E_T/\sigma_{\rm n}^2$ (dB)}
\psfrag{Perfect CSI, FD}{\small{Perfect CSI, FD}}
\psfrag{Total transmit power per stream(dB)}[c][c]{$E_T/\sigma_{\rm n}^2$ (dB)}
\psfrag{Equal power, FD}{ \small{Equal Power, FD}}
\psfrag{Equal power, NRF=8, HB}{\small{ Equal Power, $N^{\rm RF}=8$, HB, {PE-AltMin \cite{7397861}}}}
\psfrag{Water-filling, FD}{\small{ Water-filling, FD}}
\psfrag{Water-filling, NRF=8, HB}{\small{ Water-filling, $N^{\rm RF}=8$, HB, {PE-AltMin \cite{7397861}}}}
\psfrag{Water-filling, NRF=8, MOAltmin, HB, 123456789}{\small{ Water-filling, $N^{\rm RF}=8$, HB, MO-AltMin \cite{7397861}}}
\resizebox{!}{6.6cm}{\includegraphics{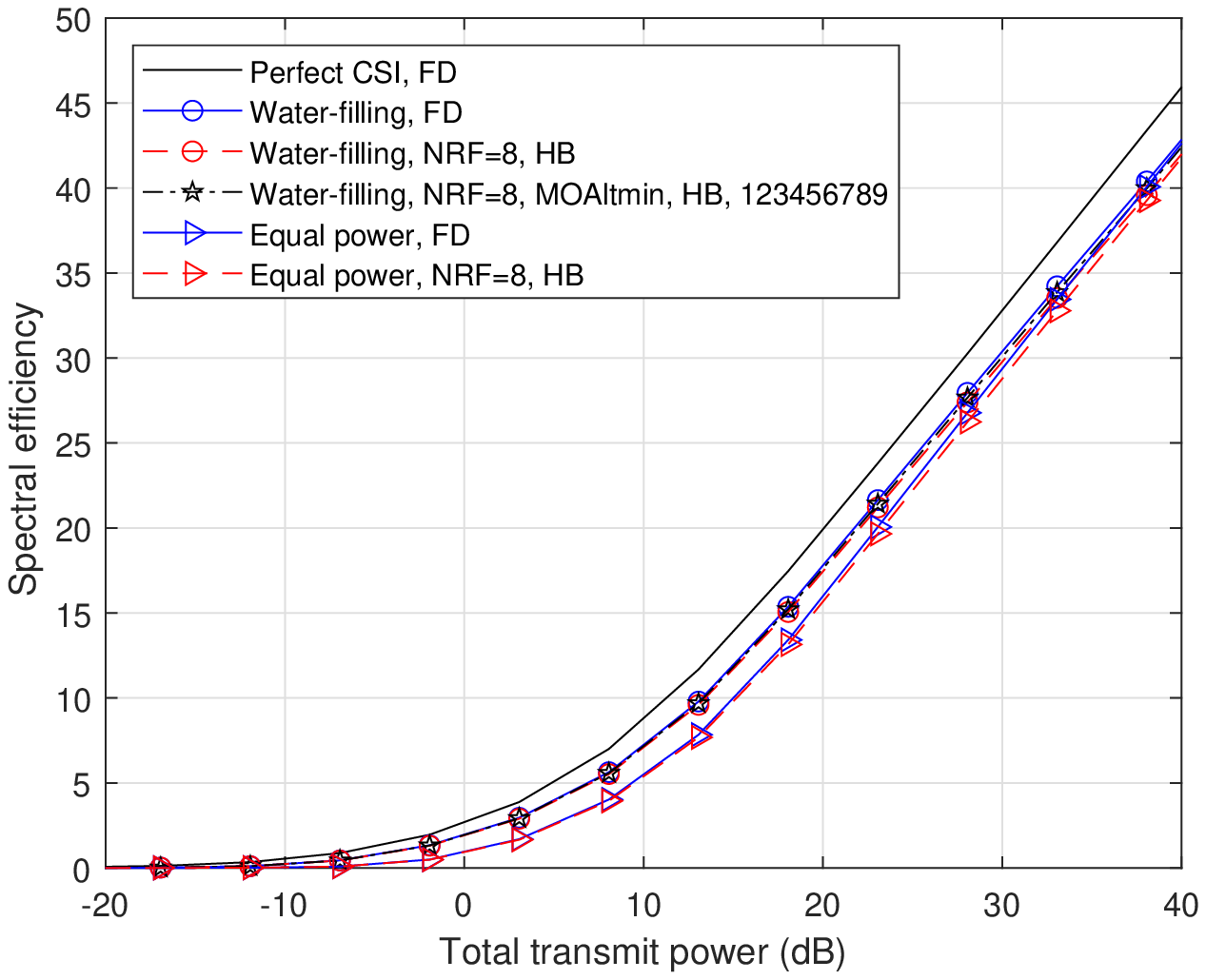}}
\caption{Spectral efficiency vs $E_T/\sigma_{\rm n}^2$, $M=16$, $N=32$, $N_{\rm s}=4$, $|\rho|=.8$, $Q=64$.}
\label{Rate_vs_P}}
    \end{minipage}
\end{figure*}

\begin{figure*}[t]
    \centering
    \begin{minipage}[t]{.47\textwidth}
        \centering{
\psfrag{MSE}{NMSE (dB)}
\psfrag{Qt}{ $Q$}
\psfrag{Equal power, FD}{ \small{Equal Power, FD}}
\psfrag{Equal power, NRF=8, HB}{\small{ Equal Power, $N^{\rm RF}=8$, HB, {PE-AltMin \cite{7397861}}}}
\psfrag{Water-filling, FD}{\small{ Water-filling, FD}}
\psfrag{Water-filling, NRF=8, HB}{\small{ Water-filling, $N^{\rm RF}=8$, HB, {PE-AltMin \cite{7397861}}}}
\psfrag{Water-filling, NRF=8, MOAltmin, HB, 123456789}{\small{ Water-filling, $N^{\rm RF}=8$, HB, MO-AltMin \cite{7397861}}}
\resizebox{!}{6.6cm}{\includegraphics{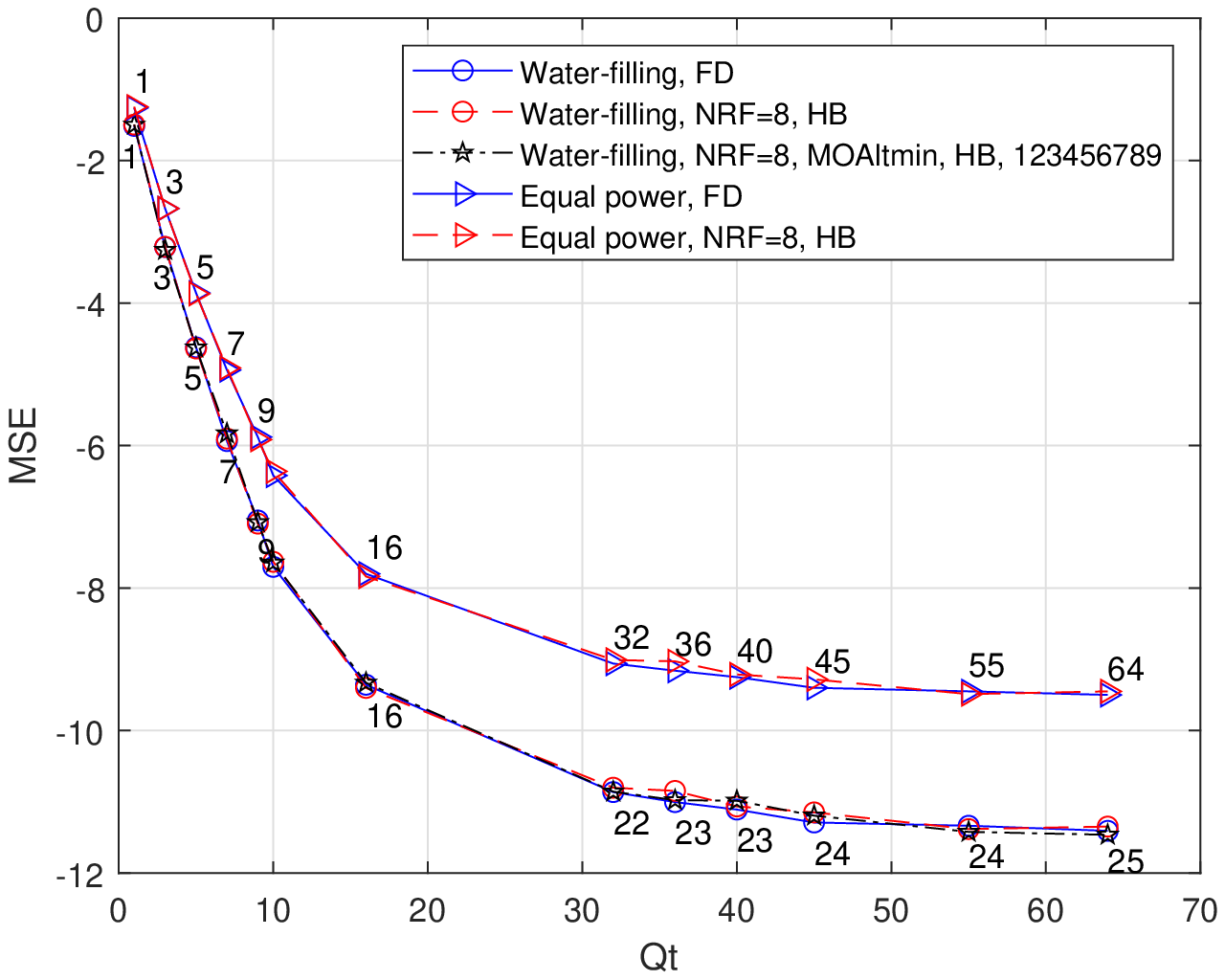}}
\caption{NMSE vs $Q$, $M=16$, $N=32$, $N_{\rm s}=4$, $|\rho|=.8$, $E_T/\sigma_{\rm n}^2 = 64$.}
\label{MSE_vs_Qt}}
    \end{minipage}%
    \hfill
    \begin{minipage}[t]{0.47\textwidth}
       \centering{
\psfrag{SE}[c][c]{SE (bits/s/Hz)}
\psfrag{Qt}{$Q$}
\psfrag{Perfect CSI, FD}{\small{Perfect CSI, FD}}
\psfrag{Equal power, FD}{ \small{Equal Power, FD}}
\psfrag{Equal power, NRF=8, HB}{\small{ Equal Power, $N^{\rm RF}=8$, HB, {PE-AltMin \cite{7397861}}}}
\psfrag{Water-filling, FD}{\small{ Water-filling, FD}}
\psfrag{Water-filling, NRF=8, HB}{\small{ Water-filling, $N^{\rm RF}=8$, HB, {PE-AltMin \cite{7397861}}}}
\psfrag{Water-filling, NRF=8, MOAltmin, HB, 123456789}{\small{ Water-filling, $N^{\rm RF}=8$, HB, MO-AltMin \cite{7397861}}}
\resizebox{!}{6.6cm}{\includegraphics{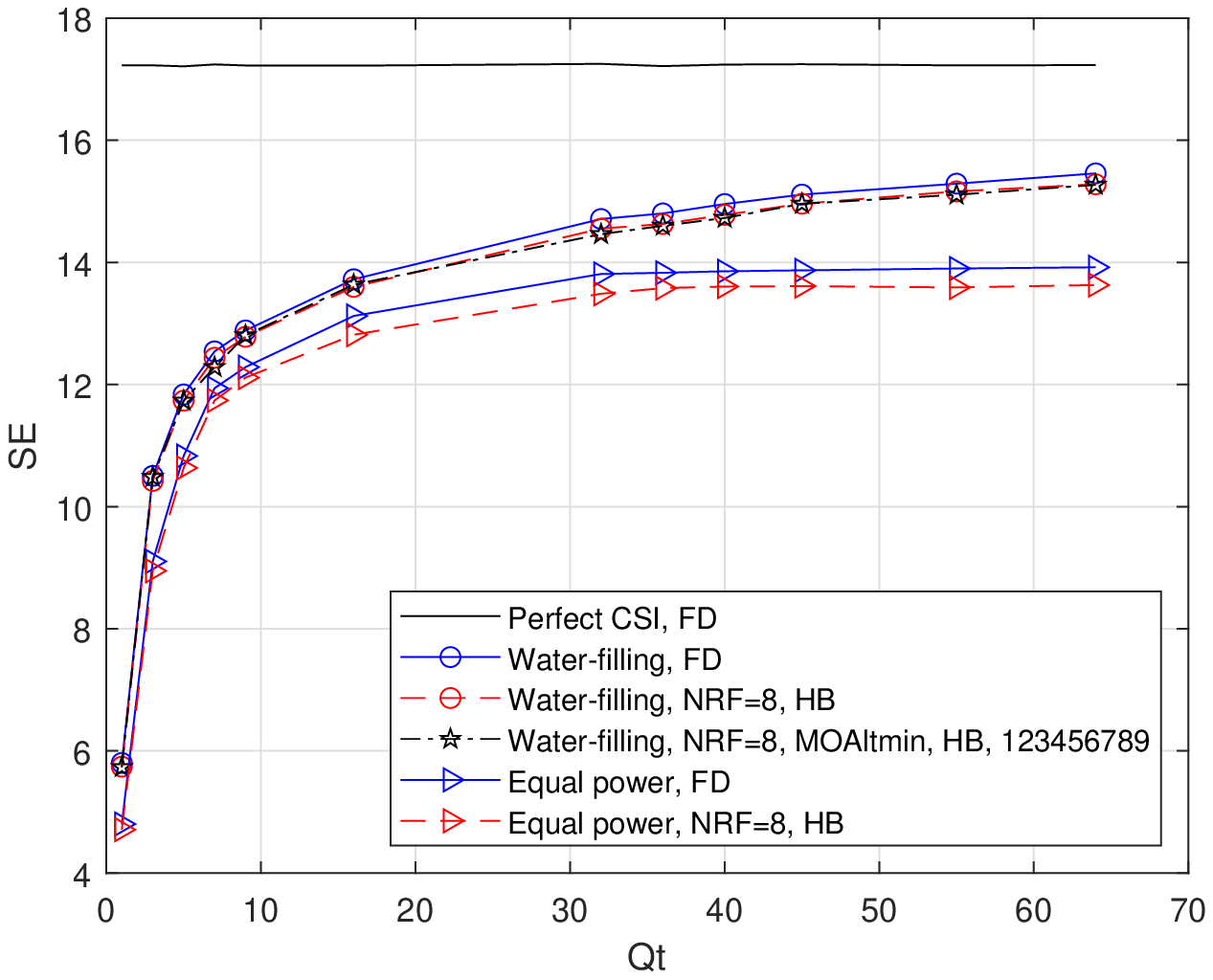}}
\caption{Spectral efficiency vs $Q$, $M=16$, $N=32$, $N_{\rm s}=4$, $|\rho|=.8$, $E_T/\sigma_{\rm n}^2 = 64$.}
\label{Rate_vs_Qt}}
    \end{minipage}
\end{figure*}

\label{sim_explanation} In all our simulations, for the sake of comparison, we provide simulation results for the MO-AltMin algorithm of \cite{7397861} that we use to solve \eqref{Joint_min}, i.e., to design the hybrid combiner.
\subsection{NMSE and Spectral Efficiency vs $E_T/\sigma_{\rm n}^2$}
In the first set of our numerical experiments, we plot the ${\rm NMSE}$ and the corresponding spectral efficiency vs. $E_T/\sigma_{\rm n}^2$. Here, we assume that $Q=MN/N^{\rm RF}$. The ${\rm NMSE}$ and the corresponding spectral efficiency vs. $E_T/\sigma_{\rm n}^2$ are depicted in Figs. \ref{MSE_vs_P} and \ref{Rate_vs_P}, respectively. We observe that for very low $E_T$, both the water-filling technique  and the equal power scheme yield almost the same ${\rm NMSE}$, because there is not enough power for channel estimation. For medium $E_T$, the water-filling method better utilizes the training energy. Instead of allocating an equal amount of power to all eigen-directions, the water-filling scheme allocates more power to the significant eigen-directions, and this performs better within fewer training time slots. At high $E_T$, both schemes perform the same, as they allocate the same amount of power to each eigen-direction. The hybrid beamforming structure performs so close to its fully-digital counterpart. This indicates the efficacy of hybrid design provided in sections \ref{HB_Precoder_design} and \ref{HB_Combiner_design}.

It is worth noting that, while the number of training time slots is fixed in the equal power scheme, the water-filling scheme utilizes fewer training time slots. Since the beamformers are aligned to the channel-eigen-directions, the channel is estimated along its eigen-directions at each time slot. At low to medium $E_T$, instead of estimating the channel components in all eigen-directions, only the components in stronger directions are estimated. The rest of the eigen-directions are not important from an MMSE point of view, hence not estimated. That is, the water-filling scheme requires fewer training time slots. Spectral efficiency-wise, we can save more time slots for data transmission.

\subsection{NMSE and Spectral Efficiency vs $Q$}
To further assess the performance of the proposed channel estimation method, we plot the ${\rm NMSE}$ and the corresponding spectral efficiency vs $Q$ in Figs. \ref{MSE_vs_Qt} and \ref{Rate_vs_Qt}. For each simulation point, we consider $E_T/\sigma_{\rm n}^2 = 64$, which has to be distributed among $Q$ time slots. This implies that, the equal power scheme allocates an equal amount of energy to all $64$ time slots.

The ${\rm NMSE}$ plot in Fig. \ref{MSE_vs_Qt} shows an improvement in channel estimation as we increase $Q$. The reason is that, with a higher $Q$, we estimate more content of the channel. However, due to the correlation among the channel entries, the rate of improvement in channel estimation declines for large $Q$. Note that, given $Q$, our channel estimation scheme aims to estimate the significant portion of the channel within $Q$ training time slots. That is, we observe a better rate of improvement for small $Q$ compared to large $Q$. As expected, the equal power scheme consumes all available $Q$ with equal amounts of power in each time slot, while the water-filling scheme employs much fewer training time slots.

In terms of the spectral efficiency,  Fig. \ref{Rate_vs_Qt} shows that, in equal power scheme, there is no further significant improvement in spectral efficiency as we increase $Q $ beyond 30 time slots. It can be seen that estimating more components of the channel matrix for large values of $Q$ does not significantly affect the spectral efficiency. Indeed, with the equal power scheme, resources are being allocated to estimate the unnecessary and insignificant components of the channel. Instead, from the MMSE point of view, the water-filling scheme allocates more resources to estimate the significant portion of the channel, which has more effect on the spectral efficiency.

\begin{figure*}
    \centering
    \begin{minipage}[t]{.47\textwidth}
      \centering{
\psfrag{MSE}{NMSE (dB)}
\psfrag{rho}{$|\rho|$}
\psfrag{Equal power, FD}{ \small{Equal Power, FD}}
\psfrag{Equal power, NRF=8, HB}{\small{ Equal Power, $N^{\rm RF}=8$, HB, {PE-AltMin \cite{7397861}}}}
\psfrag{Water-filling, FD}{\small{ Water-filling, FD}}
\psfrag{Water-filling, NRF=8, HB}{\small{ Water-filling, $N^{\rm RF}=8$, HB, {PE-AltMin \cite{7397861}}}}
\psfrag{Water-filling, NRF=8, MOAltmin, HB, 123456789}{\small{ Water-filling, $N^{\rm RF}=8$, HB, MO-AltMin \cite{7397861}}}
\resizebox{!}{6.6cm}{\includegraphics{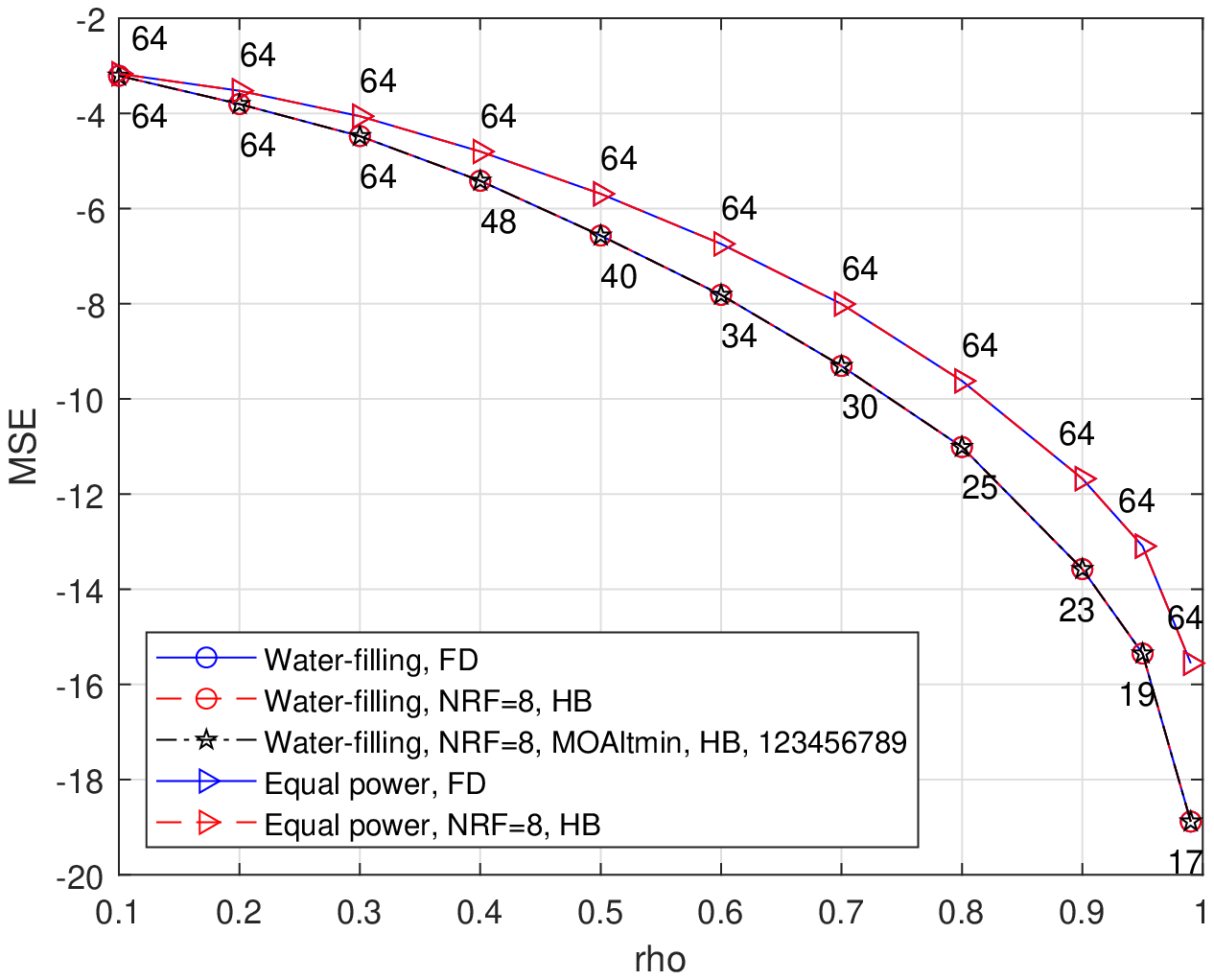}}
\caption{NMSE vs $|\rho|$, $M=16$, $N=32$, $N_{\rm s}=4$, $Q=64$, $E_T/\sigma_{\rm n}^2 = 64$ .}
\label{MSE_vs_rho}}
    \end{minipage}%
    \hfill
    \begin{minipage}[t]{0.47\textwidth}
       \centering{
\psfrag{SE}[c][c] { SE (bits/s/Hz)}
\psfrag{rho}{ $|\rho|$}
\psfrag{Perfect CSI, FD}{\small{Perfect CSI, FD}}
\psfrag{Perfect CSI, NRF=8, HB}{\small{Perfect CSI, $N^{\rm RF} = 8$, HB, {PE-AltMin \cite{7397861}}}}
\psfrag{Equal power, FD}{\small{Equal Power, FD}}
\psfrag{Equal power, NRF=8, HB}{\small{Equal Power, $N^{\rm RF}=8$, HB, {PE-AltMin \cite{7397861}}}}
\psfrag{Water-filling, FD}{\small{Water-filling, FD}}
\psfrag{Water-filling, NRF=8, HB}{\small{Water-filling, $N^{\rm RF}=8$, HB, {PE-AltMin \cite{7397861}}}}
\psfrag{Water-filling, NRF=8, MOAltmin, HB, 123456789}{\small{Water-filling, $N^{\rm RF}=8$, HB, MO-AltMin \cite{7397861}}}
\resizebox{!}{6.6cm}{\includegraphics{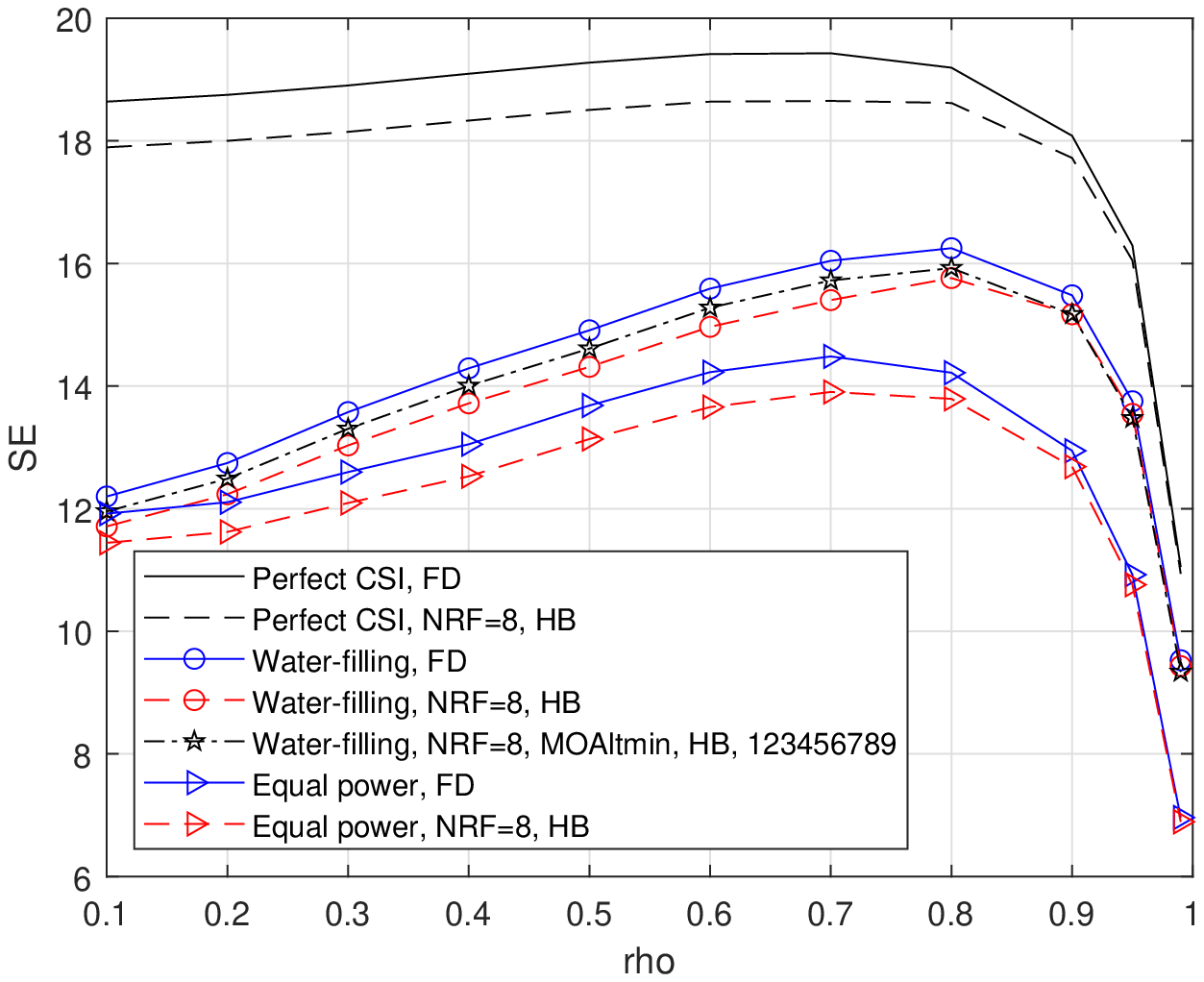}}
\caption{Spectral efficiency vs $|\rho|$, $M=16$, $N=32$, $N_{\rm s}=4$, $Q=64$, $E_T/\sigma_{\rm n}^2 = 64$.}
\label{Rate_vs_rho}}
    \end{minipage}
\end{figure*}

\begin{figure*}
    \centering
    \begin{minipage}[t]{.47\textwidth}
        \centering{
\psfrag{MSE (dB)}[c][c]{NMSE (dB)}
\psfrag{Total transmit power per stream(dB)}[c][c]{$E_T/\sigma_{\rm n}^2$ (dB)}
\psfrag{Equal power, FD}{\small{Equal Power, FD}}
\psfrag{Equal power, NRF=16, HB}{\small{Equal Power, $N^{\rm RF}=16$, HB, {PE-AltMin \cite{7397861}}}}
\psfrag{Water-filling, FD}{\small{Water-filling, FD}}
\psfrag{Water-filling, NRF=16, HB, 123456789123456789}{\small{Water-filling, $N^{\rm RF}=16$, HB,  {PE-AltMin \cite{7397861}}}}
\resizebox{!}{6.6cm}{\includegraphics{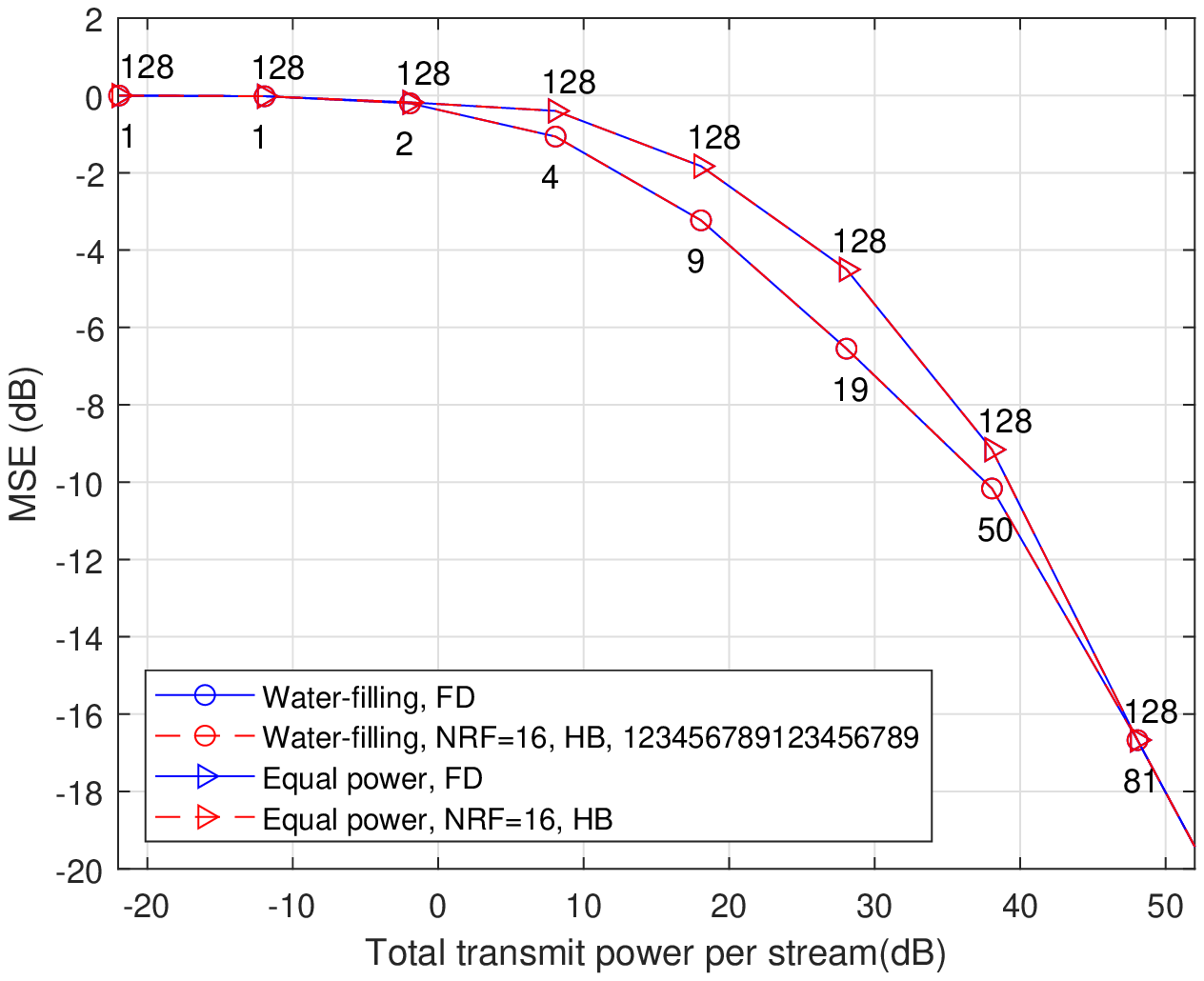}}
\caption{NMSE vs $E_T/\sigma_{\rm n}^2$, $M=32$, $N=64$, $N_{\rm s}=8$, $|\rho|=.8$, $Q=128$.}\label{MSE_vs_PM32N64}
}
\end{minipage}%
\hfill
\begin{minipage}[t]{0.47\textwidth}
       \centering{
\psfrag{Spectral efficiency }[c][c]{SE (bits/s/Hz)}
\psfrag{Total transmit power (dB)}[c][c]{$E_T/\sigma_{\rm n}^2$ (dB)}
\psfrag{Perfect CSI, FD}{\small{Perfect CSI, FD}}
\psfrag{Perfect CSI, NRF=16, HB}{\small{Perfect CSI, $N^{\rm RF}=16$, HB, {PE-AltMin \cite{7397861}}}}
\psfrag{Equal power, FD}{\small{Equal Power, FD}}
\psfrag{Equal power, NRF=16, HB}{\small{Equal Power, $N^{\rm RF}=16$, HB, {PE-AltMin \cite{7397861}}}}
\psfrag{Water-filling, FD}{\small{Water-filling, FD}}
\psfrag{Water-filling, NRF=16, HB, 123456789123456789}{\small{Water-filling, $N^{\rm RF}=16$, HB, {PE-AltMin \cite{7397861}}}}
\resizebox{!}{6.6cm}{\includegraphics{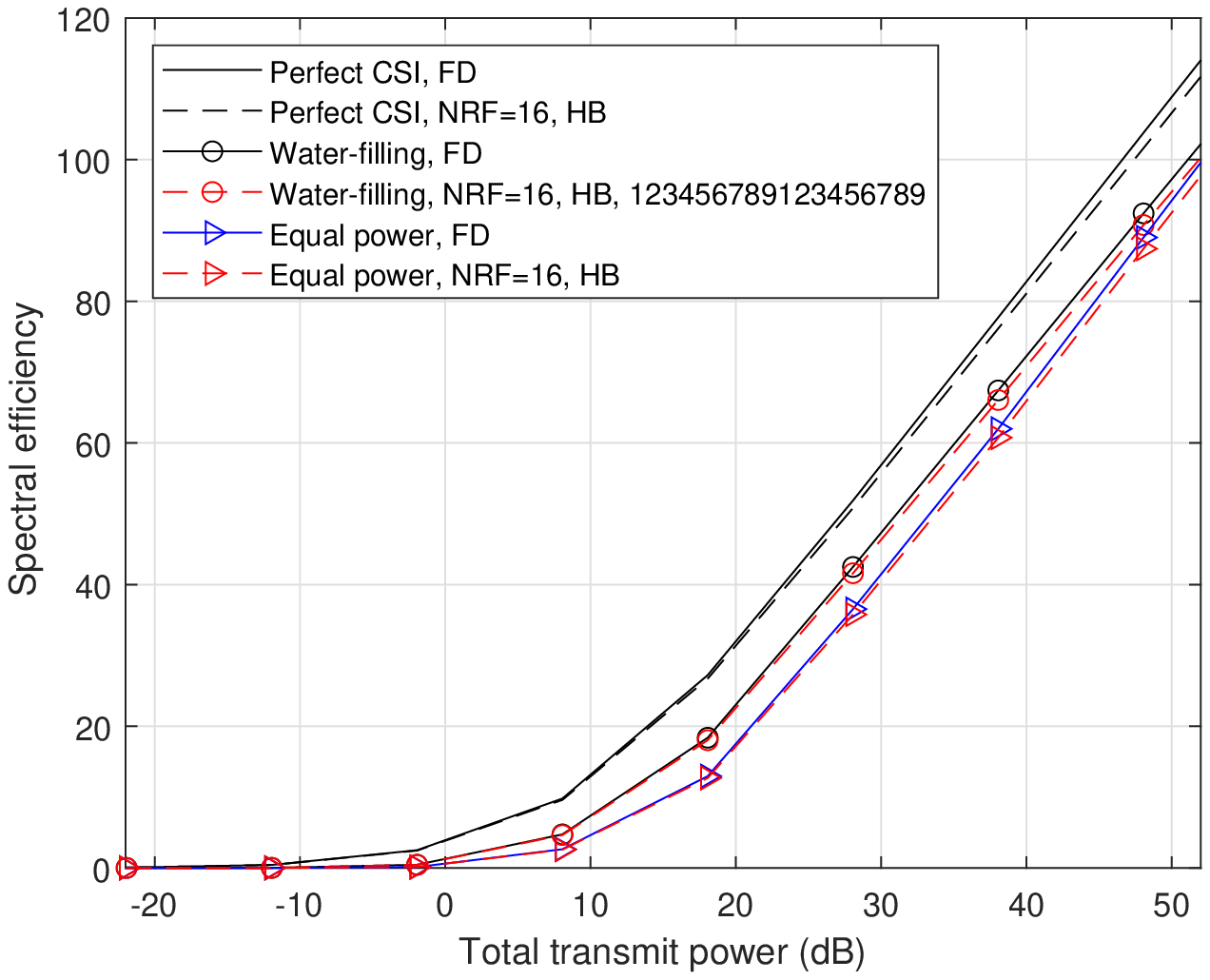}}
\caption{Spectral efficiency vs $E_T/\sigma_{\rm n}^2$, $M=32$, $N=64$, $N_{\rm s}=8$, $|\rho|=.8$, $Q=128$.}
\label{Rate_vs_PM32N64}}
    \end{minipage}
\end{figure*}

\subsection{NMSE and Spectral Efficiency vs $|\rho|$}\label{Sim_secD}
${\rm NMSE}$ and spectral efficiency vs $|\rho|$ are depicted in Figs. \ref{MSE_vs_rho} and \ref{Rate_vs_rho}. In this set of numerical examples, we assume $Q=MN/N^{\rm RF}$. We also consider $E_T/\sigma_{\rm n}^2 = 64$ for training to be distributed among the $Q$ training time slots. As we increase $|\rho|$, the ${\rm NMSE}$ decreases for both methods. The reason is that for high $|\rho|$, the channel matrix entries become more correlated. By exploiting such a higher correlation, the virtual channel in the eigen-domain can be represented with fewer parameters. Therefore, for fixed energy used for channel training, we have better estimation in both schemes. At relatively lower $|\rho|$, the water-filling scheme yields a better estimate of the channel compared to the equal power scheme. This is mainly due to the fact that, at low $|\rho|$, all eigen-directions have almost the same significance from the MMSE point of view. Therefore, both scenarios allocate an equal amount of power in all eigen-directions. However, as channel entries become more correlated, different eigen-directions have different levels of significance. Unlike the equal power allocation scheme, the water-filling scheme allocates more energy to stronger eigen-directions, resulting in a better estimate of the channel.

The spectral efficiency vs $|\rho|$ is depicted in Fig. \ref{Rate_vs_rho}. For fixed $E_T$, we observe that, for very small $|\rho|$, both water-filling and equal power schemes perform almost the same, because they have quite similar NMSE. However, as we increase $|\rho|$, there is a significant gap between these two schemes. This can be explained by the fact that at low $|\rho|$, all eigen-directions have similar significance as opposed to the case with high $|\rho|$, where the channels in certain directions are stronger than the others. Therefore, exploiting the correlation among the channel entries, not all eigen-directions receive same amount of energy. Instead, significant eigen-directions receive more energy. At very high $|\rho|$, although we obtain a better NMSE, the spectral efficiency decreases because, at this range of $|\rho|$, the channel tends to be rank-one, which has lower multiplexing gain.

\begin{figure*}
    \centering
    \begin{minipage}[t]{.47\textwidth}
        \centering{
\psfrag{MSE (dB)}[c][c]{NMSE (dB)}
\psfrag{Total transmit power per stream(dB)}[c][c]{$|\rho|$}
\psfrag{Equal power, FD}{\small{Equal Power, FD}}
\psfrag{Equal power, NRF=16, HB}{\small{Equal Power, $N^{\rm RF}=16$, HB, {PE-AltMin \cite{7397861}}}}
\psfrag{Water-filling, FD}{\small{Water-filling, FD}}
\psfrag{Water-filling, NRF=16, HB, 123456789123456789}{\small{Water-filling, $N^{\rm RF}=16$, HB, {PE-AltMin \cite{7397861}}}}
\resizebox{!}{6.6cm}{\includegraphics{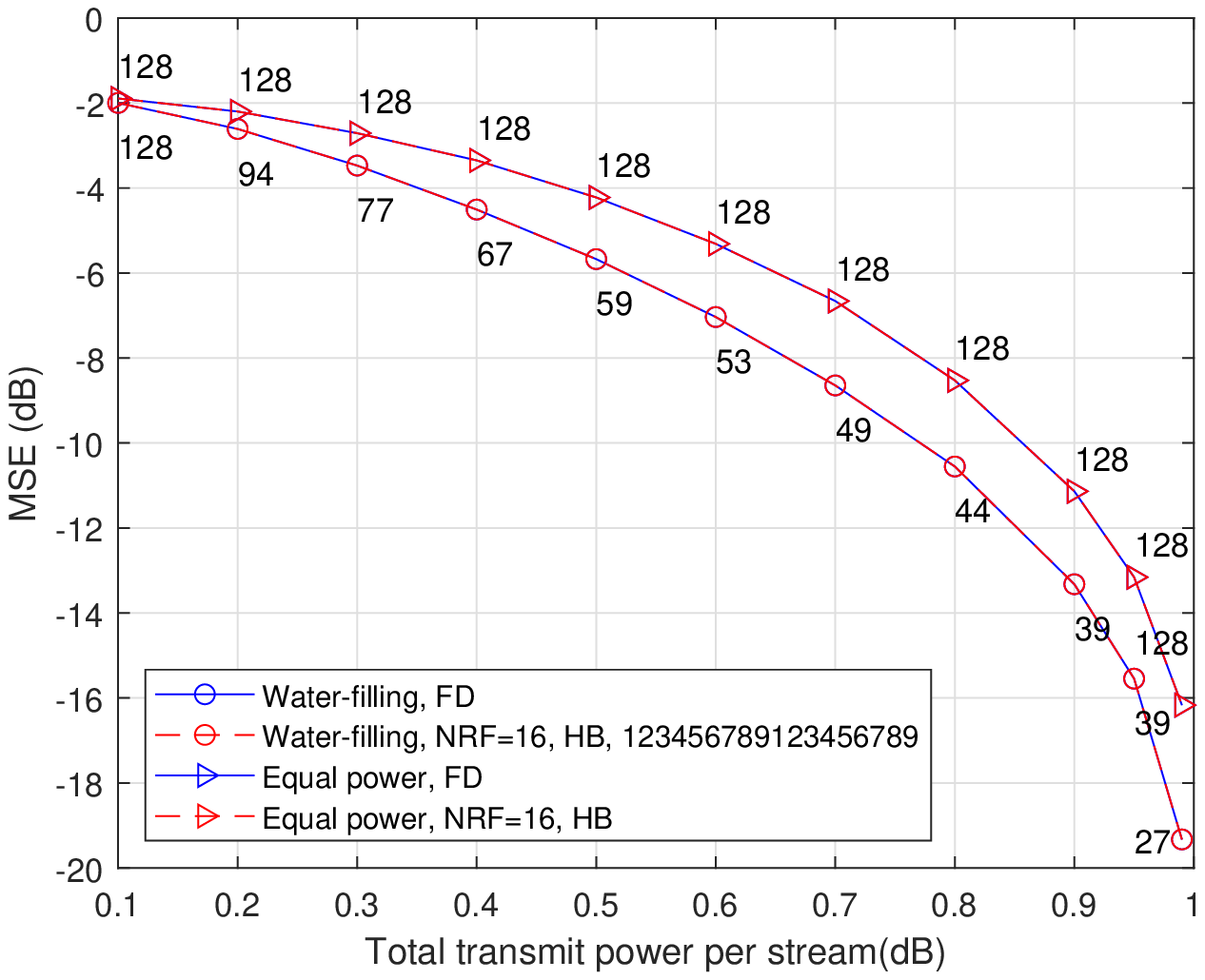}}
\caption{NMSE vs $|\rho|$, $M=32$, $N=64$, $N_{\rm s}=8$, $Q=128$, $E_T/\sigma_{\rm n}^2\thickapprox 38 {\rm (dB)}$.}
\label{MSE_vs_rhoM32N64}
}
\end{minipage}%
    \hfill
    \begin{minipage}[t]{0.47\textwidth}
       \centering{
\psfrag{SE}[c][c]{SE (bits/s/Hz)}
\psfrag{rho}[c][c]{$|\rho|$}
\psfrag{Perfect CSI, FD}{\small{Perfect CSI, FD}}
\psfrag{Perfect CSI, Nrf=16, HB}{\small{Perfect CSI, $N^{\rm RF}=16$, HB, {PE-AltMin \cite{7397861}}}}
\psfrag{Equal power, FD}{\small{Equal Power, FD}}
\psfrag{Equal power, NRF=16, HB}{\small{Equal Power, $N^{\rm RF}=16$, HB, {PE-AltMin \cite{7397861}}}}
\psfrag{Water-filling, FD}{\small{Water-filling, FD}}
\psfrag{Water-filling, NRF=16, HB, 123456789123456789}{\small{Water-filling, $N^{\rm RF}=16$, HB, {PE-AltMin \cite{7397861}}}}
\resizebox{!}{6.6cm}{\includegraphics{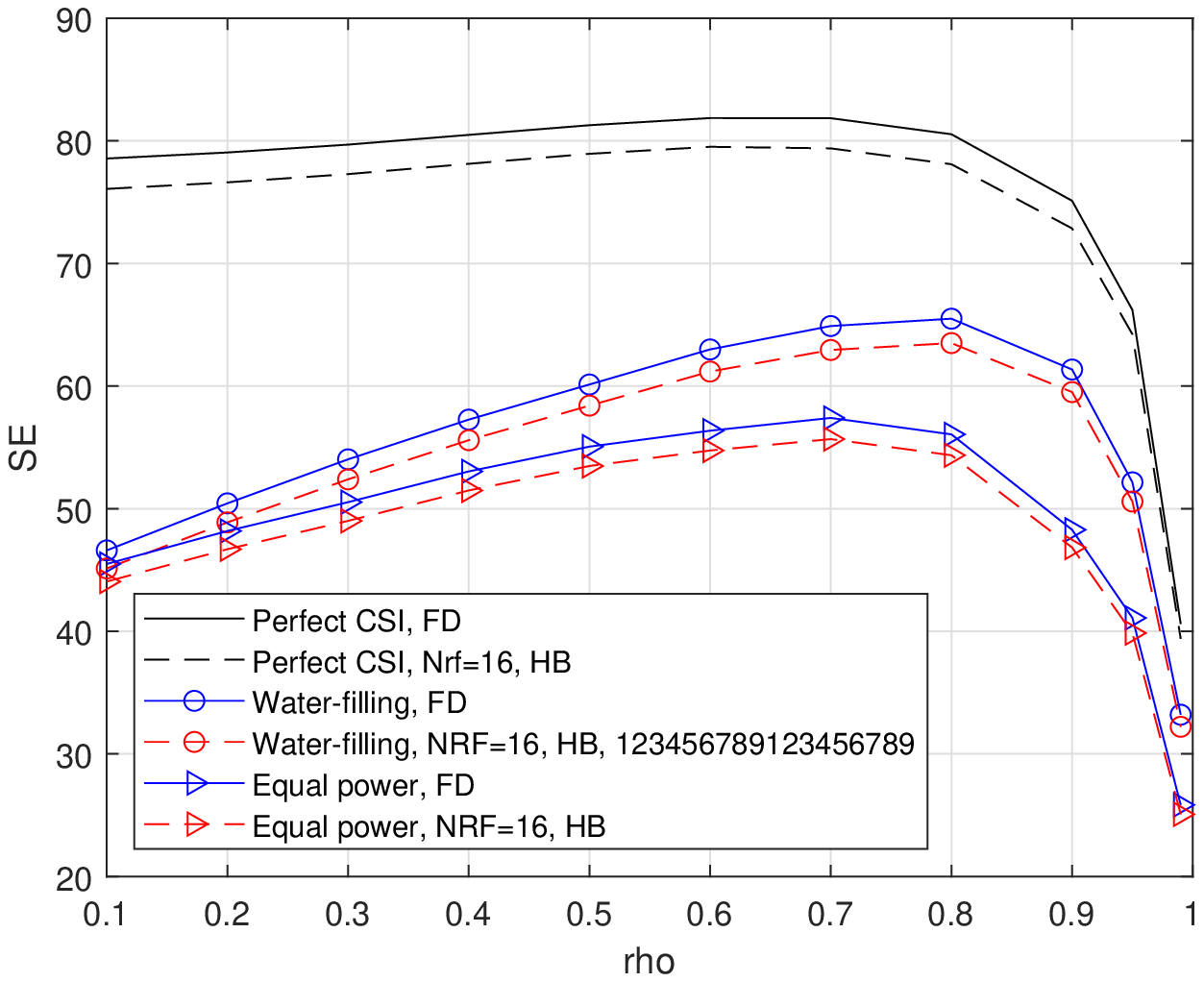}}
\caption{Spectral efficiency vs $|\rho|$, $M=32$, $N=64$, $N_{\rm s}=8$, $Q=128$, $E_T/\sigma_{\rm n}^2\thickapprox 38 {\rm (dB)}$.}
\label{Rate_vs_rhoM32N64}}
    \end{minipage}
\end{figure*}

\begin{figure*}
    \centering
    \begin{minipage}{0.47\textwidth}
        \centering{
\psfrag{MSE (dB)}[c][c]{NMSE (dB)}
\psfrag{Total transmit power per stream(dB)}[c][c]{$E_T/\sigma_{\rm n}^2$ (dB)}

\psfrag{Equal power, Ns=4, FD}{\small{Equal Power, $N^{\rm RF}=4$, $N_s=4$, FD}}
\psfrag{Equal power, NRF=4, Ns=4, FD}{\small{Equal Power, $N^{\rm RF}=4$, $N_s=4$, HB}}
\psfrag{Water-filling, Ns=4, FD}{\small{Water-filling, $N^{\rm RF}=4$, $N_s=4$, FD}}
\psfrag{Water-filling, NRF=4, Ns=4, HB, 12345}{\small{Water-filling, $N^{\rm RF}=4$, $N_s=4$, HB}}

\psfrag{Equal power, Ns=2, FD}{\small{Equal Power, $N^{\rm RF}=8$, $N_s=2$, FD}}
\psfrag{Equal power, NRF=8, Ns=2, FD}{\small{Equal Power, $N^{\rm RF}=8$, $N_s=2$, HB}}
\psfrag{Water-filling, Ns=2, FD}{\small{Water-filling, $N^{\rm RF}=8$, $N_s=2$, FD}}
\psfrag{Water-filling, NRF=8, Ns=4, HB}{\small{Water-filling, $N^{\rm RF}=8$, $N_s=2$, HB}}
\resizebox{!}{6.6cm}{\includegraphics{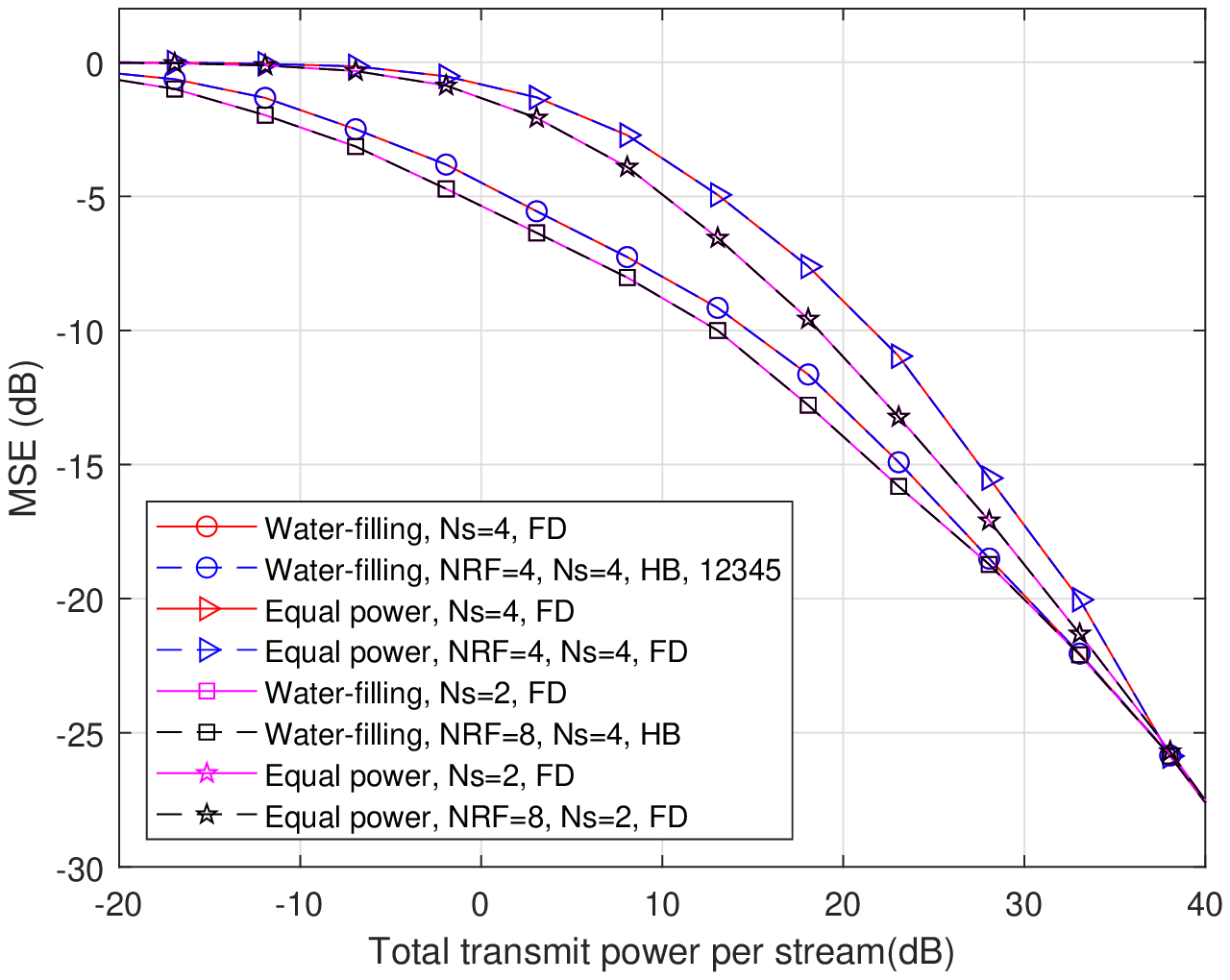}}
\caption{NMSE vs $E_T/\sigma_{\rm n}^2$, $M=16$, $N=32$, $|\rho|=.8$.}
\label{MSE_vs_PNs4Nrf4}
}
    \end{minipage}%
   \hfill
    \begin{minipage}{0.47\textwidth}
       \centering{
\psfrag{Spectral efficiency }[c][c]{SE (bits/s/Hz)}
\psfrag{Total transmit power (dB)}[c][c]{$E_T/\sigma_{\rm n}^2$ (dB)}
\psfrag{Perfect CSI, FD}{\small{Perfect CSI, FD}}
\psfrag{Equal power, Ns=4, FD}{\small{Equal Power, $N^{\rm RF}=4$, $N_s=4$, FD}}
\psfrag{Equal power, NRF=4, Ns=4, FD}{\small{Equal Power, $N^{\rm RF}=4$, $N_s=4$, HB}}
\psfrag{Water-filling, Ns=4, FD}{\small{Water-filling, $N^{\rm RF}=4$, $N_s=4$, FD}}
\psfrag{Water-filling, NRF=4, Ns=4, HB, 12345}{\small{Water-filling, $N^{\rm RF}=4$, $N_s=4$, HB}}

\psfrag{Equal power, Ns=2, FD}{\small{Equal Power, $N^{\rm RF}=8$, $N_s=2$, FD}}
\psfrag{Equal power, NRF=8, Ns=2, FD}{\small{Equal Power, $N^{\rm RF}=8$, $N_s=2$, HB}}
\psfrag{Water-filling, Ns=2, FD}{\small{Water-filling, $N^{\rm RF}=8$, $N_s=2$, FD}}
\psfrag{Water-filling, NRF=8, Ns=4, HB}{\small{Water-filling, $N^{\rm RF}=8$, $N_s=2$, HB}}
\resizebox{!}{6.6cm}{\includegraphics{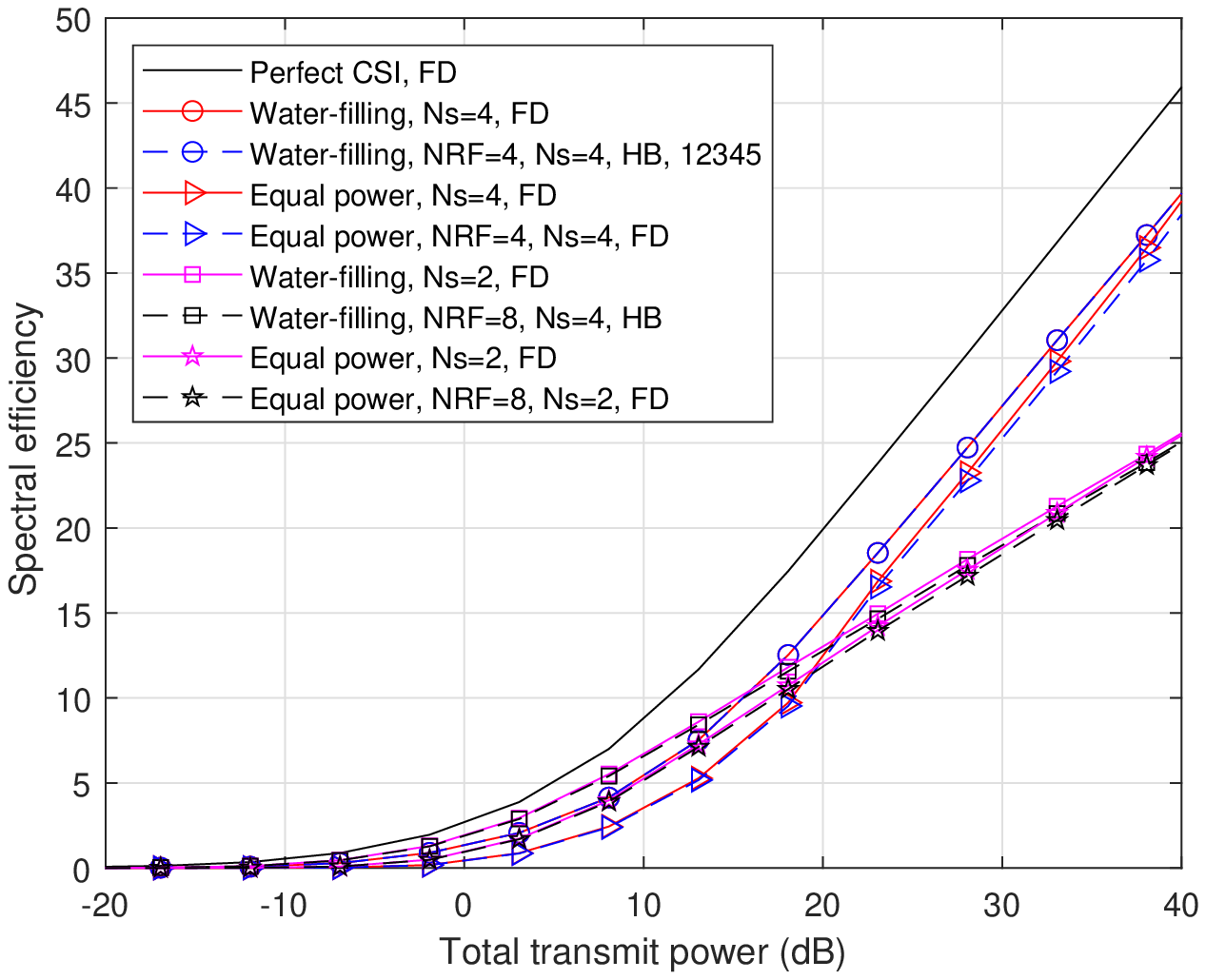}}
\caption{Spectral efficiency vs $E_T/\sigma_{\rm n}^2$, $M=16$, $N=32$, $|\rho|=.8$.}
\label{Rate_vs_PNs4Nrf4}}
    \end{minipage}
\end{figure*}

\vspace{-3mm}
\subsection{Performance for Different $(M, N)$}
We now provide further simulations for $M=32$, $N=64$, $N^{\rm RF}=16$, and $N_s=8$. The NMSE and spectral efficiency vs. $E_T/\sigma_{\rm n}^2$ are plotted in Figs. \ref{MSE_vs_PM32N64} and \ref{Rate_vs_PM32N64}, respectively. Similar to Figs. \ref{MSE_vs_P} and \ref{Rate_vs_P}, the proposed water-filling technique outperforms the equal power allocation in  practical ranges of $E_T/\sigma_{\rm n}^2$. At low $E_T/\sigma_{\rm n}^2$,  both techniques perform the same. At very large $E_T/\sigma_{\rm n}^2$, the proposed water-filling technique allocates almost equal amount of power to different eigen-directions, thereby yielding the same performance as the equal power technique does.
The NMSE and spectral efficiency vs $|\rho|$ are plotted in Figs. \ref{MSE_vs_rhoM32N64}, and \ref{Rate_vs_rhoM32N64}, respectively. For fixed $E_T$, and for relatively small $|\rho|$, both the water-filling method and the equal power technique perform almost the same. This is mainly because, at small values of $|\rho|$, the channel entries are largely uncorrelated and the channel components have almost the same levels of significance. However, as we increase $|\rho|$, the channel components have different level of significance. Leveraging this characteristic, the water-filling technique yields a better channel estimate within a smaller $Q$.
\vspace{-3mm}

\subsection{Performance for Different $(N^{\rm RF}, N_{\rm s})$}
\label{figs_Ns4Nrf4}
Figs.~\ref{MSE_vs_PNs4Nrf4} and \ref{Rate_vs_PNs4Nrf4} illustrates the NMSE and spectral efficiency curves for different pairs  of $(N^{\rm RF}, N_{\rm s})$. In terms of NMSE, the scenario with $(N^{\rm RF}, N_{\rm s})= (8,2) $ performs the same as the one with  $(N^{\rm RF}, N_{\rm s})= (8,4)$ provided in Fig. \ref{MSE_vs_P}. The reason is that in these two scenarios, we employed the same number of the RF chains. However, from a spectral efficiency perspective, the scenario with  $(N^{\rm RF}, N_{\rm s})=(8,4)$ yields a better performance compared to the one where $(N^{\rm RF}, N_{\rm s})=(8,2)$. This is mainly attributed to the multiplexing gain caused by the difference in $N_{\rm s}$.
Compared to the scenario with $(N^{\rm RF}, N_{\rm s})=(4,4)$, the scenario with $(N^{\rm RF}, N_{\rm s})=(8,4)$ performs better in terms of  NMSE mainly because more RF chains are employed. In Fig.~\ref{Rate_vs_PNs4Nrf4}, the scenario with $(N^{\rm RF}, N_{\rm s})=(8,4)$ offers a better spectral efficiency for the same $N_{\rm s}$.
The reason is that the scenario with $(N^{\rm RF}, N_{\rm s})=(8,4)$ provides a better channel estimate.

\vspace{-3mm}
{\subsection{$Q_{\rm nz}$ vs the Number of Antennas}\label{subsection_G}
To show the dependency of $Q_{\rm nz}$
on the number of antenna, in Fig.~\ref{Q_rho_test_intext}, we plot the $Q_{\rm nz}/Q$  vs. $M$ for $N^{\rm RF} = 8$, $N=16$, and $E_T/\sigma_{\rm n}^2 = 64$. It can be seen that for $|\rho| =0.1$ and for small $M$, all $Q$ training time slots are used, i.e., $Q_{\rm nz}= Q$. This is mainly due to the fact that, for this range of $|\rho|$, the channel components along all eigen-directions have \emph{almost} the same significance, and therefore, receive almost the same amount of energy during the training. As we increase $M$, and for the same training energy budget, we observe that not all training time slots are being utilized for training. Instead, the energy is being allocated to estimate the channel components along the stronger eigen-directions. This observation becomes more evident when we increase $|\rho|$ (i.e., channel correlation is increased), where the training energy is dedicated to only estimate the channel components along the stronger eigen-directions.}

\begin{figure}
\begin{minipage}[t]{.47\textwidth}
\centering{
\psfrag{M}[c][c] { $M$}
\psfrag{Q}[c][c]{$Q_{\rm nz}/Q$ $(\%)$ }
\psfrag{rho=1,123456789012345}{$|\rho|=0.1$}
\psfrag{rho=5}{$|\rho|=0.5$}
\psfrag{rho=8}{$|\rho|=0.8$}
\resizebox{!}{6.6cm}{\includegraphics{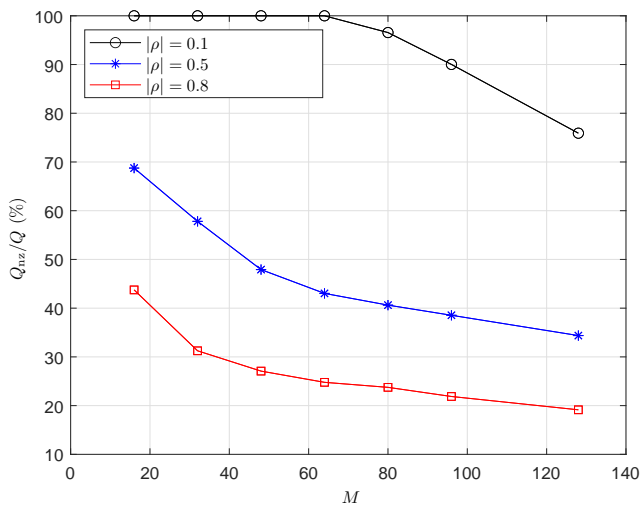}}
\caption{{$Q_{\rm nz}/Q$ vs $M$, for $Q = MN/N^{\rm RF}$, $N^{\rm RF} = 8$, $N=16$, $E_T/\sigma_{\rm n}^2 = 64$. }}
\label{Q_rho_test_intext}}
\end{minipage}
\end{figure}

\label{remark3} {\bf Remark 3:} The overall computational complexity of the proposed technique with respect to the number of training time slots is $\mathcal{O}(Q)$. This implies that adding more training time slots will linearly increase the complexity of the proposed technique. The improvement in performance, however, is not linear. Specifically, as shown in Figs. \ref{MSE_vs_Qt} and  \ref{Rate_vs_Qt}, the improvement in NMSE and spectral efficiency  is not linear. In fact, for $ Q \geq 32$, the improvement becomes marginal. Such a marginal improvement in performance is mainly due to the fact that the channel entries are correlated, and therefore, spending more time resources for channel estimation does not provide a significant improvement in channel estimation or spectral efficiency. This becomes further evident when  $|\rho|$ is very close to $1$, indicating that the channel can be represented by a few dominant eigenvectors, and estimating the channel along those eigen-directions only provides good performance.

\section{Conclusions} \label{sec:conclusion}
We proposed a training-based channel estimation technique for  correlated massive MIMO systems where hybrid beamforming is employed. Using the Kronecker model to express such correlated channels, and relying on the fact that the channel entries are uncorrelated in its eigen-domain, we estimated the channel components in this domain. Since the number of RF chains is much less than the number of transmitter/receiver antennas, the training has to be done in multiple time slots. Minimizing the channel estimate mean squared error criterion under a total training energy budget, we optimally designed the hybrid beamformers in each time slot that led to linearly estimate the channel. The optimal precoder and combiner in each training time slot are aligned to the set of strongest eigen-directions of the channel. The number of training time slots is determined by the training energy budget. Our simulation results indicate that, when the training energy budget is relatively low,  only the portion of the channel along the strongest eigen-directions is estimated. At high training energy budgets, all eigen-directions participate in channel estimation. Importantly, by exploiting the knowledge of the channel correlations, we can reduce the number of time slots required for channel training.

As a future direction to this work, one can consider   different architectures
such as partially-connected analog beamformers or  analog beamformers, which are implemented using only
analog switches. Another direction for future work is considering  more practical settings where the effects of finite-resolution  analog-to-digital converters (ADCs) and digital-to-analog converters (DACs) are taken into account.

\appendices

\section{Simplifying $\bGamma^2(\mathcal{V},  \mathcal{{W}})$} \label{Simp_gamma2}
Since $\bR_n(\mathcal{W})=\sigma_{\rm n}^2{\rm blkdg}\left[\left(\bW_1^H\bW_1\right)\;\;\cdots\;\; \left(\bW_{Q}^H\bW_{Q}\right)\right]$, we can write
\begin{align} \label{Rn_inv}
\bR_n^{-1}(\mathcal{W})&=  \frac{1}{\sigma_{\rm n}^2} {\rm blkdg}\left[\left(\bZ_1\bD_1^2\bZ_1^H\right)^{-1}\;\;\cdots\;\; \left(\bZ_{Q}\bD_{Q}^2\bZ_{Q}^H\right)^{-1}\right] \nonumber\\
&=  \frac{1}{\sigma_{\rm n}^2} {\rm blkdg}\left[\bZ_1\bD_1^{-2}\bZ_1^H\;\;\cdots\;\; \bZ_{Q}\bD_{Q}^{-2}\bZ_{Q}^H\right]\nonumber\\
&= \frac{1}{\sigma_{\rm n}^2} {\rm blkdg}\left[\bZ_1\;\;\cdots\;\; \bZ_{Q}\right]
 {\rm blkdg}\left[\bD_1^{-2}\;\;\cdots\;\; \bD_{Q}^{-2}\right] \nonumber\\
 &\qquad \times {\rm blkdg}\left[\bZ_1^H\;\;\cdots\;\; \bZ_{Q}^H\right].
\end{align}
Then, $\bGamma^2 (\mathcal{V},  \mathcal{W})$ in \eqref{gamma2_def} is written, in terms of $\tilde{\bZ}$ and $\tilde{\bD,}$ as:
\begin{align} \label{gamma_diag}
 \bGamma^2(\mathcal{V}, \mathcal{W})= \frac{1}{\sigma_{\rm n}^2}\bPsi^H\bPhi^H(\mathcal{V},  \mathcal{W})\tilde{\bZ}\tilde{\bD}\tilde{\bZ}^H\bPhi(\mathcal{V},  \mathcal{W})\bPsi.
\end{align}
To further simplify $\bGamma^2(\mathcal{V},  \mathcal{W})$, defining $\tilde \bv_q \triangleq \bv_q /\| \bv_q\|$ and $z_q = \|\bv_q\|$, we can write:
\begin{align} \label{z_phi_psi_diag}
 &\tilde{\bZ}^H\bPhi(\mathcal{V},  \mathcal{W})\bPsi = {\rm blkdg}\left[\bZ_1^{H}\;\;\cdots\;\; \bZ_{Q}^H\right]
\left(
\begin{array}{c}
\bv_1^T \otimes\bW_1^H\\
\bv_2^T \otimes\bW_2^H \\
\vdots  \\
\bv_{Q}^T \otimes\bW_{Q}^H
\end{array}
\right) \bPsi\nonumber\\
&=\resizebox{.95\hsize}{!}{$\left(
  \begin{array}{cccc}
   \|\bv_1\|\tilde{\bv}_1^T\otimes \bZ_1^H\bW_1^H  \\
    \|\bv_2\| \tilde{\bv}_2^T\otimes \bZ_1^H\bW_2^H \\
      \vdots  \\
     \|\bv_{Q}\| \tilde{\bv}_{Q}^T\otimes \bZ_{Q}^H\bW_{Q}^H
  \end{array}
\right)\bPsi=\left(
  \begin{array}{cccc}
   \|\bv_1\|\tilde{\bv}_1^T\otimes \bD_1^H\bK_1^H  \\
     \|\bv_2\|\tilde{\bv}_2^T\otimes \bD_2^H\bK_2^H  \\
      \vdots  \\
      \|\bv_{Q}\|\tilde{\bv}_{Q}^T\otimes \bD_{Q}^H\bK_{Q}^H
  \end{array}
\right) \bPsi$} \nonumber\\
&= \resizebox{.83\hsize}{!}{${\rm blkdg}\left[\|\bv_1\|\bD_1^H\;\;\cdots\;\; \|\bv_Q\| \bD_{Q}^H\right]
 \underbrace{ \left(
  \begin{array}{cccc}
  \tilde{\bv}_1^T\otimes \bK_1^H  \\
     \tilde{\bv}_2^T\otimes \bK_2^H  \\
      \vdots  \\
      \tilde{\bv}_{Q}^T\otimes \bK_{Q}^H
  \end{array}
\right)}_{\triangleq \bUpsilon^H \left(\mathcal{\tilde{V}}, \mathcal{{K}}\right)}\bPsi$},
\end{align}
where we define $ \mathcal{\tilde V}\triangleq \left(\tilde \bv_1, \tilde \bv_2, \ldots, \tilde \bv_Q\right)$ and  $\mathcal{{K}}= \left(\bK_1, \bK_2, \cdots, \bK_Q\right)$. Using \eqref{z_phi_psi_diag}, we can write \eqref{gamma_diag} as:
\begin{align} \label{gamma_long_mtx}
 \bGamma^2(\mathcal{V},  \mathcal{{K}})&= \frac{1}{\sigma_{\rm n}^2} \bPsi^H \bUpsilon \left(\mathcal{\tilde{V}}, \mathcal{{K}}\right)\mathbfcal{D}(\mathcal{Z})  \bUpsilon^H \left(\mathcal{\tilde{V}}, \mathcal{{K}}\right) \bPsi,
\end{align}
where we define: $\mathbfcal{D}(\mathcal{Z})\triangleq {\rm blkdg}\left[
  z_{1}^{2} \bI_{N^{\rm RF}} ,\;\ldots\;, z_Q^2 \bI_{N^{\rm RF}}\right ]$ and $\mathcal{Z}\triangleq (z_1, z_2,\ldots, z_Q)$.

\section{ Relation between $\tilde{\bGamma}^2\left(\mathcal{Z}^{\rm o},\mathcal{\tilde{V}}^{\rm o}, \mathcal{K}^{\rm o} \right)$ and $\mathbfcal{D}\left(\mathcal{Z}^{\rm o}\right)$ } \label{gama_prop}
Before we begin, we introduce the following lemma that will help us to show the relation between $\tilde{\bGamma}^2\left(\mathcal{Z}^{\rm o},\mathcal{\tilde{V}}^{\rm o}, \mathcal{K}^{\rm o} \right)$ and $\mathbfcal{D}\left(\mathcal{Z}^{\rm o}\right)$.

\begin{lemma}\label{trace_lemma}
For any $\mathcal{V}$, $\mathcal{\tilde{V}}$, and $\mathcal{K}$, the following equalities hold
\begin{align} \label{diag_one_eq}
{\rm diag}\left(\bUpsilon^H \left(\mathcal{\tilde{V}}, \mathcal{{K}}\right)\bUpsilon \left(\mathcal{\tilde{V}}, \mathcal{{K}}\right)\right)&={\bf 1}_{MN}
  \\ \label{trace_lemma_eq}
   {\rm Tr}\left(\bGamma^2(\mathcal{Z}, \mathcal{\tilde{V}}, \mathcal{K})\right)&= \frac{1}{\sigma_{\rm n}^2}{\rm Tr}\left(\mathbfcal{D(\mathcal{Z})}\right)
    \end{align}
\end{lemma}
\begin{proof}
To show \eqref{diag_one_eq}, based on \eqref{gamma_long_mtx}, we note that the $q$th block diagonal entries of $\bUpsilon^H \left(\mathcal{\tilde{V}}, \mathcal{{K}}\right)\bUpsilon \left(\mathcal{\tilde{V}}, \mathcal{{K}}\right)$ is given by $\left(\tilde{\bv}_q^T\otimes \bK_q^H\right)\left(\tilde{\bv}_q^T\otimes \bK_q^H\right)^H=\bI_{N^{\rm RF}}$, for $q=1, 2, \ldots, Q$. This now implies that \eqref{diag_one_eq} holds true. To show \eqref{trace_lemma_eq}, from \eqref{gamma_long_mtx}, we can write
\begin{align} \label{trace_lemma_eq_proof}
   &{\rm Tr}\left(\bGamma^2(\mathcal{Z}, \mathcal{\tilde{V}}, \mathcal{K})\right) \\
   &= {\rm Tr}\left(\frac{1}{\sigma_{\rm n}^2} \bPsi^H \bUpsilon \left(\mathcal{\tilde{V}}, \mathcal{{K}}\right) \mathbfcal{D}\left(\mathcal{Z}\right) \bUpsilon^H\left(\mathcal{\tilde{V}}, \mathcal{{K}}\right)  \bPsi\right)\nonumber\\
   &\overset{(a)}{=} \frac{1}{\sigma_{\rm n}^2}{\rm Tr}\left(  \mathbfcal{D}\left(\mathcal{Z}\right) \bUpsilon^H\left(\mathcal{\tilde{V}}, \mathcal{{K}}\right)  \bPsi \bPsi^H \bUpsilon \left(\mathcal{\tilde{V}}, \mathcal{{K}}\right)\right)\nonumber\\
   &\overset{(b)}{=} \frac{1}{\sigma_{\rm n}^2}{\rm Tr}\left(  \mathbfcal{D}\left(\mathcal{Z}\right) \bUpsilon^H\left(\mathcal{\tilde{V}}, \mathcal{{K}}\right)  \bUpsilon \left(\mathcal{\tilde{V}}, \mathcal{{K}}\right)\right)\overset{(c)}{=} \frac{1}{\sigma_{\rm n}^2}{\rm Tr}\left(  \mathbfcal{D}\left(\mathcal{Z}\right) \right),\nonumber
\end{align}
where $(a)$ follows from the fact that ${\rm Tr}(\bA\bB)= {\rm Tr}(\bB\bA)$, $(b)$ follows from $\bPsi \bPsi^H=\bI$, and $(c)$ follows from \eqref{diag_one_eq}.
\end{proof}

Given $\resizebox{.875\hsize}{!}{$\tilde{\bGamma}^2\left(\mathcal{Z}^{\rm o},\mathcal{\tilde{V}}^{\rm o}, \mathcal{K}^{\rm o} \right)= \tilde{ \bPsi}^H \bUpsilon \left(\mathcal{\tilde{V}}^{\rm o}, \mathcal{{K}}^{\rm o}\right) \mathbfcal{D}\left(\mathcal{Z}^{\rm o}\right) \bUpsilon^H \left(\mathcal{\tilde{V}}^{\rm o}, \mathcal{{K}}^{\rm o}\right)  \tilde{\bPsi}$} $ is a diagonal matrix, we aim to show that its non-zero diagonal entries are the same as $\mathbfcal{D}\left(\mathcal{Z}^{\rm o}\right)$. For notation simplicity, let us drop the dependency of each matrix to their variables, and denote $\tilde{\bGamma}^{{\rm o}^2}\triangleq\tilde{\bGamma}^2\left(\mathcal{Z}^{\rm o},\mathcal{\tilde{V}}^{\rm o}, \mathcal{K}^{\rm o} \right)$, $\bUpsilon^{\rm o}\triangleq \bUpsilon \left(\mathcal{\tilde{V}}^{\rm o}, \mathcal{{K}}^{\rm o}\right) $ and $\mathbfcal{D}^{\rm o}\triangleq \mathbfcal{D}\left(\mathcal{Z}^{\rm o}\right)$. Since $\tilde{\bGamma}^{{\rm o}^2}$ is a rank-$L$ diagonal matrix, we can write
\begin{align} \label{gamma_opt_diag}
 \frac{1}{\sigma_{\rm n}^2}\tilde{ \bPsi}^H \bUpsilon^{\rm o}\mathbfcal{D}^{\rm o} \bUpsilon^{{{\rm o}}^H}  \tilde{\bPsi} =\tilde{\bGamma}^{{\rm o}^2} \overset{(a)}{=}\mathbfcal{Q}\underbrace{\left(\begin{array}{cc}
   \tilde{\bGamma}_{\rm nz}^{{\rm o}^2}& \\
   & \b0
  \end{array}
\right)}_{\triangleq \tilde{\bGamma}_{\rm s}^{{\rm o}^2}}\mathbfcal{Q}^H
\end{align}
where $(a)$ follows from the eigenvalue decomposition of $\tilde{\bGamma}^{{\rm o}^2}$, $\mathbfcal{Q}$ is a unitary matrix and $\tilde{\bGamma}_{\rm nz}^{{\rm o}^2}$ is an $L \times L$ diagonal matrix corresponding to the non-zero diagonal entries of $\tilde{\bGamma}^{{\rm o}^2}$. Defining $\check{\tilde{\bPsi}}\triangleq {\tilde{\bPsi}}\mathbfcal{Q}$, then \eqref{gamma_opt_diag} is expressed as:
\begin{align} \label{gamma_opt_diag_re}
 \frac{1}{\sigma_{\rm n}^2}\check{\tilde{\bPsi}}^H \bUpsilon^{\rm o}\mathbfcal{D}^{\rm o} \bUpsilon^{{{\rm o}}^H}  \check{\tilde{\bPsi}} = \tilde{\bGamma}_{\rm s}^{{\rm o}^2}.
\end{align}
Now, let us augment $\bUpsilon^{\rm o}$ and $\mathbfcal{D}^{\rm o}$ and define the following $MN \times MN$ matrices
\begin{align}
\mathbfcal{\tilde{D}}^{\rm o}&\triangleq\left(\begin{array}{c:c}
   \mathbfcal{D}^{\rm o}& \b0\\ \hdashline
 \b0  & \b0
  \end{array}
\right),\,\,\,\,
\tilde{\bUpsilon}^{\rm o}\triangleq\left(\begin{array}{c:c}
{\bUpsilon}^{\rm o}&  \mathbfcal{C}
\end{array}
\right)\label{tilde_upsilon_def}.
\end{align}
where $\mathbfcal{C}$ can be any $MN \times (MN-L)$ matrix, and $\b0$ in
$\mathbfcal{\tilde{D}}^{\rm o}$ is an $(MN-L) \times (MN-L)$ all-zero matrix. Let us choose $\mathbfcal{C}$ such that ${\rm diag}\left(\tilde{\bUpsilon}^{{\rm o}^H}\tilde{\bUpsilon}^{\rm o}\right)={\bf 1}_{MN}$, meaning that the columns of $\mathbfcal{C}$ should have unit norms. One such choice for $\mathbfcal{C}$ that satisfies this condition is to choose its columns from the basis vectors that spans the null space of ${\bUpsilon}^{\rm o}$. Then, it can be verified that $\frac{1}{\sigma_{\rm n}^2}\check{\tilde{\bPsi}}^H \tilde{\bUpsilon}^{\rm o}\tilde{\mathbfcal{D}}^{\rm o} \tilde{\bUpsilon}^{{{\rm o}}^H}  \check{\tilde{\bPsi}}= \tilde{\bGamma}_{\rm s}^{{\rm o}^2}$ holds true. Since $\tilde{\bGamma}_{\rm s}^{{\rm o}^2}$ is the optimal diagonal matrix that minimizes the objective function, along with the fact that $\check{\tilde{\bPsi}}$ is a unitary matrix, and ${\rm diag}\left(\tilde{\bUpsilon}^{{\rm o}^H}\tilde{\bUpsilon}^{\rm o}\right)={\bf 1}_{MN}$ and ${\rm Tr}\left(\tilde{\bGamma}_{\rm s}^{{\rm o}^2}\right)= \frac{1}{\sigma_{\rm n}^2}{\rm Tr}\left(\tilde{\mathbfcal{D}}^{\rm o}\right)$, we use \eqref{gamma_short_mtx_v_tilde_opt} and conclude that $\tilde{\bGamma}_{\rm s}^{{\rm o}^2}= \frac{1}{\sigma_{\rm n}^2}\tilde{\mathbfcal{D}}^{\rm o}$. This implies that $\mathbfcal{D}^{\rm o}={\sigma_{\rm n}^2} \tilde{\bGamma}_{\rm nz}^{{\rm o}^2}$. The proof is complete.

%

\bibliographystyle{IEEEtran}
\bibliography{reference}

\begin{thebibliography}{10}
\providecommand{\url}[1]{#1}
\csname url@samestyle\endcsname
\providecommand{\newblock}{\relax}
\providecommand{\bibinfo}[2]{#2}
\providecommand{\BIBentrySTDinterwordspacing}{\spaceskip=0pt\relax}
\providecommand{\BIBentryALTinterwordstretchfactor}{4}
\providecommand{\BIBentryALTinterwordspacing}{\spaceskip=\fontdimen2\font plus
\BIBentryALTinterwordstretchfactor\fontdimen3\font minus
  \fontdimen4\font\relax}
\providecommand{\BIBforeignlanguage}[2]{{%
\expandafter\ifx\csname l@#1\endcsname\relax
\typeout{** WARNING: IEEEtran.bst: No hyphenation pattern has been}%
\typeout{** loaded for the language `#1'. Using the pattern for}%
\typeout{** the default language instead.}%
\else
\language=\csname l@#1\endcsname
\fi
#2}}
\providecommand{\BIBdecl}{\relax}
\BIBdecl

\bibitem{Javad_ICASSP_2020}
J.~{Mirzaei}, F.~{Sohrabi}, R.~{Adve}, and S.~{ShahbazPanahi}, ``{MMSE}-based
  channel estimation for hybrid beamforming massive {MIMO} with correlated
  channels,'' in \emph{Proc. IEEE Int. Conf. Acoust., Speech, Signal Process.
  (ICASSP), Barcelona, Spain}, May 2020, pp. 5015--5019.

\bibitem{6971234}
D.~{Ciuonzo}, P.~S. {Rossi}, and S.~{Dey}, ``Massive {MIMO} channel-aware
  decision fusion,'' \emph{IEEE Trans. Signal Process.}, vol.~63, no.~3, pp.
  604--619, 2015.

\bibitem{8752284}
B.~M. {Lee} and H.~{Yang}, ``Massive {MIMO} with massive connectivity for
  industrial {Internet} of {Things},'' \emph{IEEE Trans. Ind. Electron.},
  vol.~67, no.~6, pp. 5187--5196, 2020.

\bibitem{7395392}
A.~{Shirazinia}, S.~{Dey}, D.~{Ciuonzo}, and P.~{Salvo Rossi}, ``Massive {MIMO}
  for decentralized estimation of a correlated source,'' \emph{IEEE Trans.
  Signal Process.}, vol.~64, no.~10, pp. 2499--2512, 2016.

\bibitem{1519678}
X.~Zhang, A.~F. Molisch, and S.-Y. Kung, ``Variable-phase-shift-based
  {RF}-baseband code sign for {MIMO} antenna selection,'' \emph{IEEE Trans.
  Signal Process.}, vol.~53, no.~11, pp. 4091--4103, Nov. 2005.

\bibitem{Telatar}
E.~Telatar, ``Capacity of multi-antenna gaussian channels,'' \emph{Eur.\ Trans.
  Telecommun.}, vol.~10, no.~6, pp. 585--595, 1999.

\bibitem{4599181}
A.~Wiesel, Y.~C. Eldar, and S.~Shamai, ``Zero-forcing precoding and generalized
  inverses,'' \emph{IEEE Trans. Signal Process.}, vol.~56, no.~9, pp.
  4409--4418, Sept. 2008.

\bibitem{Luo_TSP}
Q.~{Shi}, M.~{Razaviyayn}, Z.~{Luo}, and C.~{He}, ``An iteratively weighted
  {MMSE} approach to distributed sum-utility maximization for a {MIMO}
  interfering broadcast channel,'' \emph{IEEE Trans. Signal Process.}, vol.~59,
  no.~9, pp. 4331--4340, 2011.

\bibitem{7160780}
A.~Alkhateeb, G.~Leus, and R.~W. Heath, ``Limited feedback hybrid precoding for
  multi-user millimeter wave systems,'' \emph{IEEE Trans. Wireless Commun.},
  vol.~14, no.~11, pp. 6481--6494, Nov. 2015.

\bibitem{7448873}
A.~Alkhateeb and R.~W. Heath, ``Frequency-selective hybrid precoding for
  limited feedback millimeter wave systems,'' \emph{IEEE Trans. Commun.},
  vol.~64, no.~5, pp. 1801--1818, May 2016.

\bibitem{6717211}
O.~E. Ayach, S.~Rajagopal, S.~Abu-Surra, Z.~Pi, and R.~W. Heath, ``Spatially
  sparse precoding in millimeter wave {MIMO} systems,'' \emph{IEEE Trans.
  Wireless Commun.}, vol.~13, no.~3, pp. 1499--1513, Mar. 2014.

\bibitem{Foad_JSTSP_2016}
F.~Sohrabi and W.~Yu, ``Hybrid digital and analog beamforming design for
  large-scale antenna arrays,'' \emph{IEEE J. Sel. Topics Signal Process.},
  vol.~10, no.~3, pp. 501--513, Apr. 2016.

\bibitem{7913599}
{F. Sohrabi} and {W. Yu}, ``Hybrid analog and digital beamforming for mm{W}ave
  {OFDM} large-scale antenna arrays,'' \emph{IEEE J. Sel. Areas Commun.},
  vol.~35, no.~7, pp. 1432--1443, July 2017.

\bibitem{DBLP:journals/corr/abs-1712-03485}
S.~S. {Ioushua} and Y.~C. {Eldar}, ``A family of hybrid analog-digital
  beamforming methods for massive {MIMO} systems,'' \emph{IEEE Trans. Signal
  Proces.}, vol.~67, no.~12, pp. 3243--3257, June 2019.

\bibitem{7400949}
R.~W. Heath, N.~Gonz\'{a}lez-Prelcic, S.~Rangan, W.~Roh, and A.~M. Sayeed, ``An
  overview of signal processing techniques for millimeter wave {MIMO}
  systems,'' \emph{IEEE J. Sel. Topics Signal Process.}, vol.~10, no.~3, pp.
  436--453, Apr. 2016.

\bibitem{6777295}
J.~Choi, D.~J. Love, and P.~Bidigare, ``Downlink training techniques for {FDD}
  massive {MIMO} systems: Open-loop and closed-loop training with memory,''
  \emph{IEEE J. Sel. Topics Signal Process.}, vol.~8, no.~5, pp. 802--814, Oct.
  2014.

\bibitem{1597555}
M.~Biguesh and A.~B. Gershman, ``Training-based {MIMO} channel estimation: A
  study of estimator tradeoffs and optimal training signals,'' \emph{IEEE
  Trans. Signal Process.}, vol.~54, no.~3, pp. 884--893, Mar. 2006.

\bibitem{6940305}
C.~K. Wen, S.~Jin, K.~K. Wong, J.~C. Chen, and P.~Ting, ``Channel estimation
  for massive {MIMO} using {G}aussian-mixture {B}ayesian learning,'' \emph{IEEE
  Trans. Wireless Commun.}, vol.~14, no.~3, pp. 1356--1368, Mar. 2015.

\bibitem{Javad_Tcom_2019}
J.~{Mirzaei}, R.~S. {Adve}, and S.~{Shahbazpanahi}, ``Semi-blind time-domain
  channel estimation for frequency-selective multiuser massive {MIMO}
  systems,'' \emph{IEEE Trans. Commun.}, vol.~67, no.~2, pp. 1045--1058, Feb.
  2019.

\bibitem{Javad_globcom_2018}
J.~{Mirzaei}, R.~{Adve}, and S.~{ShahbazPanahi}, ``Semi-blind channel
  estimation for frequency-selective massive {MIMO} systems,'' in \emph{Proc.
  IEEE Global Commun. Conf. (GLOBECOM), Abu Dhabi, UAE,}, Dec. 2018, pp. 1--6.

\bibitem{6489376}
M.~L. Malloy and R.~D. Nowak, ``Near-optimal adaptive compressed sensing,'' in
  \emph{Proc. IEEE Asilomar Conf. Signal, Syst., Compute., Pacific Grove, CA,
  US,}, Nov. 2012, pp. 1935--1939.

\bibitem{AlkhateebHBest}
A.~Alkhateeb, O.~E. Ayach, G.~Leus, and R.~W. Heath, ``Channel estimation and
  hybrid precoding for millimeter wave cellular systems,'' \emph{IEEE J. Sel.
  Topics Signal Process.}, vol.~8, no.~5, pp. 831--846, Oct. 2014.

\bibitem{7370753}
R.~{Méndez-Rial}, C.~{Rusu}, N.~{González-Prelcic}, A.~{Alkhateeb}, and R.~W.
  {Heath}, ``Hybrid {MIMO} architectures for millimeter wave communications:
  Phase shifters or switches?'' \emph{IEEE Access}, vol.~4, pp. 247--267, Jan.
  2016.

\bibitem{DBLP:journals/corr/AlkhateebLH15}
A.~{Alkhateeb}, G.~{Leus}, and R.~W. {Heath}, ``Compressed sensing based
  multi-user millimeter wave systems: How many measurements are needed?'' in
  \emph{Proc. IEEE Int. Conf. Acoust., Speech, Signal Process. (ICASSP),
  Brisbane, QLD, Australia}, Apr. 2015, pp. 2909--2913.

\bibitem{8093607}
A.~Liao, Z.~Gao, Y.~Wu, H.~Wang, and M.~S. Alouini, ``2{D} unitary {ESPRIT}
  based super-resolution channel estimation for millimeter-wave massive {MIMO}
  with hybrid precoding,'' \emph{IEEE Access}, vol.~5, pp. 24\,747--24\,757,
  Nov. 2017.

\bibitem{Foad_TCOM2019}
F.~{Bellili}, F.~{Sohrabi}, and W.~{Yu}, ``Generalized approximate message
  passing for massive {MIMO} mm{W}ave channel estimation with {L}aplacian
  prior,'' \emph{IEEE Trans.\ Commun.}, vol.~67, no.~5, pp. 3205--3219, May
  2019.

\bibitem{6148295}
M.~A. Iwen and A.~H. Tewfik, ``Adaptive strategies for target detection and
  localization in noisy environments,'' \emph{IEEE Trans. Signal Process.},
  vol.~60, no.~5, pp. 2344--2353, May 2012.

\bibitem{7306533}
Z.~Gao, L.~Dai, D.~Mi, Z.~Wang, M.~A. Imran, and M.~Z. Shakir, ``Mm{W}ave
  massive-{MIMO}-based wireless backhaul for the 5{G} ultra-dense network,''
  \emph{IEEE Wireless Commun.}, vol.~22, no.~5, pp. 13--21, Oct. 2015.

\bibitem{6181796}
D.~Ramasamy, S.~Venkateswaran, and U.~Madhow, ``Compressive adaptation of large
  steerable arrays,'' in \emph{Inf. Theory and App. Workshop, San Diego, CA,
  USA}, Feb. 2012, pp. 234--239.

\bibitem{Evans_ICASSP18}
J.~{Sung}, J.~{Choi}, and B.~L. {Evans}, ``Narrowband channel estimation for
  hybrid beamforming millimeter wave communication systems with one-bit
  quantization,'' in \emph{Proc. IEEE Int. Conf. Acoust., Speech, Signal
  Process. (ICASSP), Calgary, AB.}, Apr. 2018, pp. 3914--3918.

\bibitem{7938435}
Z.~Guo, X.~Wang, and W.~Heng, ``Millimeter-wave channel estimation based on
  2-{D} beamspace {MUSIC} method,'' \emph{IEEE Trans. Wireless Commun.},
  vol.~16, no.~8, pp. 5384--5394, Aug. 2017.

\bibitem{7955996}
S.~{Buzzi} and C.~{D'Andrea}, ``Subspace tracking algorithms for millimeter
  wave {MIMO} channel estimation with hybrid beamforming,'' in \emph{WSA 2017;
  21th Int. ITG Workshop on Smart Antennas, Berlin, Germany,}, 2017, pp. 1--6.

\bibitem{DBLP:journals/corr/BuzziD17a}
\BIBentryALTinterwordspacing
S.~Buzzi and C.~D'Andrea, ``Multiuser millimeter wave {{MIMO}} channel
  estimation with hybrid beamforming,'' Apr. 2017. [Online]. Available:
  \url{http://arxiv.org/abs/1704.06764}
\BIBentrySTDinterwordspacing

\bibitem{8227727}
S.~S. Ioushua and Y.~C. Eldar, ``Pilot contamination mitigation with reduced
  {RF} chains,'' in \emph{IEEE Workshop Signal Process. Adv. Wireless Commun.
  (SPAWC), Sapporo, Japan}, July 2017, pp. 1--5.

\bibitem{Eldar_9026804}
{S. S. Ioushua} and {Y. C. Eldar}, ``Pilot sequence design for mitigating pilot
  contamination with reduced {RF} chains,'' \emph{IEEE Trans. Commun.},
  vol.~68, no.~6, pp. 3536--3549, June 2020.

\bibitem{1261339}
J.~H. Kotecha and A.~M. Sayeed, ``Transmit signal design for optimal estimation
  of correlated {MIMO} channels,'' \emph{IEEE Trans. Signal Process.}, vol.~52,
  no.~2, pp. 546--557, Feb. 2004.

\bibitem{7397861}
X.~{Yu}, J.~{Shen}, J.~{Zhang}, and K.~B. {Letaief}, ``Alternating minimization
  algorithms for hybrid precoding in millimeter wave {MIMO} systems,''
  \emph{IEEE J. Sel. Topics Signal Process.}, vol.~10, no.~3, pp. 485--500,
  Apr. 2016.

\bibitem{1033689}
S.~{Zhou} and G.~B. {Giannakis}, ``Optimal transmitter eigen-beamforming and
  space-time block coding based on channel mean feedback,'' \emph{IEEE Trans.
  Signal Process.}, vol.~50, no.~10, pp. 2599--2613, Oct. 2002.

\bibitem{4133041}
Y.~{Liu}, T.~F. {Wong}, and W.~W. {Hager}, ``Training signal design for
  estimation of correlated {MIMO} channels with colored interference,''
  \emph{IEEE Trans. Signal Process.}, vol.~55, no.~4, pp. 1486--1497, Apr.
  2007.

\bibitem{4608751}
F.~{Gao}, T.~{Cui}, and A.~{Nallanathan}, ``Optimal training design for channel
  estimation in decode-and-forward relay networks with individual and total
  power constraints,'' \emph{IEEE Trans. Signal Process.}, vol.~56, no.~12, pp.
  5937--5949, Dec. 2008.

\bibitem{5340650}
E.~{Bjornson} and B.~{Ottersten}, ``A framework for training-based estimation
  in arbitrarily correlated {R}ician {MIMO} channels with {R}ician
  disturbance,'' \emph{IEEE Trans. Signal Process.}, vol.~58, no.~3, pp.
  1807--1820, Mar. 2010.

\bibitem{4803746}
M.~{Biguesh}, S.~{Gazor}, and M.~H. {Shariat}, ``Optimal training sequence for
  {MIMO} wireless systems in colored environments,'' \emph{IEEE Trans. Signal
  Process.}, vol.~57, no.~8, pp. 3144--3153, Aug. 2009.

\bibitem{6484165}
A.~Aubry, I.~Esnaola, A.~M. Tulino, and S.~Venkatesan, ``Achievable rate region
  for {Gaussian} {MIMO} {MAC} with partial {CSI},'' \emph{IEEE Trans. Info
  Theory}, vol.~59, no.~7, pp. 4139--4170, 2013.

\bibitem{985982}
C.-N. Chuah, D.~N.~C. Tse, J.~M. Kahn, and R.~A. Valenzuela, ``Capacity scaling
  in {MIMO} wireless systems under correlated fading,'' \emph{IEEE Trans. Inf.
  Theory}, vol.~48, no.~3, pp. 637--650, Mar. 2002.

\bibitem{887129}
K.~I. Pedersen, J.~B. Andersen, J.~P. Kermoal, and P.~Mogensen, ``A stochastic
  multiple-input-multiple-output radio channel model for evaluation of
  space-time coding algorithms,'' in \emph{IEEE Veh. Tech. Conf. (VTC), Boston,
  MA, USA.}, vol.~2, Aug. 2000, pp. 893--897.

\bibitem{892194}
D.~Chizhik, F.~Rashid-Farrokhi, J.~Ling, and A.~Lozano, ``Effect of antenna
  separation on the capacity of {BLAST} in correlated channels,'' \emph{IEEE
  Commun. Lett.}, vol.~4, no.~11, pp. 337--339, Nov. 2000.

\bibitem{7727938}
J.~Yuan, M.~Matthaiou, S.~Jin, and F.~Gao, ``Tightness of {J}ensen's bounds and
  applications to {MIMO} communications,'' \emph{IEEE Trans. Commun.}, vol.~65,
  no.~2, pp. 579--593, Feb. 2017.

\bibitem{965098}
K.~Yu, M.~Bengtsson, B.~Ottersten, D.~McNamara, P.~Karlsson, and M.~Beach,
  ``Second order statistics of {NLOS} indoor {MIMO} channels based on 5.2 {GH}z
  measurements,'' in \emph{Proc. IEEE Global Commun. Conf. (GLOBECOM), San
  Antonio, TX, USA.}, vol.~1, Aug. 2001, pp. 156--160.

\bibitem{1018011}
J.~R. Fonollosa, R.~Gaspa, X.~Mestre, A.~Pages, M.~Heikkila, J.~P. Kermoal,
  L.~Schumacher, A.~Pollard, and J.~Ylitalo, ``The {IST} metra project,''
  \emph{IEEE Commun. Mag.}, vol.~40, no.~7, pp. 78--86, July 2002.

\bibitem{5464944}
A.~J. {Tenenbaum} and R.~S. {Adve}, ``Energy optimization across training and
  data for multiuser minimum sum-{MSE} linear precoding,'' in \emph{IEEE Annual
  Conf. Inf. Sciences Sys., Princeton, NJ, USA}, Mar. 2010, pp. 1--6.

\bibitem{Ortega_book}
J.~M. Ortega and W.~C. Rheinboldt, \emph{Iterative Solution of Nonlinear
  Equations in Several Variables}, Classics in Applied Mathematics, SIAM, 2000.

\bibitem{1494825}
{Hua Zhang}, {Ye Li}, A.~{Reid}, and J.~{Terry}, ``Channel estimation for
  {MIMO} {OFDM} in correlated fading channels,'' in \emph{Int. Conf. Commun.},
  vol.~4, Aug. 2005, pp. 2626--2630 Vol. 4.

\bibitem{1021913}
J.~P. {Kermoal}, L.~{Schumacher}, K.~I. {Pedersen}, P.~E. {Mogensen}, and
  F.~{Frederiksen}, ``A stochastic {MIMO} radio channel model with experimental
  validation,'' \emph{IEEE J. Sel. Areas Commun.}, vol.~20, no.~6, pp.
  1211--1226, 2002.

\bibitem{656151}
R.~B. {Ertel}, P.~{Cardieri}, K.~W. {Sowerby}, T.~S. {Rappaport}, and J.~H.
  {Reed}, ``Overview of spatial channel models for antenna array communication
  systems,'' \emph{IEEE Personal Communications}, vol.~5, no.~1, pp. 10--22,
  1998.

\bibitem{Gazor_space}
S.~Gazor and H.~Rad, ``Space-time-frequency characterization of {MIMO} wireless
  channels,'' \emph{IEEE Trans. Wireless Commun.}, vol.~5, pp. 2369--2375, 09
  2006.

\bibitem{937059}
S.~A. {Jafar}, {Sriram Vishwanath}, and A.~{Goldsmith}, ``Channel capacity and
  beamforming for multiple transmit and receive antennas with covariance
  feedback,'' in \emph{IEEE Int. Conf. Commun.}, June 2001.

\bibitem{5158128}
C.~{Zhan}, K.~{Jheng}, Y.~{Chen}, T.~{Jheng}, and A.~{Wu},
  ``High-convergence-speed low-computation-complexity {SVD} algorithm for
  {MIMO-OFDM} systems,'' in \emph{Int. Symp. VLSI Design, Automation and Test},
  Jul. 2009, pp. 195--198.

\bibitem{7345023}
P.~{Tsai} and C.~{Liu}, ``Reduced-complexity {SVD} with adjustable accuracy for
  precoding in large-scale {MIMO} systems,'' in \emph{IEEE Workshop Signal
  Process. Systems (SiPS)}, Oct. 2015.

\bibitem{6415388}
J.~{Hoydis}, S.~{ten Brink}, and M.~{Debbah}, ``Massive {MIMO} in the {UL/DL}
  of cellular networks: How many antennas do we need?'' \emph{IEEE J. Sel.
  Areas Commun.}, vol.~31, no.~2, pp. 160--171, 2013.

\bibitem{Rappaport_book1}
J.~Liberti and T.~Rappaport, \emph{Smart Antennas for Wireless Communications:
  {IS-95} and Third Generation {CDMA} Applications.}\hskip 1em plus 0.5em minus
  0.4em\relax Upper Saddle River, NJ, USA: Prentice-Hall, Inc., 1997.

\bibitem{Janaswamy_book}
R.~Janaswamy, \emph{Radiowave Propagation and Smart Antennas for Wireless
  Communications}.\hskip 1em plus 0.5em minus 0.4em\relax Kluwer Academic
  Publishers, 2000.

\bibitem{Ravi_course}
\BIBentryALTinterwordspacing
R.~Adve. Receive diversity. [Online]. Available:
  \url{comm.utoronto.ca/~rsadve/Notes/DiversityReceive.pdf}
\BIBentrySTDinterwordspacing

\bibitem{4277071}
A.~{Forenza}, D.~J. {Love}, and R.~W. {Heath}, ``Simplified spatial correlation
  models for clustered {MIMO} channels with different array configurations,''
  \emph{IEEE Trans. Veh. Tech.}, vol.~56, no.~4, pp. 1924--1934, 2007.

\bibitem{Rappaport_book}
T.~Rappaport, \emph{{W}ireless communications: {P}rinciples and practice},
  2nd~ed.\hskip 1em plus 0.5em minus 0.4em\relax Prentice Hall, 2002.

\end{thebibliography}

\begin{IEEEbiography}[{\includegraphics[width=1in,height=1.25in,clip, keepaspectratio]{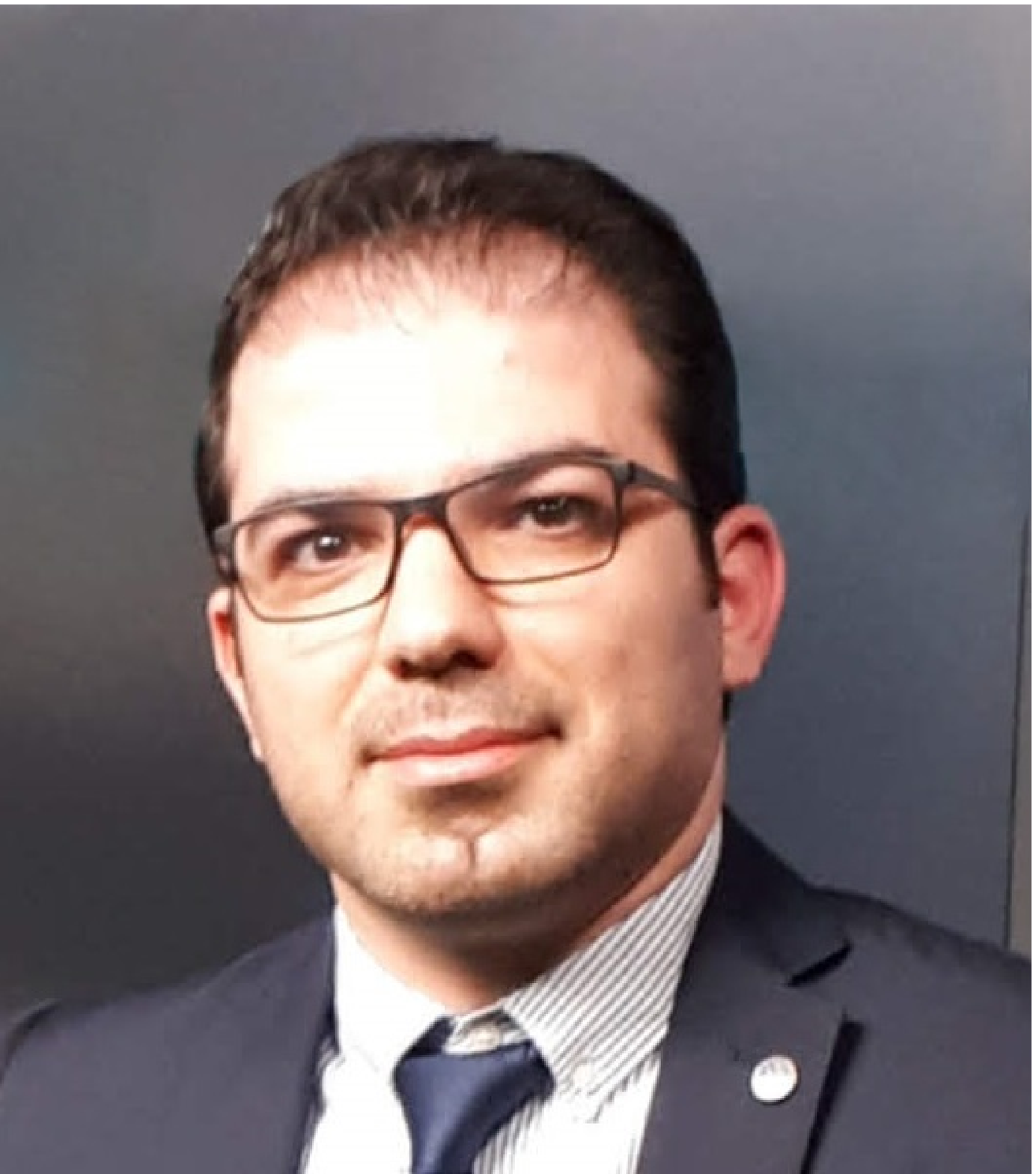}}]{Javad Mirzaei}
(S'18) was born in Tehran, Iran. He received the B.Sc. degree from Iran University of Science and Technology (IUST), Tehran, Iran, in 2010, the M.Sc. degree from Ontario Tech University, ON, Canada, in 2013, and the PhD degree from the University of Toronto, ON, Canada, in 2021, all in electrical and computer engineering. He is now a Post-Doctoral Fellow with the University of Toronto. From Feb. 2014 to Sept. 2015, he was with the Cable Shoppe Inc., Toronto, ON, Canada, where he was working in the area of broadband communication systems and IP-based 4G systems. His main research interests include wireless communication, signal processing, and machine learning.
\end{IEEEbiography}

\begin{IEEEbiography}[{\includegraphics[width=1in,height=1.25in,clip,keepaspectratio]{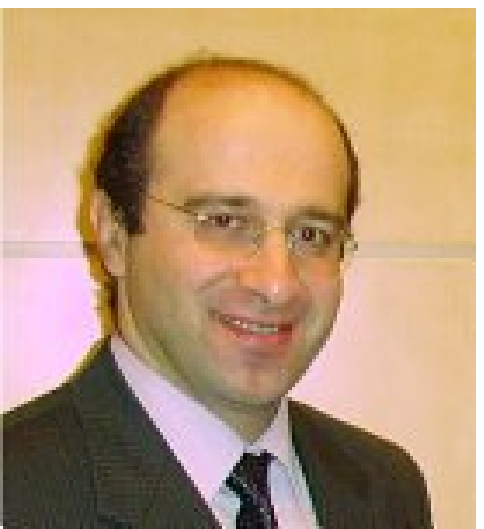}}]
{Shahram Shahbazpanahi} (M'02, SM'10) was born in Sanandaj, Kurdistan, Iran. He received the B.Sc., M.Sc., and Ph.D. degrees in electrical engineering from Sharif University of Technology, Tehran, Iran, in 1992, 1994, and 2001, respectively. From September 1994 to September 1996, he was an instructor   with the Department of Electrical Engineering, Razi University, Kermanshah, Iran. From July 2001 to March 2003, he was a Postdoctoral Fellow with the Department of Electrical and Computer Engineering, McMaster University, Hamilton, ON, Canada. From April 2003 to September 2004, he was a Visiting Researcher with the Department of Communication Systems, University of Duisburg-Essen, Duisburg, Germany. From September 2004 to April 2005, he was a Lecturer and Adjunct Professor with the Department of Electrical and Computer Engineering, McMaster University. In July 2005, he joined the Faculty of Engineering and Applied Science, University of Ontario Institute of Technology, Oshawa, ON, Canada, where he currently holds a  Professor position. His research interests include statistical and array signal processing; space-time adaptive processing; detection and estimation; multi-antenna, multi-user, and cooperative communications; spread spectrum techniques; DSP programming; and hardware/real-time software design for telecommunication systems. Dr. Shahbazpanahi has served as an Associate Editor for the IEEE TRANSACTIONS ON SIGNAL PROCESSING and the IEEE SIGNAL PROCESSING LETTERS. He has also served as a Senior Area Editor for the IEEE SIGNAL PROCESSING LETTERS. He was an elected  member of the Sensor Array and Multichannel (SAM) Technical Committee of the IEEE Signal Processing Society. He has received several awards, including the Early Researcher Award from Ontario's Ministry of Research and Innovation, the NSERC Discovery Grant (three awards), the Research Excellence Award from the Faculty of Engineering and Applied Science, the University of Ontario Institute of Technology, and the Research Excellence Award, Early Stage, from the University of Ontario Institute of Technology.
\end{IEEEbiography}

\begin{IEEEbiography}[{\includegraphics[width=1in,height=1.25in,clip,keepaspectratio]{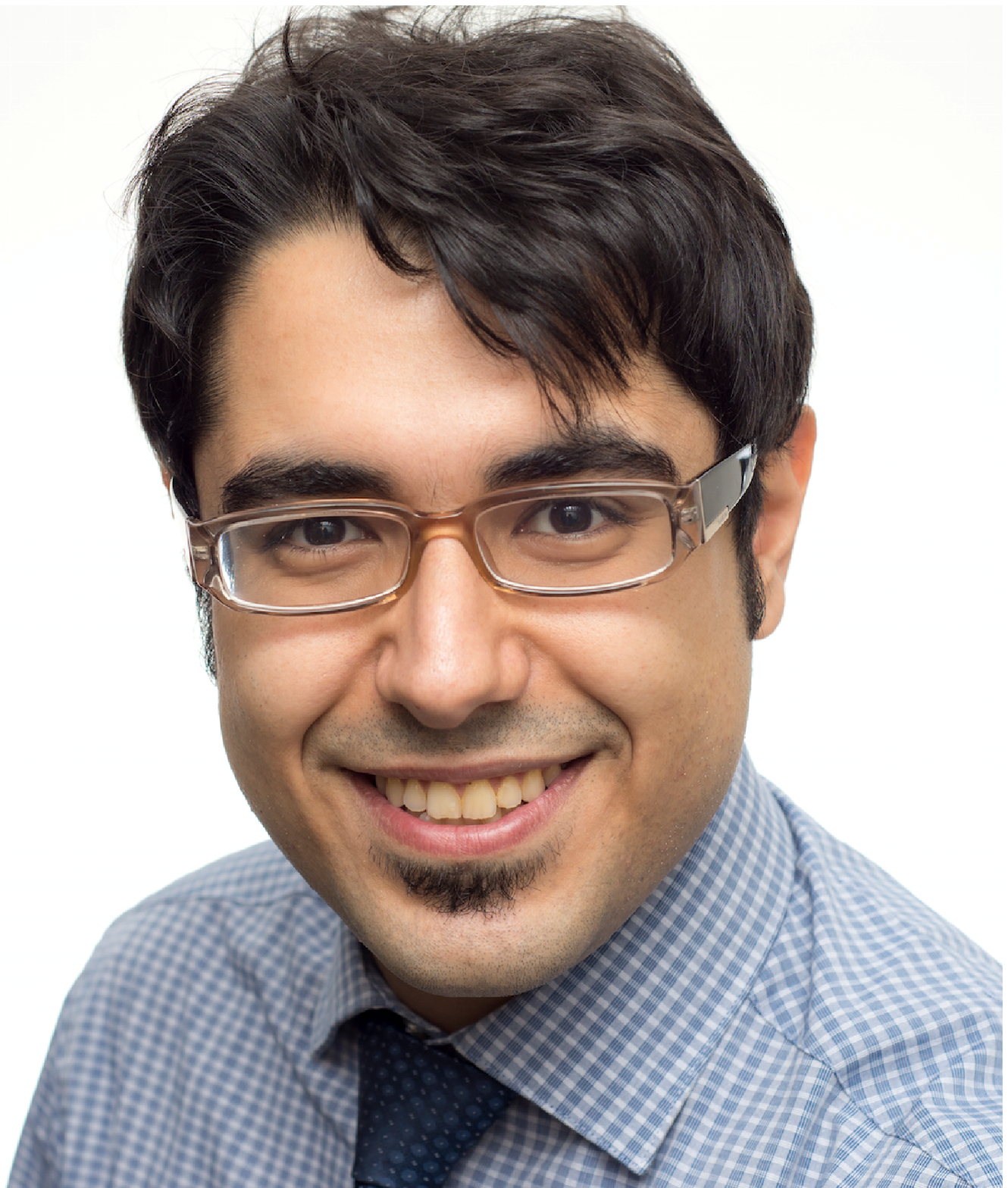}}]{Foad Sohrabi}
(S'13) received the B.A.Sc. degree from the University of Tehran, Tehran, Iran, in 2011, the M.A.Sc. degree from McMaster University, Hamilton, ON, Canada, in 2013, and the Ph.D. degree from the University of Toronto, Toronto, ON, Canada, in 2018, all in electrical and computer engineering. Since 2018, he has been a Post-Doctoral Fellow with the University of Toronto. In 2015, he was a Research Intern with Bell Labs, Alcatel-Lucent, Stuttgart, Germany. His research interests include MIMO communications, optimization theory, wireless communications, signal processing, and machine learning. He was a recipient of the IEEE Signal Processing Society Best Paper Award in 2017.
\end{IEEEbiography}

\begin{IEEEbiography}[{\includegraphics[width=1in,height=1.25in,clip,keepaspectratio]{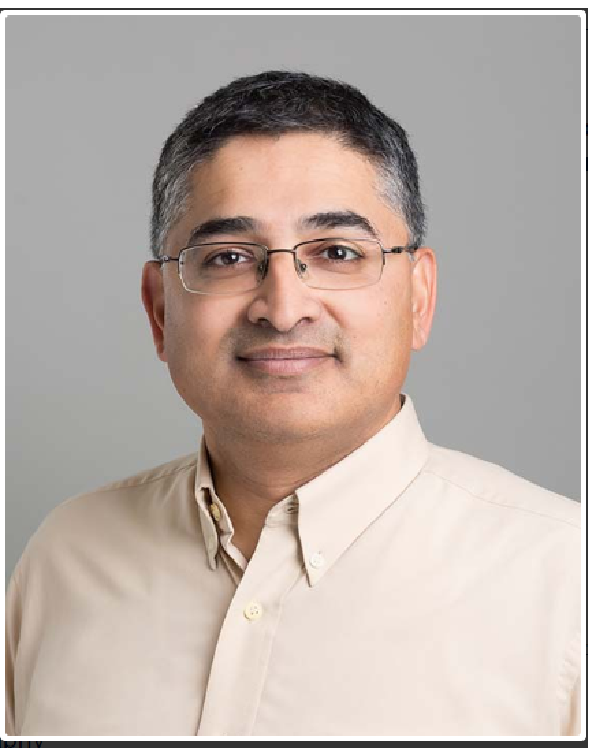}}]{Raviraj S. Adve} (S'88, M '97, SM’06, F’17) was born in Bombay, India. He received his B. Tech. in Electrical Engineering from IIT, Bombay, in 1990 and his Ph.D. from Syracuse University in 1996, His thesis received the Syracuse University Outstanding Dissertation Award. Between 1997 and August 2000, he worked for Research Associates for Defense Conversion Inc. on contract with the Air Force Research Laboratory at Rome, NY. He joined the faculty at the University of Toronto in August 2000 where he is currently a Professor. Dr. Adve’s research interests include analysis and design techniques for cooperative and heterogeneous networks, energy harvesting networks and in signal processing techniques for radar and sonar systems. He received the 2009 Fred Nathanson Young Radar Engineer of the Year award. Dr. Adve is a Fellow of the IEEE
\end{IEEEbiography}

\end{document}